\documentclass[12pt,3p,times]{elsarticle}
\usepackage{geometry}
\usepackage{amsmath}
\usepackage{amsthm}
\usepackage{amssymb}
\usepackage{amsfonts}
\usepackage{bbold}

\usepackage{bm}
\usepackage{multirow}
\usepackage{mathrsfs}
\usepackage{fixltx2e}
\usepackage[lofdepth,lotdepth]{subfig}
\usepackage{graphicx}
\usepackage{float}
\usepackage{comment}
\usepackage{dutchcal}
\usepackage[colorlinks=true,linkcolor=blue,citecolor=blue,pdfborder={0 0 0}]{hyperref}
\usepackage{verbatim}
\usepackage{xpatch}
\xpatchbibmacro{cite}{\printnames{labelname}}%
{\ifciteseen{\printnames{labelname}}{\printnames[][-\value{listtotal}]{labelname}}}
{}{}
\xpatchbibmacro{textcite}{\printnames{labelname}}%
{\ifciteseen{\printnames{labelname}}{\printnames[][-\value{listtotal}]{labelname}}}
{}{}

\numberwithin{equation}{section}

\makeatletter
\newenvironment{sub}[1]{%
  \def\subtheoremcounter{#1}%
  \refstepcounter{#1}%
  \protected@edef\theparentnumber{\csname the#1\endcsname}%
  \setcounter{parentnumber}{\value{#1}}%
  \setcounter{#1}{0}%
  \expandafter\def\csname the#1\endcsname{\theparentnumber.\Alph{#1}}%
  \ignorespaces
}{%
  \setcounter{\subtheoremcounter}{\value{parentnumber}}%
  \ignorespacesafterend
}
\makeatother
\newcounter{parentnumber}

\theoremstyle{plain}
\newtheorem{Assump}{Assumption}
\newtheorem{Theo}{Theorem}[section]
\newtheorem{Prop}{Proposition}[section]
\newtheorem{Lemma}{Lemma}

\newtheorem{Algorithm}{Algorithm}[section]

\newtheorem*{Hypo}{Hypotheses}

\theoremstyle{remark}
\newtheorem{Remark}{Remark}[section]

\newcommand{\change}[1]{{\color{black}{{#1}}}}

\newcommand{\ep}{\epsilon}
\newcommand{\bc}{\bm{c}}

\newcommand{\bB}{\bm{B}}
\newcommand{\cB}{\mathcal{B}}

\newcommand{\bH}{\bm{H}}
\newcommand{\tr}{\tilde{r}}
\newcommand{\bU}{\bm{U}}
\newcommand{\bV}{\bm{V}}
\newcommand{\bW}{\bm{W}}
\newcommand{\bX}{\bm{X}}
\newcommand{\bY}{\bm{Y}}

\newcommand{\bep}{\bm{\epsilon}}
\newcommand{\up}{\upsilon}
\newcommand{\bUp}{\bm{\Upsilon}}
\newcommand{\bxi}{\bm{\xi}}

\newcommand{\bPsi}{\bm{\Psi}}
\newcommand{\bGamma}{\bm{\Gamma}}
\newcommand{\bTheta}{\bm{\Theta}}
\newcommand{\bLambda}{\bm{\Lambda}}
\newcommand{\bOmega}{\bm{\Omega}}
\newcommand{\bzeta}{\bm{\zeta}}
\newcommand{\tgamma}{\tilde{\gamma}}
\newcommand{\tbGamma}{\tilde{\bm{\Gamma}}}
\newcommand{\tV}{\tilde{V}}
\newcommand{\tzeta}{\tilde{\zeta}}
\newcommand{\ttheta}{\tilde{\theta}}

\newcommand{\hphi}{\hat{\phi}}

\newcommand{\thphi}{\tilde{\hat{\phi}}}

\newcommand{\tpsi}{\tilde{\psi}}
\newcommand{\tbPsi}{\tilde{\bm{\Psi}}}

\newcommand{\op}{o_{p}(1)}

\graphicspath{ {Images/} }
\allowdisplaybreaks
\begin{document}
\begin{frontmatter}
\title{Bootstrap seasonal unit root test under periodic variation}
\author[1]{Nan Zou}
\ead{nan.zou@utoronto.ca}
\author[2]{Dimitris N. Politis}
\ead{dpolitis@ucsd.edu}
\address[1]{Department of Statistical Sciences, University of Toronto, Toronto, ON M5S 3G3.} 
\address[2]{Department of Mathematics, University of California-San Diego, La Jolla, CA 92093.} 
\begin{abstract}
Both seasonal unit roots and periodic variation can be prevalent in seasonal data. When testing seasonal unit roots under periodic variation, the validity of the existing methods, such as the HEGY test, remains unknown. This paper analyzes the behavior of the augmented HEGY test and the unaugmented HEGY test under periodic variation. It turns out that the asymptotic null distributions of the HEGY statistics testing the single roots at $1$ or $-1$ when there is periodic variation are identical to the asymptotic null distributions when there is no periodic variation. On the other hand, the asymptotic null distributions of the statistics testing any coexistence of roots at $1$, $-1$, $i$, or $-i$ when there is periodic variation are non-standard and are different from the asymptotic null distributions when there is no periodic variation. Therefore, when periodic variation exists, HEGY tests are not directly applicable to the joint tests for any concurrence of seasonal unit roots. As a remedy, \change{bootstrap is proposed}; in particular, the augmented HEGY test with seasonal independent and identically distributed (iid) bootstrap and the unaugmented HEGY test with seasonal block bootstrap \change{are implemented}. The consistency of these bootstrap procedures is established. The finite-sample behavior of these bootstrap tests is illustrated via simulation and prevails over their competitors'. Finally, these bootstrap tests \change{are applied} to detect the seasonal unit roots in various economic time series.  
\end{abstract}
\end{frontmatter}

\noindent

\section{Introduction}
As deterministic trend and unit root exist in time series, deterministic seasonality and seasonal unit root occur in seasonal time series. Intuitively, seasonal unit root process is a process with non-stationary stochastic seasonality. The hypothesis testing for seasonal unit root dates back to \cite{hasza1982testing,dickey1984testing}. The most widely-used test may be the HEGY test proposed by \cite{hylleberg1990seasonal}. Recent advances in this vein include \cite{del2015augmented,del2016performance,del2018semi,cavaliere2017wild}.

In addition to non-stationary stochastic seasonality, the generating processes of seasonal time series may consist of periodically varying coefficients, for example, the process may be an AutoRegressive (AR) process with periodically varying AR parameters. Examples of such periodically varying time series include the consumption series in \cite{osborn1988seasonality}, the air pollutant series in \cite{mccollister1975linear}, and the river flows in \cite{hipel1994time}. Theoretical research on periodically varying processes includes, among others, \cite{gladyshev1963periodically,tiao1980hidden}. For more information on periodically varying time series, see \cite{ghysels2001econometric,franses2004periodic,gardner2006cyclostationarity}. 

Indeed, seasonal unit roots and periodic variation sometimes coexist in seasonal data. For example, the seasonal consumption in UK has been found periodically varying and seasonally integrated by \cite{Osborn1989} and by \cite{hylleberg1990seasonal}, respectively. As a result, it is important to design seasonal unit root tests that allow for periodic variation. In particular, consider \change{quarterly} data $\{Y_{4t+s}: t=1,\dots,T$, $s=-3,\dots,0\}$ generated by 
\begin{equation}
\alpha_{s}(L)Y_{4t+s}=V_{4t+s}, \label{PAR1}
\end{equation}
where $\alpha_{s}(L)$ are seasonally varying AR filters, and $\bV_{t}=(V_{4t-3},\dots,V_{4t})'$ is a weakly stationary vector-valued process. If for all $s=-3,\dots,0$, $\alpha_{s}(L)$ have roots at $1$, $-1$, or $\pm i$, then respectively $\{Y_{4t+s}\}$ has non-stationary stochastic components with period $+\infty$, $2$, or $4$. The test for the seasonal roots at 1, $-1$, or $\pm i$ indeed precedes the removal of these non-stationary stochastic components and the inference on the detrended time series. 
To carry out this test for seasonal roots, \cite{franses1994multivariate} applies Johansen's method by \cite{johansen1988statistical}, \change{while \cite{boswijk1997} refer to the idea of likelihood ratio; however, both approaches limit scopes to finite order seasonal AR time series} and cannot directly test the existence of a certain root without first checking the number of seasonal unit roots.
As a remedy, \cite{ghysels1996periodic} design a Wald test that directly tests whether a certain root exists. \change{However, the asymptotics of \cite{ghysels1996periodic} is not totally correct according to \cite{osborn2002asymptotic}, and the simulation in \cite{ghysels1996periodic} shows the Wald test less powerful than the augmented HEGY test.}

Can we directly apply the HEGY test in the periodic setting \eqref{PAR1}? To the best of our knowledge, no literature has offered a satisfactory answer. \cite{burridge2001} analyze the behavior of the augmented HEGY test when only seasonal heteroscedasticity exists; \cite{cavaliere2017wild} take into consideration the seasonal non-stationary heteroscedasticity and the seasonal conditional heteroscedasticity but again limit their scope to heteroscedasticity; \cite{del2008} analyze the augmented HEGY test in the periodically integrated model, a model related to but different from model \eqref{PAR1}. No literature has ever touched on the behavior of the unaugmented HEGY test proposed by \cite{breitung1998}, the important semi-parametric version of the HEGY test. Since the unaugmented HEGY test does not assume the noise having an AR structure, it may suit our non-parametric model \eqref{PAR1} better.

To check the legitimacy of the HEGY test in the periodic setting \eqref{PAR1}, this paper derives the asymptotics of the unaugmented HEGY test and the augmented HEGY test. It turns out that, the asymptotic null distributions of the statistics testing the single roots at 1 or $-1$ are standard. More specifically, for each single root at 1 or $-1$, the asymptotic null distribution of the augmented HEGY statistic is identical to that of Augmented Dickey-Fuller (ADF) test by \cite{dickey1979distribution}, and the asymptotic null distribution of the unaugmented HEGY statistic is identical to that of Phillips-Perron test by \cite{phillips1988}. However, the asymptotic null distributions of the statistics testing any combination of roots at 1, $-1$, $i$, or $-i$ depend on the periodically varying coefficients, are non-standard and non-pivotal, and cannot be directly pivoted. Therefore, when periodic variation exists, the augmented and the unaugmented HEGY tests can be applied to single roots at 1 or $-1$ but cannot be straightforwardly applied to the coexistence of any roots. 

As a remedy, this paper proposes the application of bootstrap. In general, bootstrap's advantages are two fold. Firstly, bootstrap helps when the asymptotic distributions of the statistics of interest cannot be found or simulated. Secondly, even when the asymptotic distributions can be found and simulated, bootstrap method may enjoy second-order efficiency when these asymptotic distributions are pivotal. For the aforementioned problem, bootstrap serves as an appealing solution. Firstly, it is hard to estimate the periodically varying parameters in the asymptotic null distributions, and it is hard to simulate these asymptotic null distributions. Secondly, it can be conjectured that the bootstrap seasonal unit root test inherits second order efficiency from the bootstrap non-seasonal unit root test when the asymptotic distributions are pivotal; see \cite{park2003}. The methodological literature we find on bootstrapping the HEGY test only includes \cite{burridge2004bootstrapping,cavaliere2017wild}. It will be shown in Remark \ref{Re:seasonal iid bootstrap} that none of these bootstrap approaches is consistent under the general periodic setting \eqref{PAR1}. 

To cater to the general periodic setting \eqref{PAR1}, this paper designs new bootstrap tests, namely 1) the seasonal iid bootstrap augmented HEGY test, and 2) the seasonal block bootstrap unaugmented HEGY test. When calculating the test statistics, the two tests run HEGY regression using all data in order to preserve the orthogonal structure of the HEGY regression. On the other hand, when generating bootstrap replicates, both tests conduct season-by-season regressions to duplicate the periodic structure of the original data. In particular, the first test obtains residuals from season-by-season augmented HEGY regressions, and then applies the seasonal iid bootstrap to the whitened regression errors, while the second test starts with season-by-season unaugmented HEGY regressions, and then handles the correlated errors with the seasonal block bootstrap proposed by \cite{dudek2014generalized}. We establish the Functional Central Limit Theorem (FCLT) for both bootstrap tests and then demonstrates the consistency of both bootstrap procedures. 

The paper proceeds as follows. Section 2 formalizes the settings, states the assumptions, and presents the hypotheses. Section 3 gives the asymptotic null distributions of the augmented HEGY test statistics, details the algorithm of the seasonal iid bootstrap augmented HEGY test, and establishes the consistency of the bootstrap. Section 4 presents the asymptotic null distributions of the unaugmented HEGY test statistics, specifies the algorithm of the seasonal block bootstrap unaugmented HEGY test, and proves the consistency of the bootstrap. Section 5 shows that in simulation our two bootstrap tests outperform \change{their competitors, namely, the non-seasonal bootstrap augmented HEGY test by \cite{burridge2004bootstrapping}} and the Wald test by \cite{ghysels1996periodic}. Section 6 applies our two bootstrap tests to various economic time series. Appendix includes all technical proofs. 

\section{Periodically varying time series}\label{sec:settings}
Consider the real-valued quarterly data $\{Y_{4t+s}: t=1,\dots,T$, $s=-3,\dots,0\}$ generated by the seasonal model 
\begin{equation}
\alpha_{s}(L)Y_{4t+s}=V_{4t+s}, \label{PAR2}
\end{equation}
where $LY_{4t+s}=Y_{4t+s-1}$ and $\alpha_{s}(L)=1-\sum_{j=1}^{4}\alpha_{j,s}L^{j}$. 
\change{Suppose that for all $s=-3,\dots,0$, the roots of $\alpha_{s}(L)$ are on or outside the unit circle. If for all $s=-3,\dots,0$, $\alpha_{s}(L)$ has roots on the unit circle, then suppose that for $s=-3,\dots,0$, $\alpha_{s}(L)$ share the same set of roots on the unit circle}, this set of roots is a subset of $\{1,-1,\pm i\}$, and $Y_{-3}=Y_{-2}=Y_{-1}=Y_{0}=0$; otherwise, suppose that our data is a stretch of the process $\{Y_{4t+s}, t=\dots,-1,0,1,\dots$, $s=-3,\dots,0\}$. Let $V_{4t+s}$ and $\alpha_{j,s}$ be the prediction errors and coefficients of \eqref{PAR2}, respectively. More specifically, $\alpha_{j,s}$ is defined such that \change{for each $s=-3,\dots,0$,}
$$V_{4t+s}\stackrel{def}{=}\alpha_{s}(L)Y_{4t+s}$$
is orthogonal to $Y_{4t+s-1},\dots,Y_{4t+s-4}$. Let $\bep_{t}=(\ep_{4t-3},\dots,\ep_{4t})'$ and $\bB\bep_{t}=\bep_{t-1}$. Then $\bB=L^{4}$. Denote by AR$(p)$ an AR process with order $p$, by MA Moving Average, by VMA$(\infty)$ a Vector MA process with infinite moving average order, and by VARMA$(p,q)$ a Vector ARMA process with AR order $p$ and MA order $q$. Let $Re(z)$ be the real part of complex number $z$, $\lfloor x \rfloor$ be the largest integer smaller or equal to real number $x$, and $\lceil x \rceil$ be the smallest integer larger or equal to $x$.

\begin{sub}{Assump}
\begin{Assump} \label{assump 1a}
Assume $$\bV_{t}=\bTheta(\bB)\bep_{t}$$
where $\bTheta(\bB)=\sum_{i=0}^{\infty}\bTheta_{i}\bB^{i}$; the $(j,k)$ entry of $\bTheta_{i}$, denoted by $\bTheta_{i}^{(j,k)}$, satisfies $\sum_{i=1}^{\infty}i|\bTheta_{i}|^{(j,k)}<\infty$ for all $j$ and $k$; the determinant of $\bTheta(z)$ has all roots outside the unit circle; $\bTheta_{0}$ is a lower diagonal matrix whose diagonal entries equal 1; $\bep_{t}$ is a vector-valued white noise process with mean zero and covariance matrix $\bOmega$; and $\bOmega$ is diagonal.
\end{Assump}
Assumption \ref{assump 1a} assumes that $\{\bV_{t}\}$ is VMA$(\infty)$ with respect to white noise innovations. This is equivalent to the assumption that $\{\bV_{t}\}$ is a weakly stationary process with no deterministic part in the multivariate Wold decomposition. The assumptions on $\bTheta_{0}$ and the determinant of $\bTheta(z)$ ensure the causality and the invertibility of $\{\bV_{t}\}$ and the identifiability of $\bOmega$. 
\begin{Assump} \label{assump 1b}
Assume $$\bV_{t}=\bPsi(\bB)^{-1}\bLambda(\bB)\bep_{t}\equiv\bTheta(\bB)\bep_{t}$$
where $\bPsi(\bB)=\sum_{i=0}^{p}\bPsi_{i}\bB^{i}$; $\bLambda(\bB)=\sum_{i=0}^{q}\bLambda_{i}\bB^{i}$; the determinants of $\bPsi(z)$ and $\bLambda(z)$ have all roots outside the unit circle; $\bPsi_{0}$ and $\bLambda_{0}$ are lower diagonal matrices whose diagonal entries are 1; $\bep_{t}$ is a vector-valued white noise process with mean zero and covariance matrix $\bOmega$; and $\bOmega$ is diagonal. 
\end{Assump}
\end{sub}
Assumption \ref{assump 1b} restricts $\{\bV_{t}\}$ to be VARMA$(p,q)$ with respect to white noise innovation. Compared to the VMA$(\infty)$ model in Assumption \ref{assump 1a}, VARMA$(p,q)$'s main constraint is its exponentially decaying autocovariance. Again, the assumptions on $\bPsi_{0}$, $\bLambda_{0}$ and the determinant of $\bPsi(z)$ and $\bLambda(z)$ in Assumption \ref{assump 1b} ensure the causality and the invertibility of $\{\bV_{t}\}$ and the identifiablity of $\bOmega$. 

At this stage $\{\bep_{t}\}$ is only assumed to be a white noise sequence of random vectors. In fact, $\{\bep_{t}\}$ needs to be weakly dependent as well; however, $\{\bep_{t}\}$ needs not to be iid.  
\begin{sub}{Assump} 
\begin{Assump}
\label{assump 2a}
(i) $\{\bep_{t}\}$ is a fourth-order stationary, martingale difference vector-valued process. (ii) $\exists K>0$, $\forall$ $i$, $j$, $k$, and $l$, $\sum_{h=-\infty}^{\infty}|Cov(\ep_{i}\ep_{j},\ep_{k-h}\ep_{l-h})|<K$.
\end{Assump}
\begin{Assump}
\label{assump 2b}
(i) $\{\bep_{t}\}$ is a strictly stationary, strong-mixing vector-valued process with finite $4+\delta$ moment for some $\delta>0$. (ii) $\{\bep_{t}\}$'s strong mixing coefficient $a(k)$ satisfies $\sum_{k=1}^{\infty}k(a(k))^{\delta/(4+\delta)}<\infty$.
\end{Assump}
Notice that the assumption on the stationarity of the vector-valued process $\{\bep_{t}\}$ is weaker than an assumption on the stationarity of the scalar-valued process $\{\ep_{4t+s}\}$. In addition, the strong mixing condition in Assumption \ref{assump 2b} actually guarantees (ii) of Assumption \ref{assump 2a}; see Lemma \ref{boundedness}.
\end{sub}
\begin{Hypo}
We tackle the following set of null hypotheses. The alternative hypotheses are the complements of the null hypotheses. 
\begin{align*}
&H_{0}^{1}:&&\alpha_{s}(1)=0,\ &&\forall s=-3,\dots,0.\\
&H_{0}^{2}:&&\alpha_{s}(-1)=0 ,\ &&\forall s=-3,\dots,0.\\
&H_{0}^{1,2}:&&\alpha_{s}(1)=\alpha_{s}(-1)=0, \ &&\forall s=-3,\dots,0.\\
&H_{0}^{3,4}:&&\alpha_{s}(i)=\alpha_{s}(-i)=0,\ &&\forall s=-3,\dots,0.\\
&H_{0}^{1,3,4}:&&\alpha_{s}(1)=\alpha_{s}(i)=\alpha_{s}(- i)=0,\ &&\forall s=-3,\dots,0.\\
&H_{0}^{2,3,4}:&&\alpha_{s}(-1)=\alpha_{s}(i)=\alpha_{s}(-i)=0,\ &&\forall s=-3,\dots,0.\\
&H_{0}^{1,2,3,4}:&&\alpha_{s}(1)=\alpha_{s}(-1)=\alpha_{s}(i)=\alpha_{s}(-i)=0,\ &&\forall s=-3,\dots,0.
\end{align*}
\end{Hypo}

Indeed, the alternative hypotheses can be written as one-sided hypotheses. Notice that for all $s=-3,\dots,0$, $\alpha_{s}(0)=1$, $\alpha_{s}(\cdot)$ is continuous, and the roots of $\alpha_{s}(\cdot)$ are either on or outside the unit circle. By the intermediate value theorem, $\alpha_{s}(1)\neq0$ implies that $\alpha_{s}(1)>0$, $\alpha_{s}(-1)\neq0$ implies that $\alpha_{s}(-1)>0$, and $\alpha_{s}(i)\neq0$ implies that $Re(\alpha_{s}(i))>0$. 

To analyze the roots of $\alpha_{s}(L)$, \cite{hylleberg1990seasonal} propose the partial fraction decomposition 
$$
  \frac{\alpha_{s}(L)}{1-L^{4}}=
  \lambda_{0,s}+\frac{\lambda_{1,s}}{1-L}+\frac{\lambda_{2,s}}{1+L}+\frac{\lambda_{3,s}L+\lambda_{4,s}}{1+L^{2}}; 
$$
thus
\begin{equation}
\begin{aligned}
\alpha_{s}(L)&=\lambda_{0,s}(1-L^{4})\\
&+\lambda_{1,s}(1+L)(1+L^{2})+\lambda_{2,s}(1-L)(1+L^{2}) \\
&+\lambda_{3,s}(1-L)(1+L)L+\lambda_{4,s}(1-L)(1+L).
\end{aligned} 
\label{partial fraction}
\end{equation}
\\
Substituting \eqref{partial fraction} into \eqref{PAR2}, we get 
\begin{equation}
(1-L^{4})Y_{4t+s}=\sum_{j=1}^{4}\pi_{j,s}Y_{j,4t+s-1}+V_{4t+s},
\label{PHEGY}
\end{equation}
where
\begin{equation}
\begin{aligned}
Y_{1,4t+s}&=(1+L)(1+L^{2})Y_{4t+s},&
Y_{2,4t+s}&=-(1-L)(1+L^{2})Y_{4t+s},\\
Y_{3,4t+s}&=-L(1-L^{2})Y_{4t+s},&
Y_{4,4t+s}&=-(1-L^{2})Y_{4t+s},\\
\pi_{1,s}&=-\lambda_{1,s},&
\pi_{2,s}&=-\lambda_{2,s},\\ 
\pi_{3,s}&=-\lambda_{4,s},&
\pi_{4,s}&=\lambda_{3,s}\change{.}
\end{aligned}
\label{PHEGY2}
\end{equation}
By \eqref{partial fraction} and \eqref{PHEGY2}, $\pi_{j,s}$ relates to the root of $\alpha_{s}(z)$.
\begin{Prop}[\cite{hylleberg1990seasonal}]
\label{prop:HEGY}
\begin{align*}
&\alpha_{s}(1)=0 \iff  \pi_{1,s}=0, &&\alpha_{s}(1)\neq0 \iff  \pi_{1,s}<0,\\
&\alpha_{s}(-1)=0 \iff  \pi_{2,s}=0, &&\alpha_{s}(-1)\neq0 \iff  \pi_{2,s}<0,\\
&\alpha_{s}(i)=0 \iff \alpha_{s}(-i)=0 \iff  \pi_{3,s}=\pi_{4,s}=0,
&&\alpha_{s}(i)\neq 0 \iff \alpha_{s}(-i)\neq0 \iff \pi_{3,s}<0.
\end{align*}
\end{Prop}

By Proposition \ref{prop:HEGY}, the test for the null hypotheses can be carried on by checking the corresponding $\pi_{j,s}$, where $\pi_{j,s}$ can be estimated by Ordinary Least Squares (OLS) regression. To estimate $\pi_{j,s}$ by OLS, one might first attempt to implement the OLS season by season with season-specific coefficients. Unfortunately, this season-by-season regression indeed has a non-orthogonal design matrix; see \cite{ghysels2001econometric}, \change{p. 158,} and Lemma \ref{le:unaug real1}. On the other hand, since the non-periodic regressions \eqref{aug HEGY} and \eqref{unaug HEGY} preserve the orthogonality, we will instead apply the non-periodic regression equations \eqref{aug HEGY} and \eqref{unaug HEGY}. 

When we regress $\{Y_{4t+s}\}$ with non-periodic regression equations \eqref{aug HEGY} and \eqref{unaug HEGY}, the periodically varying sequence $\{V_{4t+s}\}$ is fitted in misspecified non-periodic AR models. Consider, as an example, fitting $\{V_{4t+s}\}$ in a misspecified AR(1) model $V_{\tau}=\hphi V_{\tau-1}+\hat{\zeta}_{\tau}$. Then $\hphi=\tgamma(1)/\tgamma(0)+o_{p}(1)$, where 
\begin{equation}
\tgamma(h)=\frac{1}{4}\sum_{s=-3}^{0}E[V_{4t+s}V_{4t+s-h}].
\label{eqn:tgamma}    
\end{equation}
Since $\tgamma(\cdot)$ is positive semi-definite, we can find a weakly stationary sequence $\{\tV_{\tau}\}$ with mean zero and autocovariance function $\tgamma(\cdot)$. We call $\{\tV_{\tau}\}$ a misspecified constant parameter representation of $\{V_{4t+s}\}$; see also \cite{osborn1991implications}. We will refer to $\{\tV_{\tau}\}$ in later sections. 
\section{Seasonal iid bootstrap augmented HEGY Test}
\subsection{Augmented HEGY test}
\cite{hylleberg1990seasonal} assume that the $\pi_{j,s}$ and $V_{4t+s}$ in \eqref{PHEGY} do not depend on $s$. Consequently, they propose to run the OLS regression equation 
\begin{equation}\label{aug HEGY}
(1-L^{4})Y_{\tau}=\sum_{j=1}^{4}\hat{\pi}_{j}^{A}Y_{j,\tau-1}+\sum^{k}_{i=1}\hphi_{i}(1-L^{4})Y_{\tau-i}+\hat{\zeta}_{\tau}^{A},
\end{equation}
where augmentations $(1-L^{4})Y_{\tau-i}$, $i=1,2,\dots,k$, pre-whiten the time series $(1-L^{4})Y_{\tau}$ up to an order of $k$. If $k\rightarrow \infty$ as sample size $T\rightarrow \infty$, the residual $\{\hat{\zeta}_{\tau}^{A}\}$ will be asymptotically uncorrelated.

Let $A$ stands for ``Augmented''. Let $\hat{\pi}_{j}^{A}$ be the OLS estimator in \eqref{aug HEGY}, $t_{j}^{A}$ be the t-statistics corresponding to $\hat{\pi}_{j}^{A}$, and $F_{3,4}^{A}$ be the F-statistic corresponding to $\hat{\pi}_{3}^{A}$ and $\hat{\pi}_{4}^{A}$.
Other $F$-statistics $F_{1,2}^{A}$, $F_{1,3,4}^{A}$, $F_{2,3,4}^{A}$, and $F_{1,2,3,4}^{A}$ can be defined similarly. Where there is no periodic variation, \cite{hylleberg1990seasonal} proposes to reject $H_{0}^{1}$ if $\hat{\pi}_{1}^{A}$ is too small, reject $H_{0}^{2}$ if $\hat{\pi}_{2}^{A}$ is too small, reject $H_{0}^{3,4}$ if $F_{3,4}^{A}$ is too large, and reject other composite hypotheses if their corresponding $F$-statistics are too large.

\subsection{Augmented HEGY test under model misspecification}
Now we apply the augmented HEGY test to periodically varying processes. Namely, we run regression equation \eqref{aug HEGY} with $\{Y_{4t+s}\}$ generated by \eqref{PAR2}. Our results show that when testing roots at 1 or $-1$ separately, the t-statistics $t_{1}^{A}$, $t_{2}^{A}$, and the $F$-statistics have pivotal asymptotic distributions. On the other hand, when testing joint roots at 1 and $-1$, and when testing hypotheses that involve roots at $\pm i$, the asymptotic distributions of the testing statistics are non-pivotal and cannot be easily pivoted. 
\begin{Theo}\label{aug real}
Assume that Assumption \ref{assump 1b} and one of Assumption \ref{assump 2a} or \ref{assump 2b} hold. Further, assume $T\rightarrow \infty$, $k=k_{T}\rightarrow \infty$, $k=o(T^{1/3})$, and $ck>T^{1/\alpha}$ for some $c>0$ and $\alpha>0$. Then under $H_{0}^{1,2,3,4}$, the asymptotic distributions of $t_{j}^{A}$, $j=1,2$, and $F$-statistics are given by 
\begin{align*}
t_{j}^{A}&\Rightarrow\frac{\int_{0}^{1}W_{j}(r)dW_{j}(r)}{\sqrt{\int_{0}^{1}W_{j}^{2}(r)dr}}\equiv\mathscr{\xi}_{j},\ \text{j=1,2,}\\
F_{1,2}^{A}&\Rightarrow\frac{1}{2}(\mathscr{\xi}_{1}^{2}+\mathscr{\xi}_{2}^{2}), \quad
F_{3,4}^{A}\Rightarrow\frac{1}{2}(\mathscr{\xi}_{3}^{2}+\mathscr{\xi}_{4}^{2}), \\
F_{1,3,4}^{A}&\Rightarrow\frac{1}{3}(\mathscr{\xi}_{1}^{2}+\mathscr{\xi}_{3}^{2}+\mathscr{\xi}_{4}^{2}), \quad
F_{2,3,4}^{A}\Rightarrow\frac{1}{3}(\mathscr{\xi}_{2}^{2}+\mathscr{\xi}_{3}^{2}+\mathscr{\xi}_{4}^{2}), \\
F_{1,2,3,4}^{A}&\Rightarrow\frac{1}{4}(\mathscr{\xi}_{1}^{2}+\mathscr{\xi}_{2}^{2}+\mathscr{\xi}_{3}^{2}+\mathscr{\xi}_{4}^{2}), \quad \text{with}\\
\mathscr{\xi}_{3}&=
\frac{\lambda_{3}^{2}\int_{0}^{1}W_{3}(r)dW_{3}(r)+\lambda_{4}^{2}\int_{0}^{1}W_{4}(r)dW_{4}(r)}{\sqrt{(\lambda_{3}^{2}+\lambda_{4}^{2})(\frac{1}{2}\lambda_{3}^{2}\int_{0}^{1}W_{3}^{2}(r)dr+\frac{1}{2}\lambda_{4}^{2}\int_{0}^{1}W_{4}^{2}(r)dr)}},\\
\mathscr{\xi}_{4}&=
\frac{\lambda_{3}\lambda_{4}(\int_{0}^{1}W_{3}(r)dW_{4}(r)-\int_{0}^{1}W_{4}(r)dW_{3}(r))}{\sqrt{(\lambda_{3}^{2}+\lambda_{4}^{2})(\frac{1}{2}\lambda_{3}^{2}\int_{0}^{1}W_{3}^{2}(r)dr+\frac{1}{2}\lambda_{4}^{2}\int_{0}^{1}W_{4}^{2}(r)dr)}},
\end{align*}
where $\bc_{1}=(1,1,1,1)'$, $\bc_{2}=(1,-1,1,-1)'$, $\bc_{3}=(0,-1,0,1)'$, $\bc_{4}=(-1,0,1,0)'$, $\lambda_{j}=\sqrt{\bc_{j}'\bTheta(1)\bOmega\bTheta(1)'\bc_{j}/4}$, $W_{j}=\bc_{j}'\bTheta(1)\bOmega^{1/2}\bW/\change{(}2\lambda_{j}\change{)}$, and $\bW(\cdot)$ is a four-dimensional standard Brownian motion.
\end{Theo}


\begin{Remark}
The asymptotic distributions presented in Theorem \ref{aug real} degenerate to the distributions in \cite{burridge2001properties} and \cite{del2012augmented} when $\{V_{4t+s}\}$ has neither periodic variation nor seasonal heteroscedasticity, and to the distributions in \cite{burridge2001} when $\{V_{4t+s}\}$ is a heteroscedastic, finite-order AR sequence with non-periodic AR coefficients.
\end{Remark}

\begin{Remark} 
Notice that each of $W_{j}$ is a standard Brownian motion. When $\{V_{4t+s}\}$ has no periodic variation, $W_{j}$'s are independent, so are the asymptotic distributions of $t_{1}^{A}$ and $t_{2}^{A}$. On the other hand, when $\{V_{4t+s}\}$ has periodic variation, $W_{j}$'s are in general dependent, so $t_{1}^{A}$ and $t_{2}^{A}$ are in general dependent, even asymptotically. Hence, when testing $H_{0}^{1,2}$, it is problematic to test $H_{0}^{1}$ and $H_{0}^{2}$ separately and calculate the size of the test with the independence of $t_{1}^{A}$ and $t_{2}^{A}$ in mind. Instead, the test of $H_{0}^{1,2}$ should be handled with $F_{1,2}^{A}$. 
\end{Remark}

\begin{Remark}\label{pm1}
Because of the dependence of $t_{1}^{A}$ and $t_{2}^{A}$, the asymptotic distribution of $F_{1,2}^{A}$ under periodic variation is different from its non-periodic counterpart. More generally, the asymptotic distributions of any aforementioned $F$-statistics under periodic variation is different from their non-periodic counterparts. Hence, the augmented HEGY test cannot be directly applied to test any coexistence of roots at $1$, $-1$, $i$, or $-i$ under potential periodic variation. (From another point of view, the asymptotic distribution of any $F$-statistics does not solely depend on the distribution of $\{\tV_{\tau}\}$, the misspecified constant parameter representation of $\{V_{4t+s}\}$; hence the $F$-tests are truly affected by the periodic variation.)
\end{Remark}

\begin{Remark} \label{heter}
When $\{V_{4t+s}\}$ is only seasonally heteroscedastic, as in \cite{burridge2001}, $\bTheta(1)$ does not occur in the asymptotic distributions of the $F$-statistics. On the other hand, when $\{V_{4t+s}\}$ has generic periodic variation, $\bTheta(1)$ impacts first the correlation between Brownian motions $W_{3}$ and $W_{4}$, and second the weights $\lambda_{3}$ and $\lambda_{4}$.
\end{Remark}

\begin{Remark} 
As \cite{burridge2001} point out, the dependence of the asymptotic distributions on weights $\lambda_{3}$ and $\lambda_{4}$ can be expected.
Indeed, $Y_{3,4t+s}=Y_{4,4t+s-1}$ is the partial sum of $\{-V_{4t+s-1},V_{4t+s-3},\dots\}$, while $Y_{3,4t+s+1}=Y_{4,4t+s}$ is the partial sum of $\{-V_{4t+s},V_{4t+s-2},\dots\}$. Since these two partial sums differ in their variances, both $\sum_{s,t} Y_{3,4t+s}$ and $\sum_{s,t} Y_{4,4t+s}$ involve two different weights $\lambda_{3}$ and $\lambda_{4}$.
\end{Remark}

\begin{Remark}\label{root}
Theorem \ref{aug real} presents the asymptotics when $\{Y_{4t+s}\}$ is generated under $H_{0}^{1,2,3,4}$, that is, when $\{Y_{4t+s}\}$ has all roots at 1, $-1$, and $\pm i$. When $\{Y_{4t+s}\}$ is generated under other null hypotheses in Section 2, that is, when $\{Y_{4t+s}\}$ has some but not all roots at 1, $-1$, and $\pm i$, we let $U_{\tau}=(1-L^{4})Y_{\tau}$, $\bU_{t}=(U_{4t-3},U_{4t-2},U_{4t-1},U_{4t})'$, and define $\bm{H}(z)$ such that $\bU_{t}=\bm{H}(\bB)\bep_{t}$. The asymptotic distributions under other null hypotheses has exactly the same form as those in Theorem \ref{aug real}, except that $\bTheta(1)$ is replaced by $\bm{H}(1)$. When there is no periodic variation, these asymptotic distributions degenerate to those in \cite{del2007using} \change{and \cite{smith2009regression}}.
\end{Remark}
\begin{Remark}
The preceding results give the asymptotic behaviors of the testing statistics under the null hypotheses. \change{Under the alternative hypotheses, we conjecture that the power of the augmented HEGY tests tends to one as the sample size goes to infinity. Indeed, if $\{Y_{4t+s}\}$ does not have a certain unit root at 1, $-1$, or $\pm i$, then by the asymptotic orthogonality of regression equation \eqref{aug HEGY}, we can without loss of generality assume that $\{Y_{4t+s}\}$ has none of the unit roots at 1, $-1$, or $\pm i$. If $\{Y_{4t+s}\}$ has none of the unit roots at 1, $-1$, or $\pm i$, then it has a stationary misspecified constant parameter representation. Then, by \cite{berk1974}, for $j=1,2,3$, $\hat{\pi}_{j}^{A}$ converge in probability; by Proposition \ref{prop:HEGY}, the limits of $\hat{\pi}_{j}^{A}$, $j=1,2,3,$ are negative.} See also Theorem 2.2 of \cite{paparoditis2016}.
\end{Remark}
\subsection{Seasonal iid bootstrap algorithm}
To accommodate the non-pivotal asymptotic null distributions of the augmented HEGY test statistics, we propose the application of bootstrap. Specifically, we first pre-whiten the data season by season to obtain uncorrelated noises. Although these noises are uncorrelated, they are not identically distributed due to seasonal heteroscedasticity. Hence, we second resample season by season to generate bootstrapped noise, as in \cite{burridge2004bootstrapping}. Finally, we post-color the bootstrapped noise. The detailed algorithm of this seasonal iid bootstrap augmented HEGY test is given below.
\begin{Algorithm}\label{seasonal iid bootstrap}
Step 1: calculate $t_{1}^{A}$ and $t_{2}^{A}$, the t-statistics corresponding to $\hat{\pi}_{1}^{A}$ and $\hat{\pi}_{2}^{A}$, and the $F$-statistics $F_{\cB}^{A}$, $\cB=\{1,2\},\{3,4\},\{1,3,4\},\{2,3,4\}$, $\{1,2,3,4\},$ from the augmented non-periodic HEGY test regression
$$ (1-L^{4})Y_{\tau}=\sum_{j=1}^{4}\hat{\pi}_{j}^{A}Y_{j,\tau-1}+\sum_{i=1}^{k}\hphi_{i}(1-L^{4})Y_{\tau-i}+\hat{\zeta}_{\tau}^{A}; $$
Step 2: record OLS estimators $\hat{\pi}_{j,s}^{A}$, $\hphi_{i,s}$ and residuals $\hat{\ep}_{4t+s}$ from the season-by-season regression 
$$(1-L^{4})Y_{4t+s}=\sum_{j=1}^{4}\hat{\pi}_{j,s}^{A}Y_{j,4t+s-1}+\sum_{i=1}^{k}\hphi_{i,s}(1-L^{4})Y_{4t+s-i}+\hat{\ep}_{4t+s};$$
Step 3: let $\check{\ep}_{4t+s}=\hat{\ep}_{4t+s}-\frac{1}{T}\sum_{t=\lfloor k/4 \rfloor +1}^{T}\hat{\ep}_{4t+s}$. Store demeaned residuals $\{\check{\ep}_{4t+s}\}$ of the four seasons separately, then independently draw four iid samples from each of their empirical distributions, and then combine these four samples into a vector $\{\ep_{4t+s}^{\star}\}$, with their seasonal orders preserved; \\
\\
Step 4: set all $\hat{\pi}_{j,s}^{A}$ corresponding to the null hypothesis to be zero. For example, set $\hat{\pi}_{3,s}^{A}=\hat{\pi}_{4,s}^{A}=0$ for all $s$ when testing roots at $\pm i$. Let $\{Y_{4t+s}^{\star}\}$ be generated by
$$(1-L^{4})Y^{\star}_{4t+s}=\sum_{j=1}^{4}\hat{\pi}_{j,s}^{A}Y^{\star}_{j,4t+s-1}+\sum_{i=1}^{k}\hphi_{i,s}(1-L^{4})Y^{\star}_{4t+s-i}+\ep^{\star}_{4t+s};$$
Step 5: calculate $t_{1}^{\star}$ and $t_{2}^{\star}$, the t-statistics corresponding to $\hat{\pi}^{\star}_{1}$ and $\hat{\pi}^{\star}_{2}$, and $F$-statistics $F_{\cB}^{\star}$ from the non-periodic regression
$$ (1-L^{4})Y^{\star}_{\tau}=\sum_{j=1}^{4}\hat{\pi}^{\star}_{j}Y_{j,\tau-1}^{\star}+\sum_{i=1}^{k}\hphi_{i}^{\star}(1-L^{4})Y_{\tau-i}^{\star}+\hat{\zeta}_{\tau}^{\star};$$
Step 6: repeat steps 3, 4, and 5 for $B$ times to get $B$ sets of t-statistics $t_{1}^{\star}$, $t_{2}^{\star}$, and $F$-statistics $F_{\cB}^{\star}$. Count separately the numbers of $t_{1}^{\star}$, $t_{2}^{\star}$, and $F_{\cB}^{\star}$, than which $t_{1}^{A}$, $t_{2}^{A}$, and the $F$-statistics $F_{\cB}^{A}$ are more extreme. If these numbers are higher than $B(1-size)$, then we consider $t_{1}^{A}$, $t_{2}^{A}$, and the $F$-statistics $F_{\cB}^{A}$ extreme, and reject the corresponding hypotheses.
\end{Algorithm}

\begin{Remark}
It is also reasonable to keep steps 1, 2, 3, 5, and 6 of the Algorithm \ref{seasonal iid bootstrap}, but change the generation of $\{Y_{4t+s}^{\star}\}$ in step 4 to
\begin{equation}
(1-L^{4})Y^{\star}_{4t+s}=\sum_{i=1}^{k}\hphi_{i,s}(1-L^{4})Y^{\star}_{4t+s-i}+\ep^{\star}_{4t+s}.
\label{algor}
\end{equation}
This new algorithm is in fact theoretically invalid for the tests of any coexistence of roots (see Remark \ref{pm1}, \ref{heter}, and \ref{root}), but it is valid for tests of any single roots at 1 or $-1$, due to the pivotal asymptotic distributions of $t_{1}^{A}$ and $t_{2}^{A}$ in Theorem \ref{aug real}. 
\end{Remark}
\begin{Remark}
\label{Re:seasonal iid bootstrap}
If we let steps 1, 3, 5, and 6 be the same as in Algorithm \ref{seasonal iid bootstrap}, but run non-periodic regression equations with non-periodic coefficients $\hat{\pi}_{j}^{A}$ and $\hphi_{i}$ in steps 2 and 4, then this version of algorithm is identical with \cite{burridge2004bootstrapping}. However, the step 2 of this new version cannot fully pre-whiten the time series, and consequently leaves the regression error $\{\hat{\zeta}_{\tau}^{A}\}$ serially correlated. When $\{\hat{\zeta}_{\tau}^{A}\}$ is bootstrapped by the seasonal iid bootstrap in step 3, this serial correlation structure is ruined. As a result, $(1-L^{4})Y_{4t+s}^{\star}$ differs from $(1-L^{4})Y_{4t+s}$ in its correlation structure, in particular $\bTheta(1)$, and consequently the conditional distributions of the bootstrap $F$-statistics $F_{\cB}^{\star}$ differ from the distributions of the original $F$-statistics $F_{\cB}^{A}$; see Remark \ref{pm1} and \ref{heter}. Similarly, the conditional distributions of the wild bootstrap $F$-statistics in \cite{cavaliere2017wild} differ from the real-world distributions of these $F$-statistics. 
\end{Remark}

\subsection{Consistency of seasonal iid bootstrap}
Now we justify the seasonal iid bootstrap augmented HEGY test (Algorithm \ref{seasonal iid bootstrap}). Since the derivation of the real-world asymptotic distributions in Theorem \ref{aug real} calls on FCLT (see Lemma \ref{le:unaug real1}), the justification of bootstrap approach also requires FCLT in the bootstrap world. From now on, let $P^{\circ}$, $E^{\circ}$, $Var^{\circ}$, $Std^{\circ}$, $Cov^{\circ}$ be the bootstrap probability, expectation, variance, standard deviation, and covariance, respectively, conditional on our data $\{Y_{4t+s}\}$.
\begin{Prop}\label{iid FCLT}
Suppose the assumptions in Theorem \ref{aug real} hold. Let $S_{T}^{\star}(u_{1},u_{2},u_{3},u_{4})$
$$=\frac{1}{\sqrt{4T}}\change{\bigg(}\sum_{\tau=1}^{\lfloor 4Tu_{1} \rfloor}\ep_{\tau}^{\star}/\sigma_{1}^{\star},\sum_{\tau=1}^{\lfloor 4Tu_{2} \rfloor}(-1)^{\tau}\ep_{\tau}^{\star}/\sigma_{2}^{\star},\sum_{\tau=1}^{\lfloor 4Tu_{3} \rfloor}\sqrt{2}\sin\change{\Big(}\frac{\pi \tau}{2}\change{\Big)}\ep_{\tau}^{\star}/\sigma_{3}^{\star},\sum_{\tau=1}^{\lfloor 4Tu_{4} \rfloor}\sqrt{2}\cos\change{\Big(}\frac{\pi \tau}{2}\change{\Big)}\ep_{\tau}^{\star}/\sigma_{4}^{\star}\change{\bigg)}',$$
where
\begin{align*}
\sigma_{1}^{\star}&=Std^{\circ}\change{\bigg[}\frac{1}{\sqrt{4T}}\sum _{\tau=1}^{ 4T }\ep_{\tau}^{\star}\change{\bigg]},&
\sigma_{2}^{\star}&=Std^{\circ}\change{\bigg[}\frac{1}{\sqrt{4T}}\sum _{\tau=1}^{ 4T }(-1)^{\tau}\ep_{\tau}^{\star}\change{\bigg]},\\
\sigma_{3}^{\star}&=Std^{\circ}\change{\bigg[}\frac{1}{\sqrt{4T}}\sum_{\tau=1}^{ 4T }\sqrt{2}\sin\change{\Big(}\frac{\pi \tau}{2}\change{\Big)}\ep_{\tau}^{\star}\change{\bigg]},&
\sigma_{4}^{\star}&=Std^{\circ}\change{\bigg[}\frac{1}{\sqrt{4T}}\sum _{\tau=1}^{ 4T }\sqrt{2}\cos\change{\Big(}\frac{\pi \tau}{2}\change{\Big)}\ep_{\tau}^{\star}\change{\bigg]}.
\end{align*}
Then, no matter which hypothesis is true, $S_{T}^{\star} \Rightarrow \bW^{\star}$ in probability as $T\rightarrow \infty$, where $\bW^{\star}(\cdot)$ is a four-dimensional standard Brownian motion.
\end{Prop}
By the FCLT given by Proposition \ref{iid FCLT} and the proof of Theorem \ref{aug real}, in probability the conditional distributions of $t^{\star}_{j}$, $j=1,2,$ and $F_{\cB}^{\star}$ converge to the limiting distributions of $t_{j}^{A}$, $j=1,2,$ and $F_{\cB}^{A}$, respectively. Indeed, since conditional on $\{Y_{4t+s}\}$, $\{Y_{4t+s}^{\star}\}$ is a finite-order seasonal AR process, the derivation of the conditional distributions of $t^{\star}_{j}$, $j=1,2,$ and $F_{\cB}^{\star}$ turns out easier than that of Theorem \ref{aug real}, and in particular does not involve the fourth moments of $\{Y_{4t+s}^{\star}\}$. Hence the consistency of the bootstrap. 
\begin{Theo}
Suppose the assumptions in Theorem \ref{aug real} hold. Let $P^{\cB}$ be the probability measure corresponding to the null hypothesis $H_{0}^{\cB}$. For example, $P^{1,2}$ corresponds to the null hypothesis $H_{0}^{1,2}$. Then,
$$\sup_{x}|P^{\circ}(t_{j}^{\star}\leq x)-P^{j}(t_{j}^{A}\leq x)|\stackrel{p}\rightarrow 0,\ j=1,2,$$
$$\sup_{x}|P^{\circ}(F_{\cB}^{\star}\leq x)-P^{\cB}(F_{\cB}^{A}\leq x)|\stackrel{p}\rightarrow 0, \ \text{where} \ \cB=\{1,2\},\{3,4\},\{1,3,4\},\{2,3,4\}, \ \text{or} \ \{1,2,3,4\}.$$
\end{Theo}

\section{Seasonal block bootstrap unaugmented HEGY test}

\subsection{Unaugmented HEGY test}
In the preceding section our analysis focuses on the augmented HEGY test, an extension of the ADF test to the seasonal unit root setting. An important alternative of the ADF test is the Phillips-Perron test (\cite{phillips1988}). While the ADF test assumes an AR structure over the noise and thus becomes parametric, its semi-parametric counterpart, Phillips-Perron test, allows a wide class of weakly dependent noises. The unaugmented HEGY test (\cite{breitung1998}), as the extension of Phillips-Perron test to the seasonal unit root, inherits the semi-parametric nature and does not assume the noise to be AR. Given periodic variation, it will be shown in Theorem \ref{unaug real} that the unaugmented HEGY test estimates seasonal unit roots consistently under a very general VMA$(\infty)$ class of noise (Assumption \ref{assump 1a}), instead of a more restrictive VARMA$(p,q)$ class of noise (Assumption \ref{assump 1b}), which we need for the augmented HEGY test. 

Now we specify the unaugmented HEGY test. Consider the non-periodic regression equation
\begin{equation}
(1-L^{4})Y_{\tau}=\sum_{j=1}^{4}\hat{\pi}_{j}^{U}Y_{j,\tau-1}+\hat{V}_{\tau},
\label{unaug HEGY}
\end{equation}
where $U$ stands for ``Unaugmented". Let $\hat{\pi}_{j}^{U}$ be the OLS estimator in \eqref{unaug HEGY}, $\hat{V}_{\tau}$ be the OLS residual, $t_{j}^{U}$ be the t-statistic corresponding to $\hat{\pi}_{j}^{U}$, and $F_{3,4}^{U}$ be the F-statistic corresponding to $\hat{\pi}_{3}^{U}$ and $\hat{\pi}_{4}^{U}$. Other $F$-statistics $F_{1,2}^{U}$, $F_{1,3,4}^{U}$, $F_{2,3,4}^{U}$, and $F_{1,2,3,4}^{U}$ can be defined analogously. Similar to the Phillips-Perron test, the unaugmented HEGY test can apply both $\hat{\pi}_{j}^{U}$ and $t_{j}^{U}$ when testing roots at 1 or $-1$. As in the augmented HEGY test, we reject $H_{0}^{1}$ if $\hat{\pi}_{1}^{U}$ (or $t_{1}^{U}$) is too small, reject $H_{0}^{2}$ if $\hat{\pi}_{2}^{U}$ (or $t_{2}^{U}$) is too small, and reject the joint hypotheses if the corresponding $F$-statistics are too large. The following results give the asymptotic null distributions of $\hat{\pi}_{j}^{U}$, $t_{j}^{U}$, $j=1,\dots,4$, and the $F$-statistics.
\subsection{Unaugmented HEGY test under model misspecification}
\begin{Theo} \label{unaug real}
Assume that Assumption \ref{assump 1a} and one of Assumption \ref{assump 2a} or Assumption \ref{assump 2b} hold. Then under $H_{0}^{1,2,3,4}$, as $T\rightarrow\infty$,
\begin{align*}
&(4T)\hat{\pi}_{j}^{U}\Rightarrow\frac{\lambda_{j}^{2}\int_{0}^{1} W_{j}(r)d W_{j}(r)+\Gamma^{(j)}}{\lambda_{j}^{2}\int_{0}^{1}W_{j}^{2}(r)dr}, \text{ for } j=1,2,\\    &(4T)\hat{\pi}_{3}^{U}\Rightarrow\frac{\lambda_{3}^{2}\int_{0}^{1} W_{3}(r)d W_{3}(r)+\lambda_{4}^{2}\int_{0}^{1} W_{4}(r)d W_{4}(r)+\Gamma^{(3)}}{\frac{1}{2}(\lambda_{3}^{2}\int_{0}^{1}W_{3}^{2}(r)dr+\lambda_{4}^{2}\int_{0}^{1}W_{4}^{2}(r)dr)},\\    
&(4T)\hat{\pi}_{4}^{U}\Rightarrow\frac{\lambda_{3}\lambda_{4}(\int_{0}^{1} W_{3}(r)d W_{4}(r)-\int_{0}^{1} W_{4}(r)d W_{3}(r))  +\Gamma^{(4)}}{\frac{1}{2}(\lambda_{3}^{2}\int_{0}^{1}W_{3}^{2}(r)dr+\lambda_{4}^{2}\int_{0}^{1}W_{4}^{2}(r)dr)},\\
t_{j}^{U}&\Rightarrow\frac{\lambda_{j}^{2}\int_{0}^{1} W_{j}(r)d W_{j}(r)+\Gamma^{(j)}}{\sqrt{\tgamma(0)\lambda_{j}^{2}\int_{0}^{1}W_{j}^{2}(r)dr}}\equiv\mathscr{D}_{j}, \text{ for } j=1,2,\\
t_{3}^{U}&\Rightarrow\frac{\lambda_{3}^{2}\int_{0}^{1} W_{3}(r)d W_{3}(r)+\lambda_{4}^{2}\int_{0}^{1} W_{4}(r)d W_{4}(r)+\Gamma^{(3)}}{\sqrt{\tgamma(0)\frac{1}{2}(\lambda_{3}^{2}\int_{0}^{1}W_{3}^{2}(r)dr+\lambda_{4}^{2}\int_{0}^{1}W_{4}^{2}(r)dr)}}\equiv\mathscr{D}_{3}\\
t_{4}^{U}&\Rightarrow\frac{\lambda_{3}\lambda_{4}(\int_{0}^{1} W_{3}(r)d W_{4}(r)-\int_{0}^{1} W_{4}(r)d W_{3}(r))  +\Gamma^{(4)}}{\sqrt{\tgamma(0)\frac{1}{2}(\lambda_{3}^{2}\int_{0}^{1}W_{3}^{2}(r)dr+\lambda_{4}^{2}\int_{0}^{1}W_{4}^{2}(r)dr)}}\equiv\mathscr{D}_{4}\\
F_{1,2}^{U}&\Rightarrow\frac{1}{2}(\mathscr{D}_{1}^{2}+\mathscr{D}_{2}^{2}), \quad
F_{3,4}^{U}\Rightarrow\frac{1}{2}(\mathscr{D}_{3}^{2}+\mathscr{D}_{4}^{2}), \\
F_{1,3,4}^{U}&\Rightarrow\frac{1}{3}(\mathscr{D}_{1}^{2}+\mathscr{D}_{3}^{2}+\mathscr{D}_{4}^{2}), \quad
F_{2,3,4}^{U}\Rightarrow\frac{1}{3}(\mathscr{D}_{2}^{2}+\mathscr{D}_{3}^{2}+\mathscr{D}_{4}^{2}), \\
F_{1,2,3,4}^{U}&\Rightarrow\frac{1}{4}(\mathscr{D}_{1}^{2}+\mathscr{D}_{2}^{2}+\mathscr{D}_{3}^{2}+\mathscr{D}_{4}^{2}),
\end{align*}
where $\bc_{1}=(1,1,1,1)'$, $\bc_{2}=(1,-1,1,-1)'$, $\bc_{3}=(0,-1,0,1)'$,  $\bc_{4}=(-1,0,1,0)'$,  $\lambda_{j}=\linebreak \sqrt{\bc_{j}'\bTheta(1)\bOmega\bTheta(1)'\bc_{j}/4}$, $W_{j}=\bc_{j}'\bTheta(1)\bOmega^{1/2}\bW/\change{(}2\lambda_{j}\change{)}$, $\bW(\cdot)$ is the same four-dimensional standard Brownian motion as in Theorem \ref{aug real}, $\tgamma(j)$ are defined in \eqref{eqn:tgamma},
$\Gamma^{(1)}=\sum_{j=1}^{\infty}\tgamma(j)$, $\Gamma^{(2)}=\sum_{j=1}^{\infty}(-1)^{j}\tgamma(j)$, $\Gamma^{(3)}=\sum_{j=1}^{\infty}\cos(\pi j/2)\tgamma(j)$, and $\Gamma^{(4)}=-\sum_{j=1}^{\infty}\sin(\pi j/2)\tgamma(j)$. 
\end{Theo}

\begin{Remark}
The results in Theorem \ref{unaug real} degenerate to the asymptotics in \cite{burridge2001} and \cite{burridge2001properties} when $\{V_{4t+s}\}$ is serially uncorrelated, to the asymptotics in \cite{breitung1998} when $\{V_{4t+s}\}$ has no periodic variation, and to the asymptotics in \cite{del2015semi} when $\{V_{4t+s}\}$ is seasonally heteroscedastic. 
\end{Remark}
\begin{Remark}
When $\{V_{4t+s}\}$ has no periodic variation, as in \cite{breitung1998}, the asymptotic distributions of $(\hat{\pi}_{1}^{U}, t_{1}^{U})$ and $(\hat{\pi}_{2}^{U}, t_{2}^{U})$ are independent. On the other hand, when $\{V_{4t+s}\}$ has periodic variation, $(\hat{\pi}_{1}^{U}, t_{1}^{U})$ and $(\hat{\pi}_{2}^{U}, t_{2}^{U})$ are dependent, as what we have seen for the augmented HEGY test in Remark \ref{pm1}. Hence, when testing $H_{0}^{1,2}$, it is problematic to test $H_{0}^{1}$ and $H_{0}^{2}$ separately and calculate the size of the test with the independence of $(\hat{\pi}_{1}^{U}, t_{1}^{U})$ and $(\hat{\pi}_{2}^{U}, t_{2}^{U})$ in mind. Instead, the test of $H_{0}^{1,2}$ should be handled with $F_{1,2}^{U}$. 
\end{Remark}
\begin{Remark}
The parameters $\lambda_{j}$ have the same definition as in Theorem \ref{aug HEGY}. Since $\lambda_{1}^{2}=\sum_{j=-\infty}^{\infty}\tgamma(j)$, and $\lambda_{2}^{2}=\sum_{j=-\infty}^{\infty}(-1)^{j}\tgamma(j)$, the asymptotic distributions of $\hat{\pi}_{j}^{U}$ and $t_{j}^{U}$, $j=1,2$, only depends on the autocorrelation function of $\{\tV_{\tau}\}$, a misspecified constant parameter representation of $\{V_{4t+s}\}$ whose autocovariance function is given by \eqref{eqn:tgamma}. Since $\{\tV_{\tau}\}$ can be considered as a non-periodic version of $\{V_{4t+s}\}$, we can conclude that the asymptotic behaviors of the tests for $H_{0}^{1}$ and $H_{0}^{2}$ are not affected by the periodic variation in $\{V_{4t+s}\}$. On the other side, the asymptotic distributions of the $F$-statistics do not solely depend on the distribution of $\{\tV_{\tau}\}$. Hence, the tests for all hypotheses other than $H_{0}^{1}$ and $H_{0}^{2}$ are affected by the periodic variation.
\end{Remark}

\begin{Remark}
To remove the nuisance parameters in the asymptotic distributions, we notice that the asymptotic behaviors of $\hat{\pi}_{j}^{U}$ and $t_{j}^{U}$, $j=1,2,$ have identical forms as in \cite{phillips1988}. In light of their approach, we can construct pivotal versions of $\hat{\pi}_{j}^{U}$ and $t_{j}^{U}$, $j=1,2,$ that converge in distribution to standard Dickey-Fuller distributions in \cite{dickey1979distribution}; see also \cite{del2015semi}. More specifically, for $j=1,2$, by Theorem \ref{unaug real} we have
\begin{align*}
(4T)\hat{\pi}_{j}^{U}-\frac{\frac{1}{2}(\lambda_{j}^{2}-\tgamma(0))}{(4T)^{-2}\sum_{\tau=1}^{4T}Y_{j,\tau-1}^{2}}&\Rightarrow\frac{\int_{0}^{1}W_{j}(r)dW_{j}(r)}{\int_{0}^{1}W_{j}^{2}(r)dr},\\
\frac{\sqrt{\tgamma(0)}}{\lambda_{j}}t_{j}^{U}-\frac{\frac{1}{2}(\lambda_{j}^{2}-\tgamma(0))}{\lambda_{j}^{2}\sqrt{(4T)^{-2}\sum_{\tau=1}^{4T}Y_{j,\tau-1}^{2}}}&\Rightarrow\frac{\int_{0}^{1}W_{j}(r)dW_{j}(r)}{\sqrt{\int_{0}^{1}W_{j}^{2}(r)dr}}.
\end{align*}
where $\lambda_{j}^{2}$ and $\tgamma(0)$ can by substituted by their consistent estimators.
\end{Remark}

\begin{Remark}
However, there is no easy way to construct pivotal statistics for $\hat{\pi}_{3}^{U}$, $t_{3}^{U}$, $\hat{\pi}_{4}^{U}$, $t_{4}^{U}$, and $F$-statistics such as $F_{3,4}^{U}$. The difficulties are two-fold. Firstly the denominators of the asymptotic distributions of these statistics contain weighted sums with unknown weights $\lambda_{3}^{2}$ and $\lambda_{4}^{2}$; secondly $W_{3}$ and $W_{4}$ are in general correlated standard Brownian motions as in Theorem \ref{aug real}. 
\end{Remark}
\begin{Remark}
The result in Theorem \ref{unaug real} can be generalized. Suppose $\{Y_{4t+s}\}$ is not generated by $H_{0}^{1,2,3,4}$, and only has some of the seasonal unit roots. Let $U_{\tau}=(1-L^{4})Y_{\tau}$, and $\bU_{t}=(U_{4t-3},U_{4t-2},U_{4t-1},U_{4t})'$. Define $\bH(z)$ such that $\bU_{t}=\bH(\bB)\bep_{t}$. The asymptotic distributions of $\hat{\pi}_{j}^{U}$, $t_{j}^{U}$, $j=1,2,$ and the $F$-statistics have the same forms as those in Theorem \ref{unaug real}, with $\bTheta(1)$ substituted by $\bH(1)$, and $\tgamma$ based on $\{U_{\tau}\}$. 
\end{Remark}
\begin{Remark}
Under one of the alternative hypotheses, we conjecture that for $j=1,2,3$, the OLS estimators $\hat{\pi}_{j}^{U}$ in \eqref{unaug HEGY} converge in probability to $\pi_{j}$, the prediction coefficient of the misspecified constant parameter representation of $\{Y_{4t+s}\}$. Since under the alternative hypotheses we can without loss of generality assume $\{Y_{4t+s}\}$ is stationary, we have $\pi_{j}<0$. Hence, as a result of this conjecture, the power of the unaugmented HEGY tests tends to one as the sample size goes to infinity.
\end{Remark}
\subsection{Seasonal block bootstrap algorithm}
Since many of the asymptotic distributions delivered in Theorem \ref{unaug real} are non-standard and non-pivotal and cannot be easily pivoted, we propose the application of bootstrap. Since the regression error $\{V_{4t+s}\}$ of \eqref{unaug HEGY} has periodic structure, we may apply the seasonal block bootstrap of \cite{dudek2014generalized}. The algorithm of the seasonal block bootstrap unaugmented HEGY test is illustrated below. 
\begin{Algorithm}\label{seasonal block boot}
Step 1: get the OLS estimators $\hat{\pi}_{1}^{U}$, $\hat{\pi}_{2}^{U}$, t-statistics $t_{1}^{U}$, $t_{2}^{U}$, and the $F$-statistics $F_{\cB}^{U}$, $\cB=\{1,2\},\{3,4\},\{1,3,4\},\{2,3,4\},$ and $\{1,2,3,4\},$ from the unaugmented HEGY regression  
$$ (1-L^{4})Y_{\tau}=\sum_{j=1}^{4}\hat{\pi}_{j}^{U}Y_{j,\tau-1}+\hat{\zeta}_{\tau}^{U}, \quad \tau=1,\dots,4T; $$
Step 2: record residual $\hat{V}_{4t+s}$ from regression
$$(1-L^{4})Y_{4t+s}=\sum_{j=1}^{4}\hat{\pi}_{j,s}^{U}Y_{j,4t+s-1}+\hat{V}_{4t+s};$$
Step 3: let $\check{V}_{4t+s}=\hat{V}_{4t+s}-\frac{1}{T}\sum_{t=1}^{T}\hat{V}_{4t+s}$, choose a integer block size $b$, and let $l=\lfloor 4T/b \rfloor$. For $t=1, b+1,\dots, (l-1)b+1$, let
$$(V_{t}^{*},\dots,V_{t+b-1}^{*})=(\check{V}_{I_{t}},\dots,\check{V}_{I_{t}+b-1}),$$
where $\{I_{t}\}$ is a sequence of iid uniform random variables taking values in $\{t-4R_{1,n},\dots,t-4,t,t+4,\dots,t+4R_{2,n}\}$
with $R_{1,n}=\lfloor (t-1)/4 \rfloor$ and $R_{2,n}=\lfloor (n-b-t+1)/4 \rfloor;$\\
\\
Step 4: set the $\hat{\pi}_{j,s}^{U}$ corresponding to the null hypothesis to be zero. For example, set $\hat{\pi}_{3,s}^{U}=\hat{\pi}_{4,s}^{U}=0$ for all $s$ when testing roots at $\pm i$. Generate $\{Y_{4t+s}^{*}\}$ by
$$(1-L^{4})Y_{4t+s}^{*}=\sum_{j=1}^{4}\hat{\pi}_{j,s}^{U}Y^{*}_{j,4t+s-1}+V_{4t+s}^{*};$$
Step 5: get OLS estimates $\hat{\pi}_{1}^{*}$, $\hat{\pi}_{2}^{*}$, t-statistics $t_{1}^{*}$, $t_{2}^{*}$, and $F$-statistics $F_{\cB}^{*}$ from regression
$$(1-L^{4})Y^{*}_{\tau}=\sum_{j=1}^{4}\hat{\pi}^{*}_{j}Y_{j,\tau-1}^{*}+\hat{\zeta}_{\tau}^{*}, \quad \tau=1,\dots,4T; $$
Step 6: repeat steps 3, 4, and 5 for $B$ times to get $B$ sets of statistics $\hat{\pi}_{1}^{*}$, $\hat{\pi}_{2}^{*}$, $t_{1}^{*}$, $t_{2}^{*}$, and $F_{\cB}^{*}$. Count separately the numbers of $\hat{\pi}_{1}^{*}$, $\hat{\pi}_{2}^{*}$, $t_{1}^{*}$, $t_{2}^{*}$, and $F_{\cB}^{*}$ than which $\hat{\pi}_{1}^{U}$, $\hat{\pi}_{2}^{U}$, $t_{1}^{U}$, $t_{2}^{U}$, and $F_{\cB}^{U}$ are more extreme. If these numbers are higher than $B(1-size)$, then consider $\hat{\pi}_{1}^{U}$, $\hat{\pi}_{2}^{U}$, $t_{1}^{U}$, $t_{2}^{U}$ and $F_{\cB}^{U}$ extreme, and reject the corresponding hypotheses. 
\end{Algorithm}

\subsection{Consistency of seasonal block bootstrap}
\begin{Prop}\label{SBB FCLT}
Let $S_{T}^{*}(u_{1},u_{2},u_{3},u_{4})$
$$=\frac{1}{\sqrt{4T}}\change{\bigg(}\sum_{\tau=1}^{\lfloor 4Tu_{1} \rfloor}V_{\tau}^{*}/\sigma_{1}^{*},\sum_{\tau=1}^{\lfloor 4Tu_{2} \rfloor}(-1)^{\tau}V_{\tau}^{*}/\sigma_{2}^{*},\sum_{\tau=1}^{\lfloor 4Tu_{3} \rfloor}\sqrt{2}\sin\change{\Big(}\frac{\pi \tau}{2}\change{\Big)}V_{\tau}^{*}/\sigma_{3}^{*},\sum_{\tau=1}^{\lfloor 4Tu_{4} \rfloor}\sqrt{2}\cos\change{\Big(}\frac{\pi \tau}{2}\change{\Big)}V_{\tau}^{*}/\sigma_{4}^{*}\change{\bigg)}', $$
where
\begin{align*}
\sigma_{1}^{*}&=Std^{\circ}\change{\bigg[}\frac{1}{\sqrt{4T}}\sum _{\tau=1}^{ 4T }V_{\tau}^{*}\change{\bigg]},&
\sigma_{2}^{*}&=Std^{\circ}\change{\bigg[}\frac{1}{\sqrt{4T}}\sum _{\tau=1}^{ 4T }(-1)^{\tau}V_{\tau}^{*}\change{\bigg]},\\
\sigma_{3}^{*}&=Std^{\circ}\change{\bigg[}\frac{1}{\sqrt{4T}}\sum _{\tau=1}^{ 4T }\sqrt{2}\sin\change{\Big(}\frac{\pi \tau}{2}\change{\Big)}V_{\tau}^{*}\change{\bigg]},&
\sigma_{4}^{*}&=Std^{\circ}\change{\bigg[}\frac{1}{\sqrt{4T}}\sum_{\tau=1}^{ 4T }\sqrt{2}\cos\change{\Big(}\frac{\pi \tau}{2}\change{\Big)}V_{\tau}^{*}\change{\bigg]}.
\end{align*}
If $b\rightarrow \infty$, $T\rightarrow \infty$, $b/\sqrt{T}\rightarrow 0$, then no matter which hypothesis is true, $S_{T}^{*} \Rightarrow \bW^{*}$ in probability, where $\bW^{*}(\cdot)$ is a four-dimensional standard Brownian motion.
\end{Prop}
By the FCLT given by Proposition \ref{SBB FCLT}, the proof of Theorem \ref{unaug real}, and the convergence of the bootstrap standard deviation $\sigma_{j}^{*}$ in \cite{dudek2014generalized}, we have that the conditional distributions of $t^{*}_{j}$, $\hat{\pi}^{*}_{j}$, $j=1,2,$ and $F_{\cB}^{*}$ in probability converges to the limiting distributions of $\hat{\pi}_{j}^{U}$, $t_{j}^{U}$, $j=1,2,$ and $F_{\cB}^{U}$, respectively. Hence the consistency of the bootstrap.
\begin{Theo}
Suppose the assumptions in Theorem \ref{unaug real} hold. Let $P^{\cB}$ be the probability measure corresponding to the null hypothesis $H_{0}^{\cB}$. For example, $P^{1,2}$ corresponds to the null hypothesis $H_{0}^{1,2}$. If $b\rightarrow \infty$, $T\rightarrow \infty$, $b/\sqrt{T}\rightarrow 0$, then
$$\sup_{x}|P^{\circ}(\pi_{j}^{*}\leq x)-P^{j}(\hat{\pi}^{U}_{j}\leq x)|\stackrel{p}\rightarrow 0,\ j=1,2,$$
$$\sup_{x}|P^{\circ}(t_{j}^{*}\leq x)-P^{j}(t_{j}^{U}\leq x)|\stackrel{p}\rightarrow 0,\ j=1,2,$$
$$\sup_{x}|P^{\circ}(F_{\cB}^{*}\leq x)-P^{\cB}(F_{\cB}^{U}\leq x)|\stackrel{p}\rightarrow 0, \ \text{where} \ \cB=\{1,2\},\{3,4\},\{1,3,4\},\{2,3,4\},\ \text{or} \ \{1,2,3,4\}.$$
\end{Theo}

\section{Simulation}
\subsection{Data generating process}\label{section:dgp}

We focus on the hypothesis testing for root at 1, \change{i.e.,} $H_{0}^{1}$ against $H_{1}^{1}$, roots at $\pm i$, \change{i.e.,} $H_{0}^{3,4}$ against $H_{1}^{3,4}$, and \change{roots at $1$, $-1$, and $\pm i$, i.e., $H_{0}^{1,2,3,4}$ against $H_{1}^{1,2,3,4}$}. In the first two hypothesis tests, we equip one sequence with all nuisance unit roots, and the other with none of the nuisance unit roots. The detailed data generating processes are listed in Table \ref{DGP}. \change{In an unreported simulation we have simulated the hypothesis test for root at $-1$, i.e., $H_{0}^{2}$ against $H_{1}^{2}$, but found the simulation result to a large extent similar to the result of root at $1$.}  

To produce power curves, we let parameter \change{$\rho=0.00, 0.02, 0.04, 0.06, 0.08,$ and $0.10$}. Notice that $\rho$ is set to be seasonally homogeneous for the sake of simplicity. Further, we generate six types of innovations $\{V_{4t+s}\}$ according to Table \ref{noise}, where $\ep_{t}\sim iid \: N(0,1)$. The values of $\phi_{s}$ \change{in Table \ref{noise}} are assigned so that the misspecified constant parameter representation (see Section \ref{sec:settings}) of the ``ar\textsubscript{per}'' sequence has almost the same AR structure as the ``ar\textsubscript{pos}'' sequence. \change{Notice that in the ``ma\textsubscript{per}'' setting in Table \ref{noise}, $V_{4t+s}=(1-\theta_{s}L)^{-1}(1-\theta_{s}\theta_{s-1}L^{2})\ep_{4t+s}$; the values of $\theta_{s}$ are assigned such that a potential seasonal unit root filter $(1+L^{2})$ is partially cancelled out by the MA filter $(1-\theta_{s}\theta_{s-1}L^{2})$ above.} 
\begin{table}[H]
\centering
\caption{Data generation processes}
\begin{tabular}{cccc}

\hline
\multicolumn{2}{c}{\multirow{2}{*}{\begin{tabular}[c]{@{}c@{}}Data Generating\\ Processes\end{tabular}}} & \multicolumn{2}{c}{Nuisance Root}                                  \\ \cline{3-4} 
\multicolumn{2}{c}{}                                                                                     & No                          & Yes                                   \\ \hline
\multirow{3}{*}{Root}                                       & 1                                            & $(1-(1-\rho) L)Y_{\tau}=V_{\tau}$     & $(1+L)(1+L^{2})(1-(1-\rho) L)Y_{\tau}=V_{\tau}$ \\ \cline{2-4} 
                                                             & $\pm i$                                      & $(1+(1-\rho) L^{2})Y_{\tau}=V_{\tau}$ & $(1+L)(1-L)(1+(1-\rho) L^{2})Y_{\tau}=V_{\tau}$ \\ \cline{2-4} 
                                                            &  \change{$1, -1, \pm i$}                                      &  \change{$(1-(1-\rho) L^{4})Y_{\tau}=V_{\tau}$} &  \\ \hline
\end{tabular}
\label{DGP}
\end{table}
\begin{table}[H]
\centering
\caption{Types of noises}
\label{noise}
\begin{tabular}{ccc}
\hline
\multirow{6}{*}{\begin{tabular}[c]{@{}c@{}}Noise\\ Type\end{tabular}} & iid                     & $V_{\tau}=\ep_{\tau}$                                                                                                                                       \\ \cline{2-3} 
                                                                      & heter                   & \begin{tabular}[c]{@{}c@{}}$V_{4t+s}=\sigma_{s}\ep_{4t+s}$,\\ $\sigma_{1}=10$, $\sigma_{2}=\sigma_{3}=\sigma_{4}=1$\end{tabular}                      \\ \cline{2-3} 
                                                                                                                            & ar\textsubscript{pos}                      & $V_{\tau}=\ep_{\tau}+0.5V_{\tau-1}$                                                            \\ \cline{2-3} 
                                                                                                                                        & ma\textsubscript{neg} & $V_{\tau}=\ep_{\tau}-0.5\ep_{\tau-1}$                                                             \\ \cline{2-3} 
                                                                      & ar\textsubscript{per}                  & \begin{tabular}[c]{@{}c@{}}$V_{4t+s}=\ep_{4t+s}+\phi_{s}V_{4t+s-1}$, \\ $\phi_{1}=0.2$, $\phi_{2}=0.45$, $\phi_{3}=0.65$, $\phi_{4}=0.8$\end{tabular} \\ \cline{2-3}
                                                                      
                                                                                                                                           &  \change{ma\textsubscript{per}}                      & \change{\begin{tabular}[c]{@{}c@{}}$V_{4t+s}=\ep_{4t+s}+\theta_{s}\ep_{4t+s-1}$, \\ $\theta_{1}=0.5$, $\theta_{2}=-1.8$, $\theta_{3}=0.5$, $\theta_{4}=-1.8$\end{tabular}} \\ \hline
\end{tabular}
\end{table}

\subsection{Testing procedure}
Here we give additional implemental details for the algorithms of the seasonal iid bootstrap augmented HEGY test described in Algorithm \ref{seasonal iid bootstrap}, the seasonal block bootstrap unaugmented HEGY test described in Algorithm \ref{seasonal block boot}, the non-seasonal bootstrap augmented HEGY test by \cite{burridge2004bootstrapping}, and the Wald test by \cite{ghysels1996periodic}.
\subsubsection{Seasonal iid bootstrap augmented HEGY test}
To improve the empirical performance of seasonal iid bootstrap algorithm (Algorithm \ref{seasonal iid bootstrap}), we select stepwise, truncate the coefficient estimators, and apply \eqref{algor} when testing roots at 1 or $-1$. Firstly, a stepwise selection procedure is applied to the regression in step 2 of Algorithm \ref{seasonal iid bootstrap}. To begin with, we choose a maximal order of lag $k_{max}$. $k_{max}$ may be chosen by AIC, BIC, or modified information criterion by \cite{ng2001} (for further discussions, see \cite{del2016performance}). In our simulation we fix $k_{max}=4$ for simplicity. Afterward, we apply a backward stepwise selection with Variance Inflating Factor (VIF) criterion to solve the multicollinearity between the regressors. In this selection, we locate the regressor with the largest VIF, remove this regressor from the regression if its VIF is larger than 10, and rerun the regression. Then we implement another stepwise selection on lags $(1-L^{4})Y_{4t+s-i}$, $i=1,2,\dots,k$, by iteratively removing lags of which the absolute values of the t-statistics are smaller than 1.65; see also \cite{burridge2004bootstrapping}. Then the estimated coefficients of the deleted regressors are set to be zero, while the estimated coefficients of the remaining regressors are recorded and used in step 2 and 4. The backward stepwise selection of the lags based on their t-statistics is also applied to step 1 and 5.

Secondly, notice that in step 2, the true parameters $\pi_{j,s}$, $j=1,2,3,$ are smaller or equal to zero under both null and alternative hypotheses. However, the OLS estimators $\hat{\pi}_{j,s}^{A}$, $j=1,2,3,$ are often positive, especially when $\pi_{j,s}=0$. This positivity not only renders the estimation of $\pi_{j,s}$ inaccurate, 
but also makes the equation in step 4 of Algorithm \ref{seasonal iid bootstrap} non-causal, and the bootstrap sequence $\{Y_{4t+s}^{\star}\}$ explosive. The solution of this problem is to truncate the OLS estimator. Let $\check{\pi}_{j,s}^{A}=\min(0,\hat{\pi}_{j,s}^{A})$, $j=1,2,3$. Immediately we get $|\check{\pi}_{j,s}^{A}-\pi_{j,s}^{A}|\leq|\hat{\pi}_{j,s}^{A}-\pi_{j,s}|$. After we substitute $\check{\pi}_{j,s}^{A}$ for $\hat{\pi}_{j,s}^{A}$ in step 4, the empirical performance of seasonal iid bootstrap improves significantly. 

\change{Thirdly, by Assumption \ref{assump 1b}, intuitively the true parameters $\phi_{i,s}$ in step 2 should make the roots of polynomial $\phi_{s}(z)\coloneqq1-\phi_{1,s}z-\phi_{2,s}z^{2}-\dots-\phi_{k,s}z^{k}$ staying outside the unit circle. On the other hand, the roots of polynomial $\hphi_{s}(z)\coloneqq 1-\hphi_{1,s}z-\hphi_{2,s}z^{2}-\dots-\hphi_{k,s}z^{k}$ are sometimes on or inside the unit circle. To correct $\hphi_{s}(z)$, suppose that $\hphi_{s}(z)$ can be factored out as $\hphi_{s}(z)=(1-r_{1,s}z)(1-r_{2,s}z)\cdots(1-r_{k,s}z)$, where $r_{j,s}$, $j = 1,2,\dots, k,$ are complex numbers. Let $\tr_{j,s}=(r_{j,s}/|r_{j,s}|)\cdot\min(1/1.1, |r_{j,s}|)$, $j=1,2,\dots,k$. Then  $1-\thphi_{1,s}z-\thphi_{2,s}z^{2}-\dots-\thphi_{k,s}z^{k}\coloneqq\thphi_{s}(z)\coloneqq(1-\tr_{1,s}z)(1-\tr_{2,s}z)\cdots(1-\tr_{k,s}z)$ has all roots outside the unit circle. After we substitute $\thphi_{i,s}$ for $\hphi_{i,s}$ in step 4, the simulation result improves.}

Fourthly, we apply the original step 4 of Algorithm \ref{seasonal iid bootstrap} when testing roots at $\pm i$, but apply the alternative step \eqref{algor} to the test of the root at 1 or $-1$. (When apply the alternative step \eqref{algor}, we select the lags and truncate the coefficients similarly.) Unpublished simulation result shows an advantage of \eqref{algor} when testing root at 1 or $-1$. This advantage occurs especially when all nuisance roots occur, or equivalently when all of the true $\pi_{j,s}$'s are zero, since in this case the inclusion of $Y^{\star}_{j,4t+s-1}$ in the original step 4 becomes redundant. 

\subsubsection{Seasonal block bootstrap unaugmented HEGY test}
To improve the empirical performance of the seasonal block bootstrap algorithm (Algorithm \ref{seasonal block boot}), we truncate the coefficient estimators, taper the blocks, and optimize the block size. Firstly, as in the seasonal iid bootstrap algorithm, we let $\check{\pi}_{j,s}^{U}=\min(0,\hat{\pi}_{j,s}^{U})$, $j=1,2,3$, and substitute $\check{\pi}_{j,s}^{U}$ for $\hat{\pi}_{j,s}^{U}$ in step 4. 

Secondly, it is known that the bootstrapped data around the edges of the bootstrap blocks are not good imitations of the original data. To reduce this ``edge effect", we apply tapered seasonal block bootstrap proposed by \cite{dudek2016generalized}, which puts less weight on the bootstrapped data around the edges. In our simulation the weight function is set identical to the function suggested by \cite{dudek2016generalized}.

Thirdly, both test statistics $\hat{\pi}_{j}^{U}$ and $t_{j}^{U}$ can be employed to run the seasonal block bootstrap unaugmented HEGY test. So do various block sizes. In an unreported simulation we check the impact of test statistics and block sizes on the empirical size and power. It turns out that, first, the choice of statistics and block sizes does not affect the empirical size and the power very much; second, the distortion of the empirical size becomes the worst when testing root at $-1$ with the presence of nuisance roots and $ma_{pos}$ noise; third, the bootstrap test based on the t-statistics and block size four gives the best result in the aforementioned worst scenario. Hence, we base our test on the t-statistics and let the block size be four in the succeeding simulations. For a thorough discussion on an optimal block size, see \cite{paparoditis2003}.

\subsubsection{Non-seasonal bootstrap augmented HEGY test}
\change{For a brief description of the non-seasonal bootstrap augmented HEGY test by \cite{burridge2004bootstrapping}, see Remark \ref{Re:seasonal iid bootstrap}. To improve its empirical performance, as in the seasonal iid bootstrap algorithm, we apply a backward stepwise selection of the lags based on their t-statistics and correct $\hphi_{j}$, $j=1,2,\dots, k$ with a polynomial factorization.  
}

\subsubsection{Wald test}
We find it necessary to pass the data through a $(1-L^{4})$ filter before sending it to the Wald test by \cite{ghysels1996periodic}; otherwise the nuisance roots in our data will result in a non-stationary noise sequence in the regression of the Wald test and a ill-behaved test statistic; see also \cite{boswijk1996unit,osborn2002asymptotic}. When selecting the order of lag of the regression, we refer to the AIC and set the largest possible order of lag to “be four. 

\subsection{Results}
Now we present in Figure \ref{fig:root1_F}, \ref{fig:root1_T}, \ref{fig:root34_F}, \ref{fig:root34_T}, and \ref{fig:root1234_F} the main simulation result. The simulation includes five types of data generating processes (see Table \ref{DGP}) and six types of noises (see Table \ref{noise}). In simulation we let sample size be $T=30$ or $T=120$, number of bootstrap replicates $B=500$, number of iterations \change{$N=2400$}, and nominal size $\alpha=0.05$. 
\subsubsection{Root at 1}
Figure \ref{fig:root1_F} gives the simulation result when our data has a potential root at 1 but no other nuisance roots at $-1$ or $\pm i$. \change{In this scenario, the seasonal block bootstrap test and the non-seasonal bootstrap test suffer from a slight size distortion in (f) and (l), where the seasonal iid test enjoys more accurate size. Except that,} the power curves of the three bootstrap tests almost overlap; they start at the correct size and tend to one when $\rho$ departs from zero, \change{get higher when the sample size grows from $T=30$ to $T=120$}, and are far above the curves of the Wald test in all of (a)-(l) but (b) and (h), where the Wald test suffers a upward size distortion. 

Figure \ref{fig:root1_T} gives the result when data has a potential root at 1 and all nuisance roots at $-1$ and $\pm i$. Notice that the size of the seasonal block bootstrap unaugmented HEGY test is distorted in (b) and (h) in Figure \ref{fig:root1_T}; this may result from the errors in estimating $\pi_{j,s}$ and the need to recover $\{Y_{4t+s}\}$ with the estimated $\pi_{j,s}$. Moreover, the size of both the seasonal block bootstrap test and \change{the non-seasonal bootstrap test} is distorted in (d) and (j); this is in part due to the fact that the unit root filter $(1-L)$ is partially cancelled by the MA filter $(1-0.5L)$. See also \cite{perron1996}. 

In contrast, the seasonal iid bootstrap augmented HEGY test has less size distortion when data has nuisance roots. This is partially because the seasonal iid bootstrap test recovers $\{Y_{4t+s}\}$ using the true values of $\pi_{j,s}$, namely zero, instead of using the estimated values. Moreover, compared to the Wald test, the seasonal iid bootstrap augmented HEGY test has much higher power. Therefore, when testing the root at 1, the seasonal iid bootstrap augmented HEGY test is recommended.

\subsubsection{Root at $\pm i$}
Figure \ref{fig:root34_F} and Figure \ref{fig:root34_T} present the results of the simulation when data has potential roots at $\pm i$ but has no or all nuisance roots at 1 and $-1$, respectively. \change{In both Figure \ref{fig:root34_F} and Figure \ref{fig:root34_T}, it turns out that all the three bootstrap tests have size distortions in (f) and (l), where, as discussed in Section \ref{section:dgp}, the seasonal unit root filter $(1+L^{2})$ is partially cancelled out by the MA filter $(1-\theta_{s}\theta_{s-1}L^{2})$. Other than that, the bootstrap tests overall achieve the correct size. Since in Figure \ref{fig:root34_F} and Figure \ref{fig:root34_T} the seasonal block bootstrap test overall has higher power than other bootstrap tests and than the Wald test, we recommend it for testing roots at $\pm i$.}

\change{
\subsubsection{Root at $1$, $-1$, and $\pm i$}
Figure \ref{fig:root1234_F} illustrates the result when we test the concurrence of roots at $1$, $-1$, and $\pm i$. Notice that all the three bootstrap tests have distorted sizes in (f) and (l), where the seasonal unit root filter $(1+L^{2})$ is partially cancelled out by the MA filter $(1-\theta_{s}\theta_{s-1}L^{2})$. In addition, when the sample size $T=30$, all the three bootstrap tests suffer size distortion in (c), (d), and (e). However, when the sample size rises to $T=120$, the seasonal iid bootstrap augmented HEGY test restores the correct size; see (i), (j), and (k). Since overall in Figure \ref{fig:root1234_F} the seasonal iid bootstrap test prevails over the Wald test, we recommend it for testing joint roots at $1$, $-1$, and $\pm i$. 
}

\begin{figure}[H]
\centering

\subfloat[][noise=iid, $T$=30]{
\includegraphics[width=0.25\textwidth]{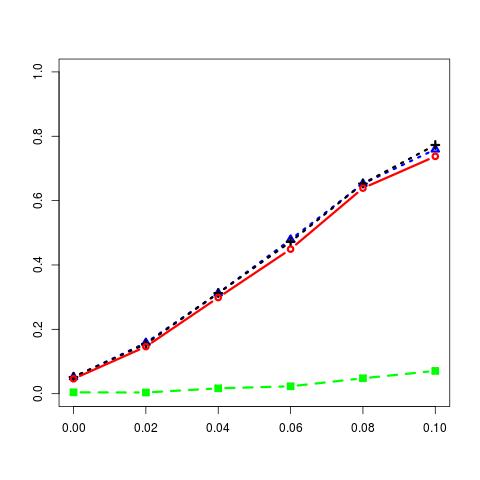}}
\qquad
\subfloat[][noise=heter, $T$=30]{
\includegraphics[width=0.25\textwidth]{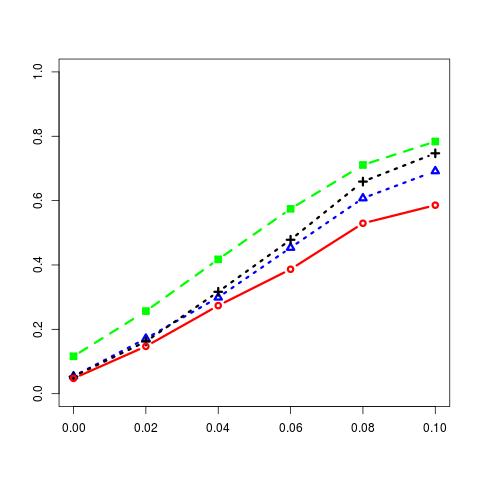}}
\qquad
\subfloat[][noise=$ar_{pos}$, $T$=30]{
\includegraphics[width=0.25\textwidth]{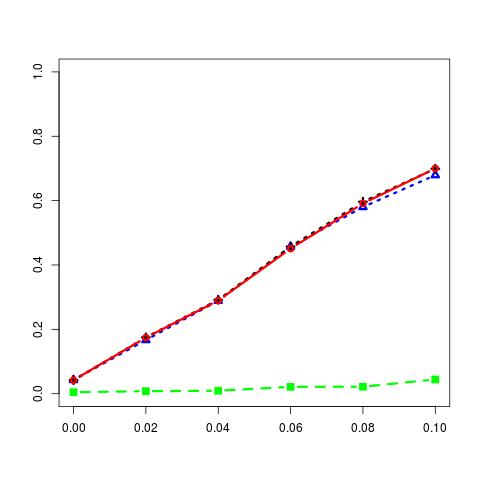}}

\subfloat[][noise=$ma_{neg}$, $T$=30]{
\includegraphics[width=0.25\textwidth]{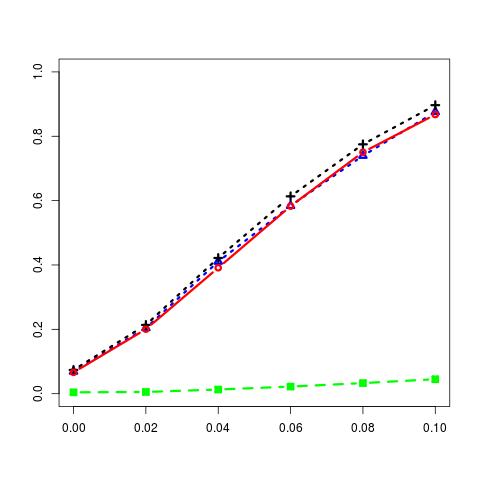}}
\qquad
\subfloat[][noise=$ar_{per}$, $T$=30]{
\includegraphics[width=0.25\textwidth]{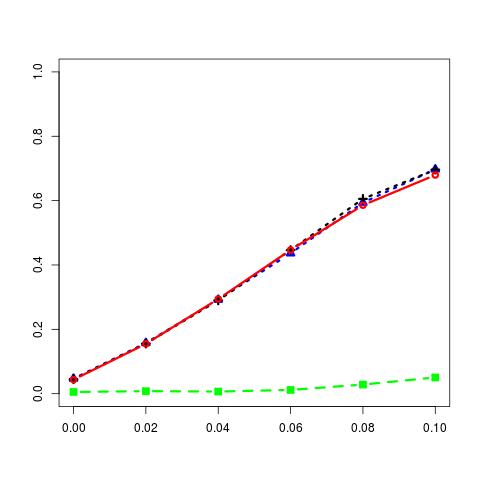}}
\qquad
\subfloat[][noise=$ma_{per}$, $T$=30]{
\includegraphics[width=0.25\textwidth]{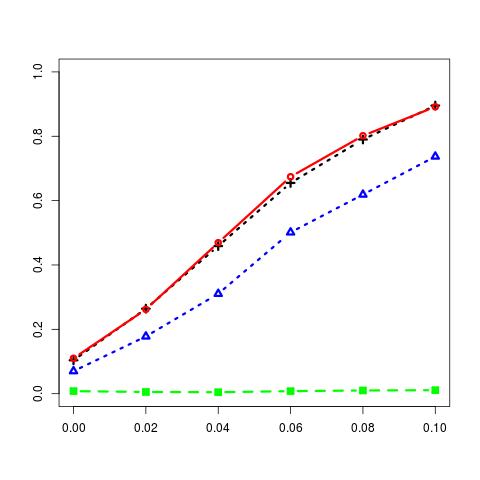}}
\qquad

\subfloat[][noise=iid, $T$=120]{
\includegraphics[width=0.25\textwidth]{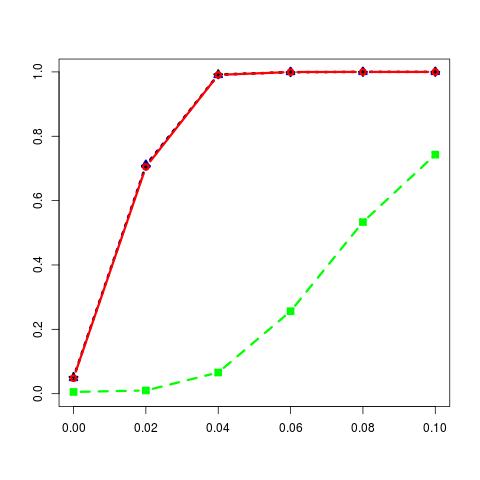}}
\qquad
\subfloat[][noise=heter, $T$=120]{
\includegraphics[width=0.25\textwidth]{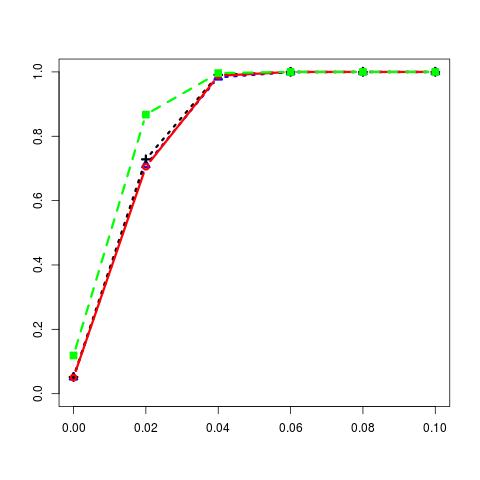}}
\qquad
\subfloat[][noise=$ar_{pos}$, $T$=120]{
\includegraphics[width=0.25\textwidth]{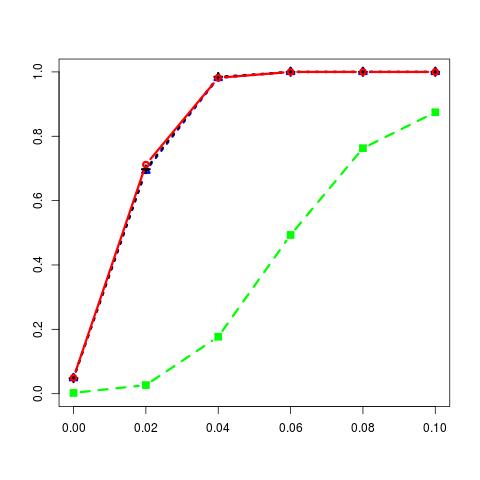}}

\subfloat[][noise=$ma_{neg}$, $T$=120]{
\includegraphics[width=0.25\textwidth]{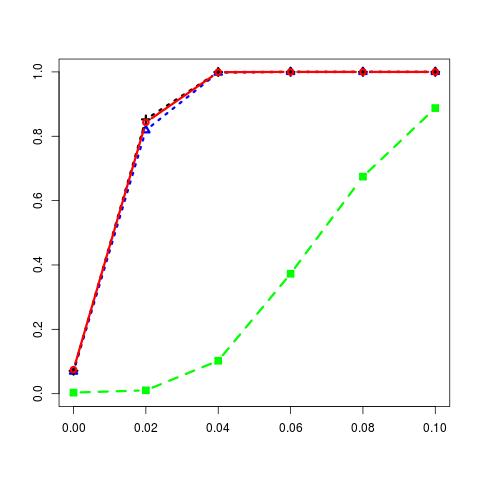}}
\qquad
\subfloat[][noise=$ar_{per}$, $T$=120]{
\includegraphics[width=0.25\textwidth]{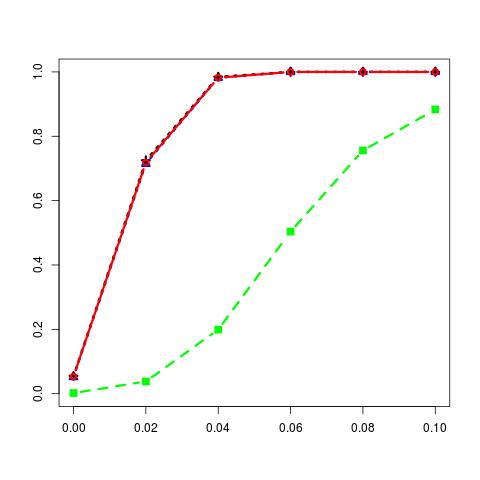}}
\qquad
\subfloat[][noise=$ma_{per}$, $T$=120]{
\includegraphics[width=0.25\textwidth]{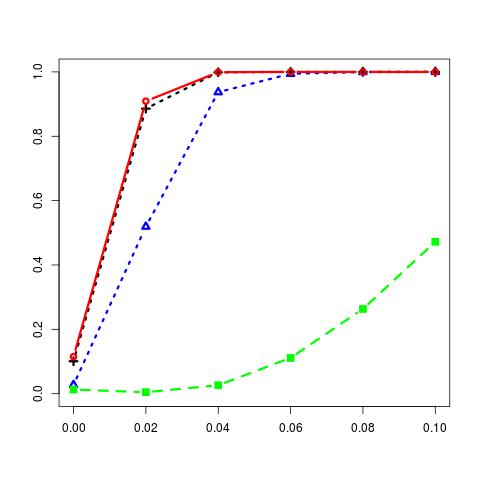}}
\qquad
\caption{Power as a function of $\rho$ when testing roots at 1 with no nuisance root. Blue dotted curve with triangle knot is for the seasonal iid bootstrap test. Red solid curve with circle knot is for the seasonal block bootstrap test. Black dotted curve with ``+" knot is for the non-seasonal bootstrap test. Green dashed curve with square knot is for the Wald test. In (a)-(f) sample size $T=30$. In (g)-(l) sample size $T=120$.  }
\label{fig:root1_F}
\end{figure}

\begin{figure}[H]
\centering
\subfloat[][noise=iid, $T$=30]{
\includegraphics[width=0.25\textwidth]{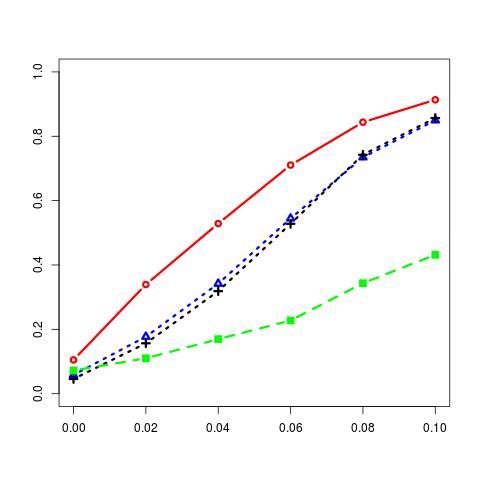}}
\qquad
\subfloat[][noise=heter, $T$=30]{
\includegraphics[width=0.25\textwidth]{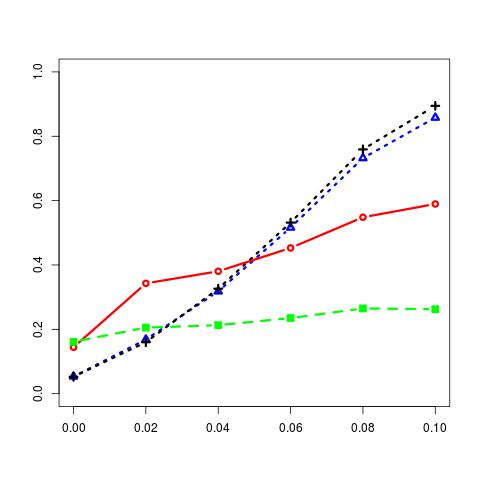}}
\qquad
\subfloat[][noise=$ar_{pos}$, $T$=30]{
\includegraphics[width=0.25\textwidth]{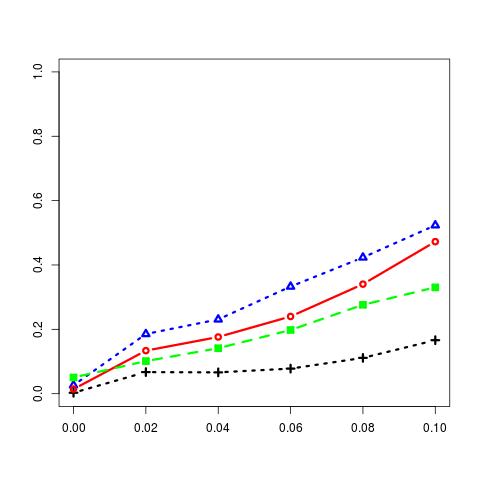}}

\subfloat[][noise=$ma_{neg}$, $T$=30]{
\includegraphics[width=0.25\textwidth]{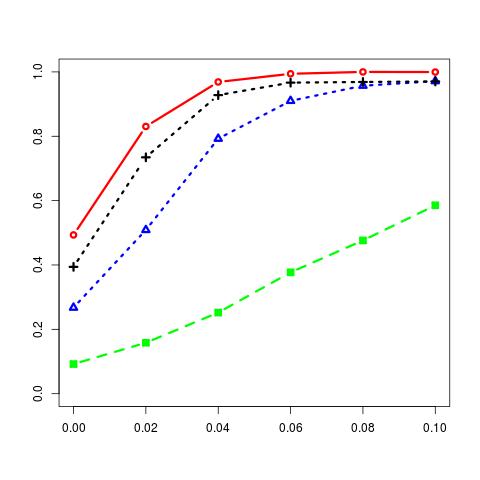}}
\qquad
\subfloat[][noise=$ar_{per}$, $T$=30]{
\includegraphics[width=0.25\textwidth]{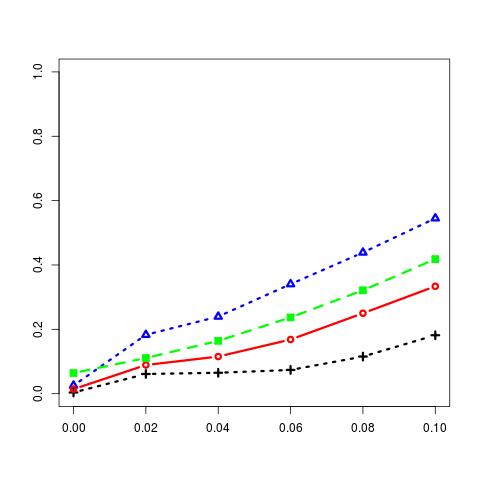}}
\qquad
\subfloat[][noise=$ma_{per}$, $T$=30]{
\includegraphics[width=0.25\textwidth]{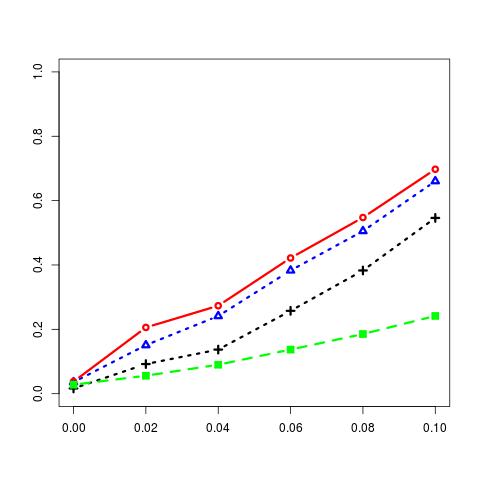}}
\qquad

\subfloat[][noise=iid, $T$=120]{
\includegraphics[width=0.25\textwidth]{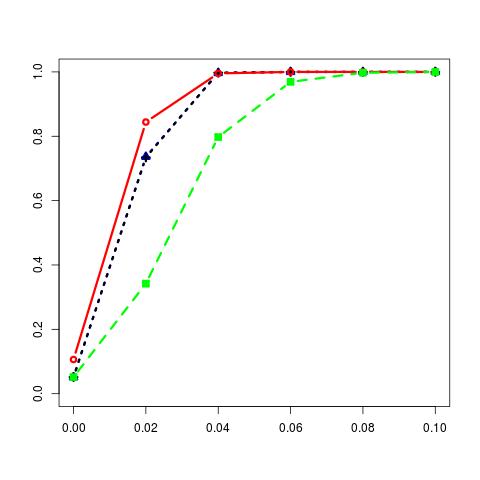}}
\qquad
\subfloat[][noise=heter, $T$=120]{
\includegraphics[width=0.25\textwidth]{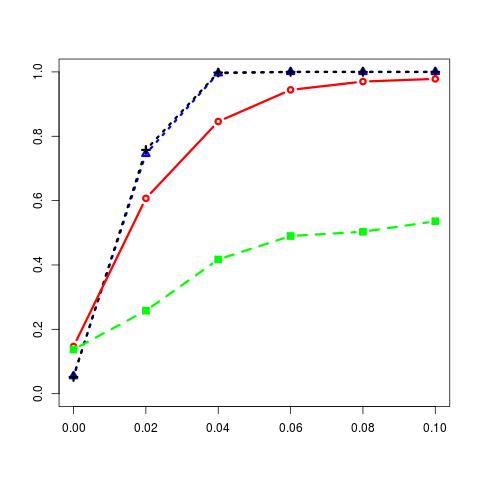}}
\qquad
\subfloat[][noise=$ar_{pos}$, $T$=120]{
\includegraphics[width=0.25\textwidth]{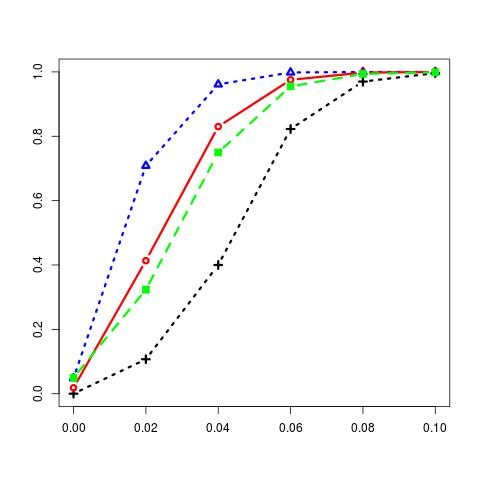}}

\subfloat[][noise=$ma_{neg}$, $T$=120]{
\includegraphics[width=0.25\textwidth]{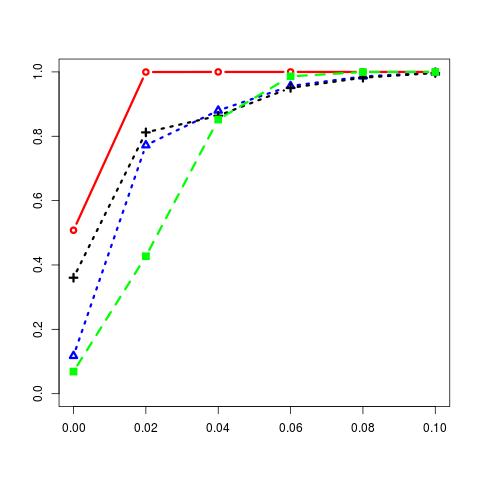}}
\qquad
\subfloat[][noise=$ar_{per}$, $T$=120]{
\includegraphics[width=0.25\textwidth]{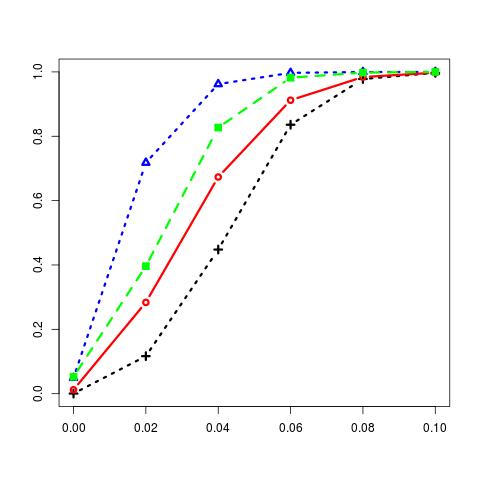}}
\qquad
\subfloat[][noise=$ma_{per}$, $T$=120]{
\includegraphics[width=0.25\textwidth]{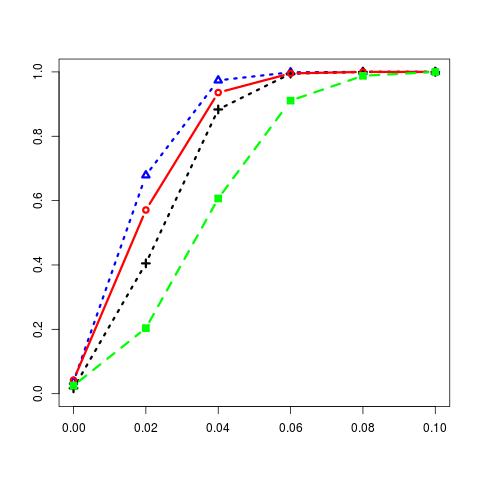}}
\qquad
\caption{Power as a function of $\rho$ when testing roots at 1 with all nuisance roots. Blue dotted curve with triangle knot is for the seasonal iid bootstrap test. Red solid curve with circle knot is for the seasonal block bootstrap test. Black dotted curve with ``+" knot is for the non-seasonal bootstrap test. Green dashed curve with square knot is for the Wald test. In (a)-(f) sample size $T=30$. In (g)-(l) sample size $T=120$.}
\label{fig:root1_T}
\end{figure}

\begin{figure}[H]
\centering

\subfloat[][noise=iid, $T$=30]{
\includegraphics[width=0.25\textwidth]{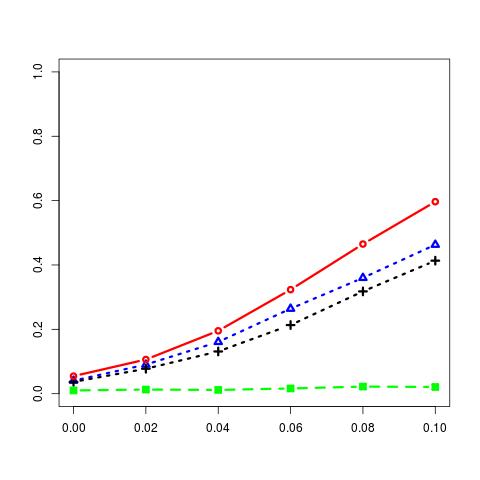}}
\qquad
\subfloat[][noise=heter, $T$=30]{
\includegraphics[width=0.25\textwidth]{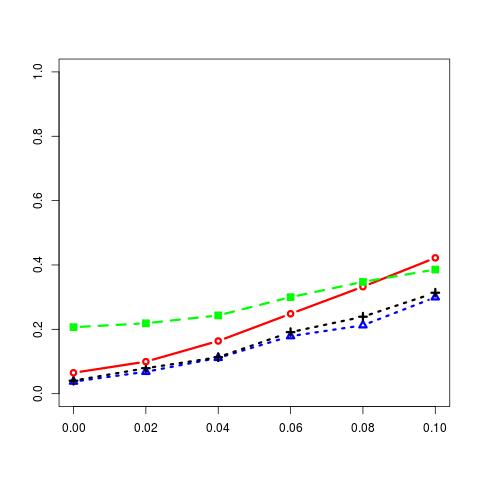}}
\qquad
\subfloat[][noise=$ar_{pos}$, $T$=30]{
\includegraphics[width=0.25\textwidth]{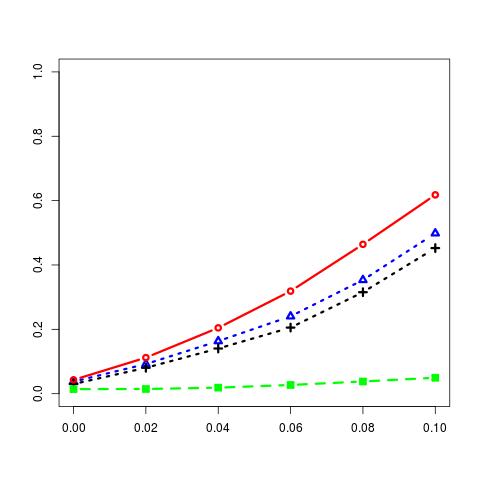}}

\subfloat[][noise=$ma_{neg}$, $T$=30]{
\includegraphics[width=0.25\textwidth]{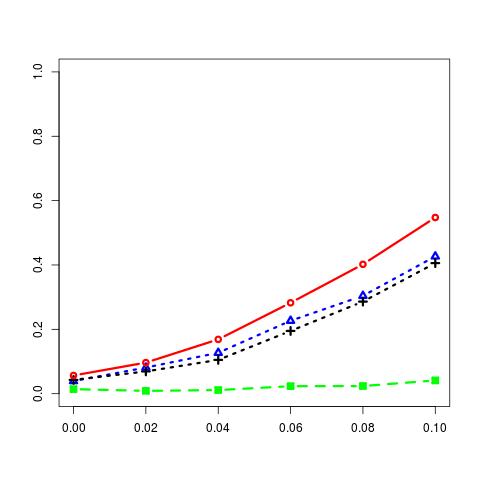}}
\qquad
\subfloat[][noise=$ar_{per}$, $T$=30]{
\includegraphics[width=0.25\textwidth]{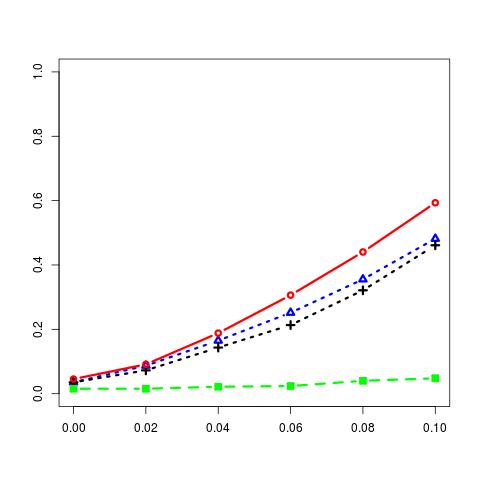}}
\qquad
\subfloat[][noise=$ma_{per}$, $T$=30]{
\includegraphics[width=0.25\textwidth]{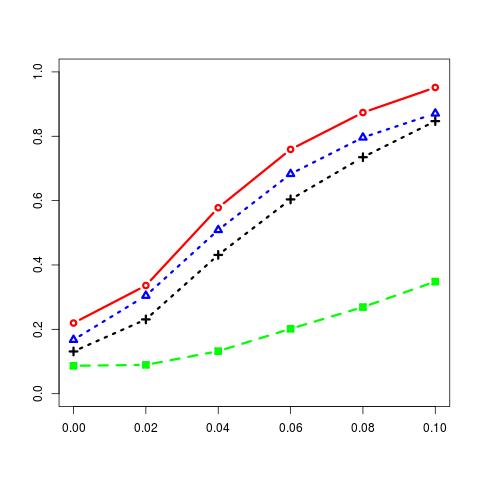}}
\qquad

\subfloat[][noise=iid, $T$=120]{
\includegraphics[width=0.25\textwidth]{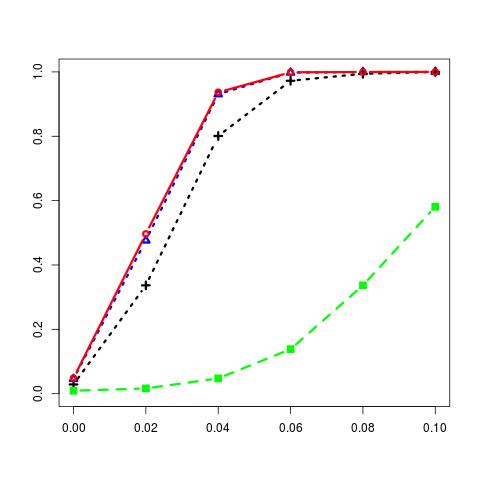}}
\qquad
\subfloat[][noise=heter, $T$=120]{
\includegraphics[width=0.25\textwidth]{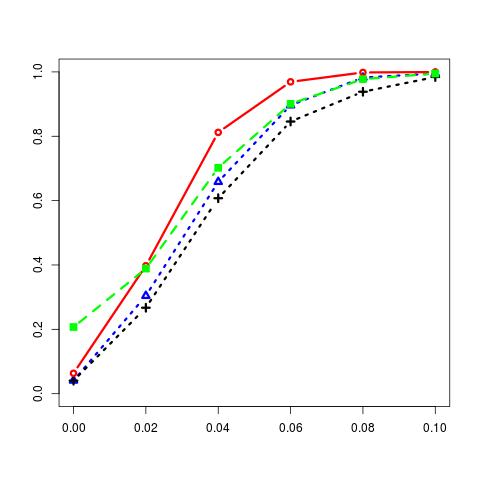}}
\qquad
\subfloat[][noise=$ar_{pos}$, $T$=120]{
\includegraphics[width=0.25\textwidth]{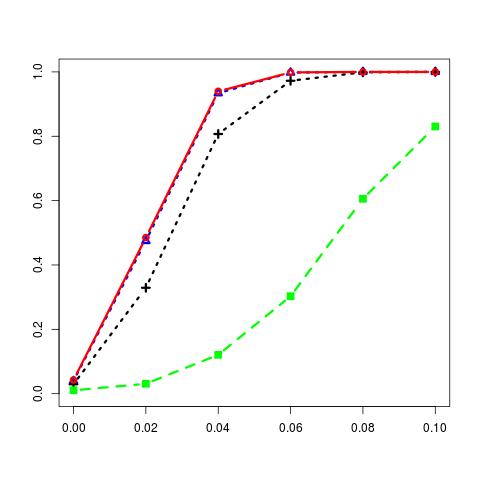}}

\subfloat[][noise=$ma_{neg}$, $T$=120]{
\includegraphics[width=0.25\textwidth]{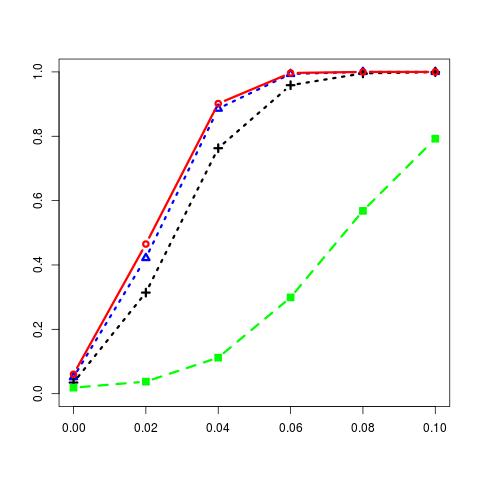}}
\qquad
\subfloat[][noise=$ar_{per}$, $T$=120]{
\includegraphics[width=0.25\textwidth]{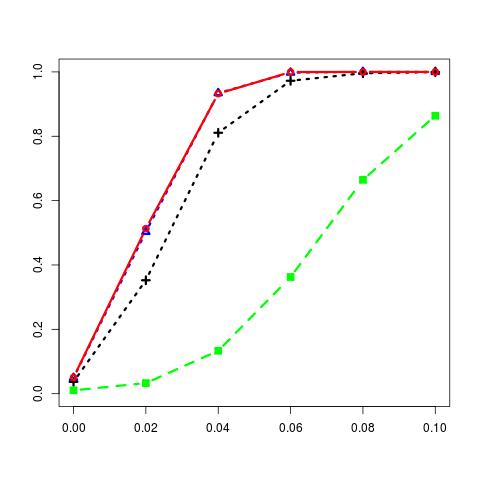}}
\qquad
\subfloat[][noise=$ma_{per}$, $T$=120]{
\includegraphics[width=0.25\textwidth]{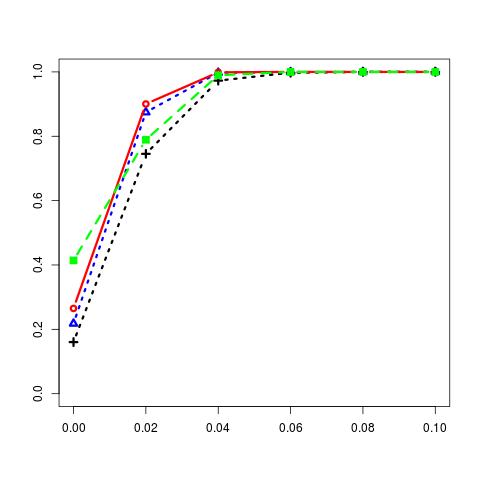}}
\qquad
\caption{Power as a function of $\rho$ when testing roots at $\pm i$ with no nuisance root. Blue dotted curve with triangle knot is for the seasonal iid bootstrap test. Red solid curve with circle knot is for the seasonal block bootstrap test. Black dotted curve with ``+" knot is for the non-seasonal bootstrap test. Green dashed curve with square knot is for the Wald test. In (a)-(f) sample size $T=30$. In (g)-(l) sample size $T=120$.  }
\label{fig:root34_F}
\end{figure}

\begin{figure}[H]
\centering
\subfloat[][noise=iid, $T$=30]{
\includegraphics[width=0.25\textwidth]{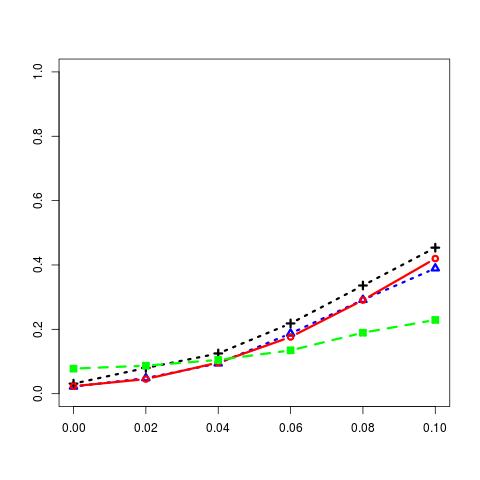}}
\qquad
\subfloat[][noise=heter, $T$=30]{
\includegraphics[width=0.25\textwidth]{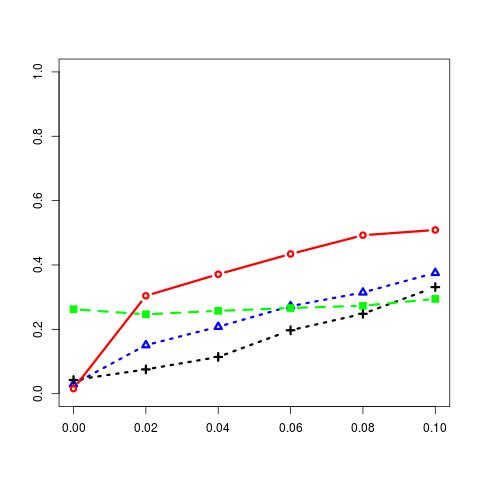}}
\qquad
\subfloat[][noise=$ar_{pos}$, $T$=30]{
\includegraphics[width=0.25\textwidth]{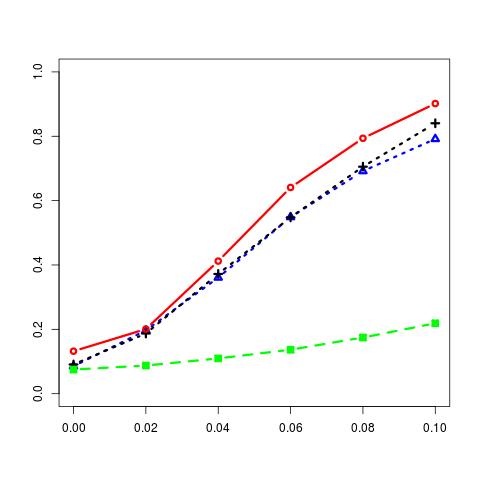}}

\subfloat[][noise=$ma_{neg}$, $T$=30]{
\includegraphics[width=0.25\textwidth]{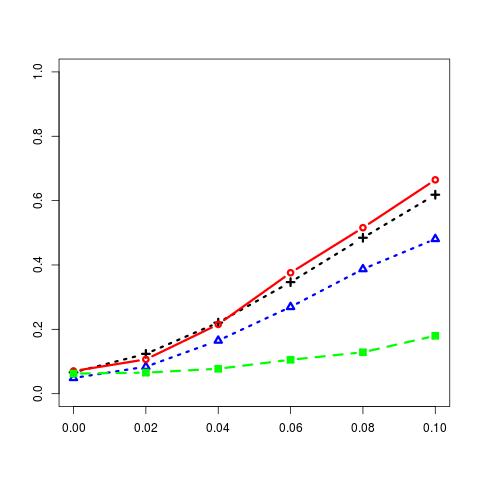}}
\qquad
\subfloat[][noise=$ar_{per}$, $T$=30]{
\includegraphics[width=0.25\textwidth]{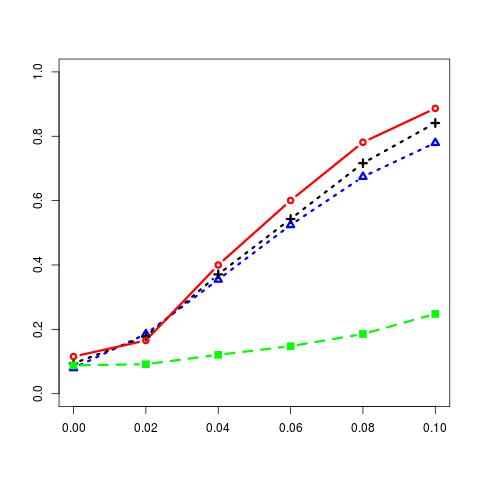}}
\qquad
\subfloat[][noise=$ma_{per}$, $T$=30]{
\includegraphics[width=0.25\textwidth]{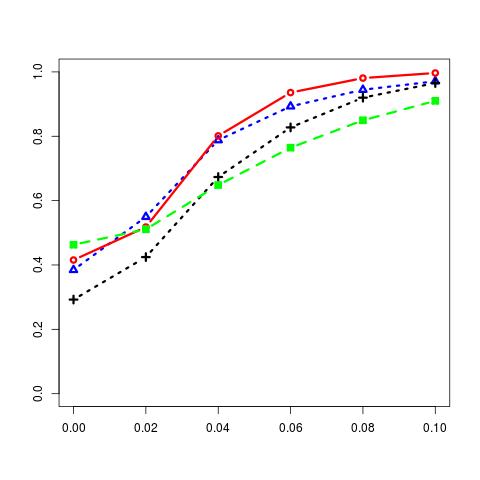}}
\qquad

\subfloat[][noise=iid, $T$=120]{
\includegraphics[width=0.25\textwidth]{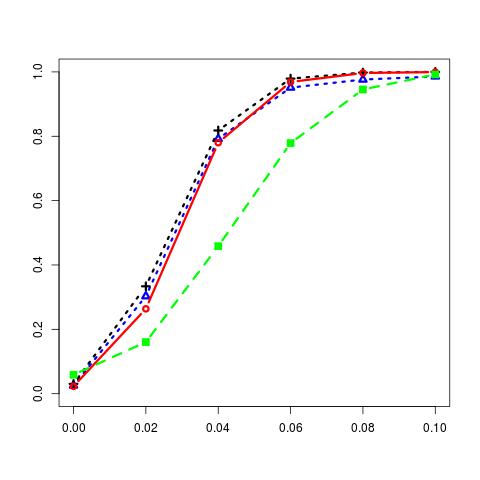}}
\qquad
\subfloat[][noise=heter, $T$=120]{
\includegraphics[width=0.25\textwidth]{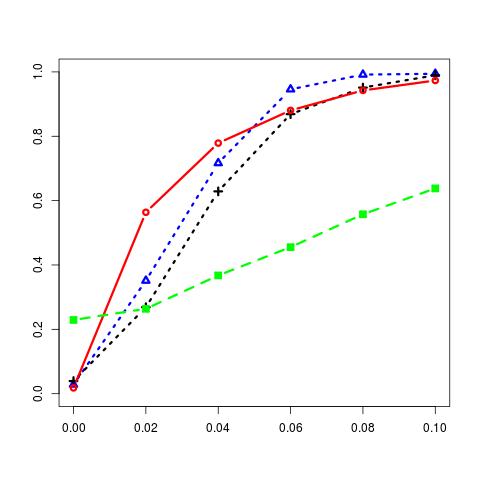}}
\qquad
\subfloat[][noise=$ar_{pos}$, $T$=120]{
\includegraphics[width=0.25\textwidth]{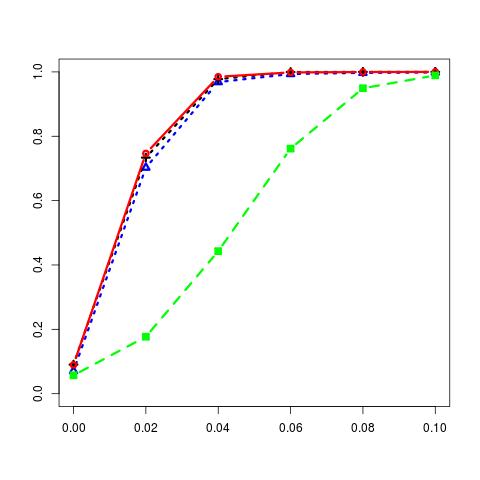}}

\subfloat[][noise=$ma_{neg}$, $T$=120]{
\includegraphics[width=0.25\textwidth]{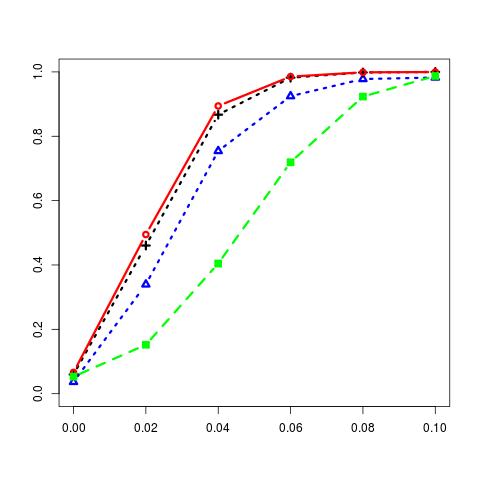}}
\qquad
\subfloat[][noise=$ar_{per}$, $T$=120]{
\includegraphics[width=0.25\textwidth]{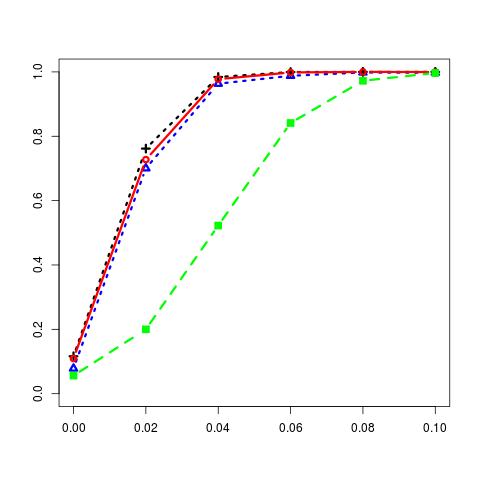}}
\qquad
\subfloat[][noise=$ma_{per}$, $T$=120]{
\includegraphics[width=0.25\textwidth]{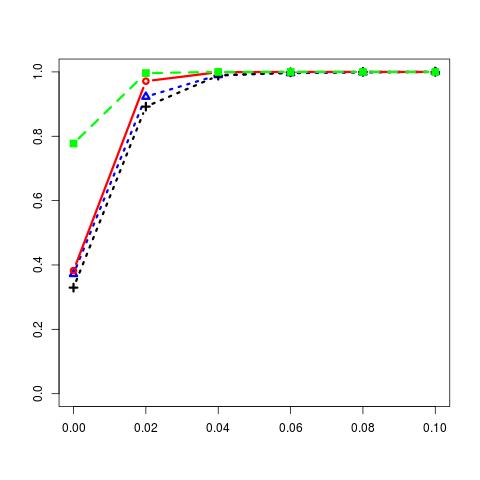}}
\qquad
\caption{Power as a function of $\rho$ when testing roots at $\pm i$ with all nuisance roots. Blue dotted curve with triangle knot is for the seasonal iid bootstrap test. Red solid curve with circle knot is for the seasonal block bootstrap test. Black dotted curve with ``+" knot is for the non-seasonal bootstrap test. Green dashed curve with square knot is for the Wald test. In (a)-(f) sample size $T=30$. In (g)-(l) sample size $T=120$.  }
\label{fig:root34_T}
\end{figure}

\begin{figure}[H]
\centering

\subfloat[][noise=iid, $T$=30]{
\includegraphics[width=0.25\textwidth]{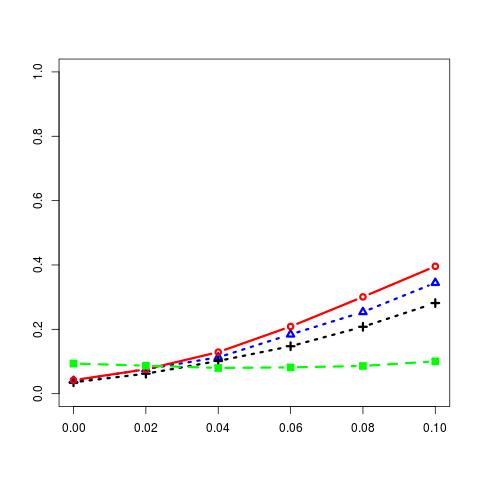}}
\qquad
\subfloat[][noise=heter, $T$=30]{
\includegraphics[width=0.25\textwidth]{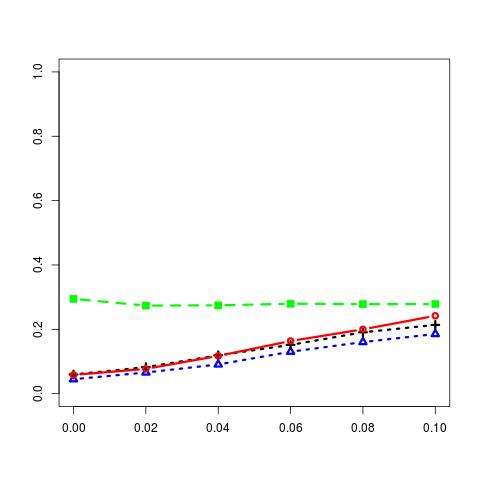}}
\qquad
\subfloat[][noise=$ar_{pos}$, $T$=30]{
\includegraphics[width=0.25\textwidth]{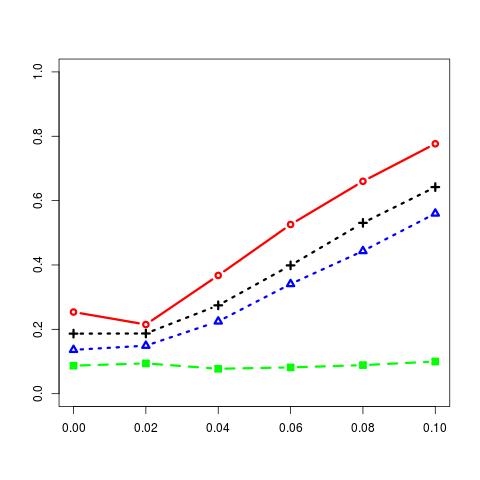}}

\subfloat[][noise=$ma_{neg}$, $T$=30]{
\includegraphics[width=0.25\textwidth]{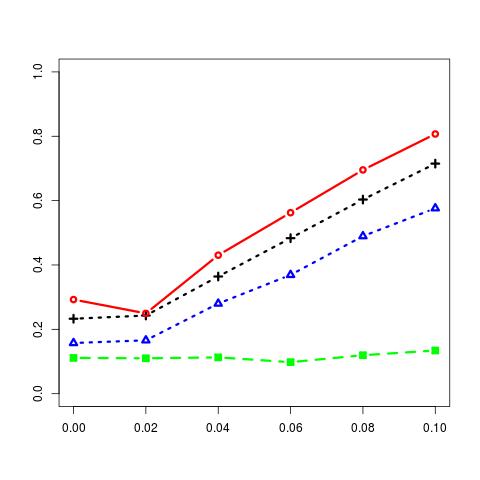}}
\qquad
\subfloat[][noise=$ar_{per}$, $T$=30]{
\includegraphics[width=0.25\textwidth]{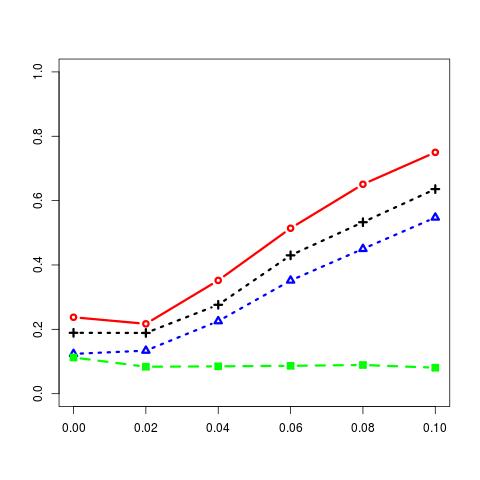}}
\qquad
\subfloat[][noise=$ma_{per}$, $T$=30]{
\includegraphics[width=0.25\textwidth]{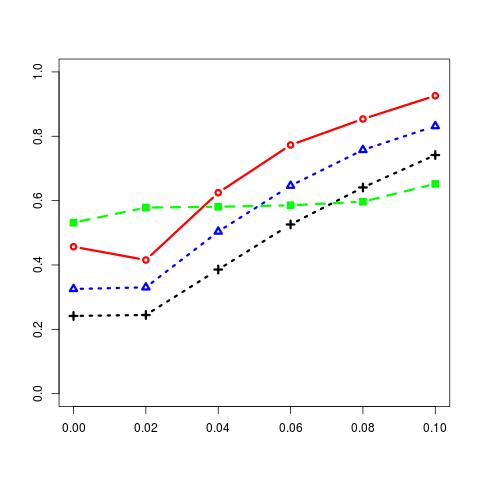}}
\qquad

\subfloat[][noise=iid, $T$=120]{
\includegraphics[width=0.25\textwidth]{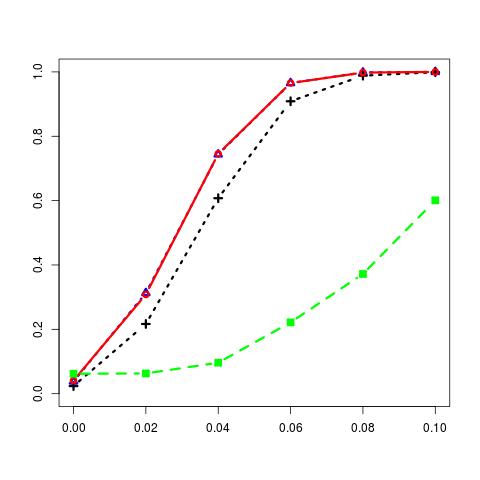}}
\qquad
\subfloat[][noise=heter, $T$=120]{
\includegraphics[width=0.25\textwidth]{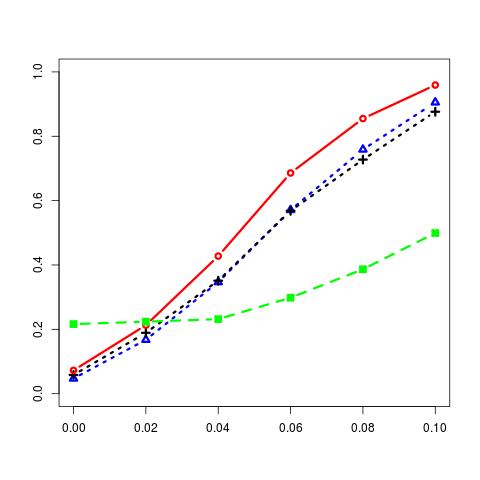}}
\qquad
\subfloat[][noise=$ar_{pos}$, $T$=120]{
\includegraphics[width=0.25\textwidth]{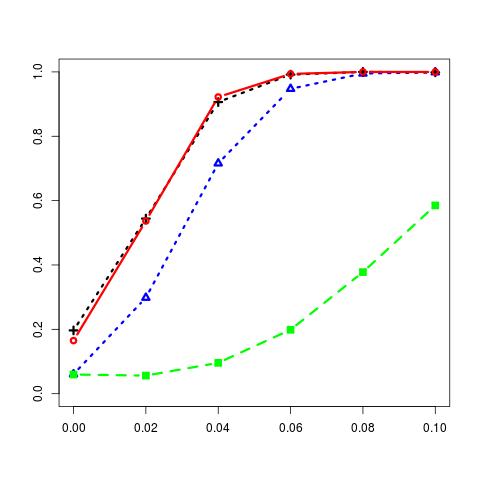}}

\subfloat[][noise=$ma_{neg}$, $T$=120]{
\includegraphics[width=0.25\textwidth]{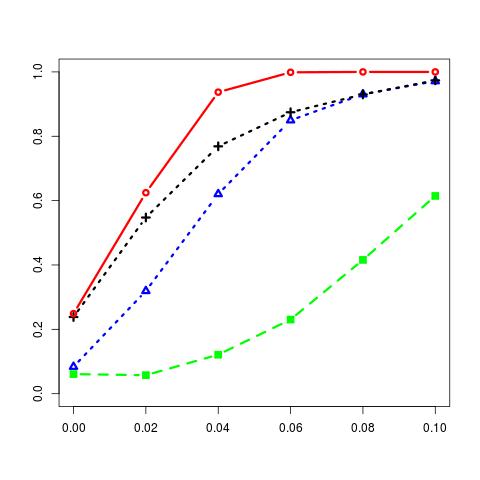}}
\qquad
\subfloat[][noise=$ar_{per}$, $T$=120]{
\includegraphics[width=0.25\textwidth]{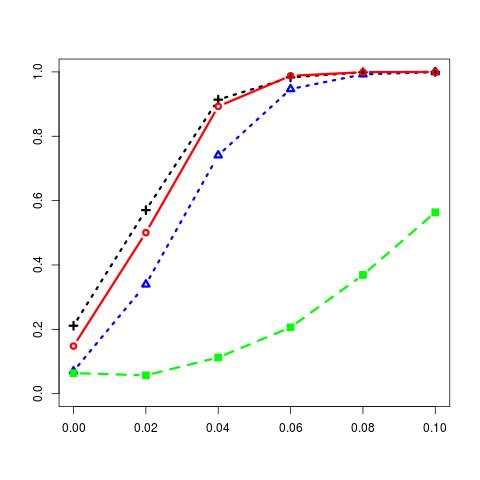}}
\qquad
\subfloat[][noise=$ma_{per}$, $T$=120]{
\includegraphics[width=0.25\textwidth]{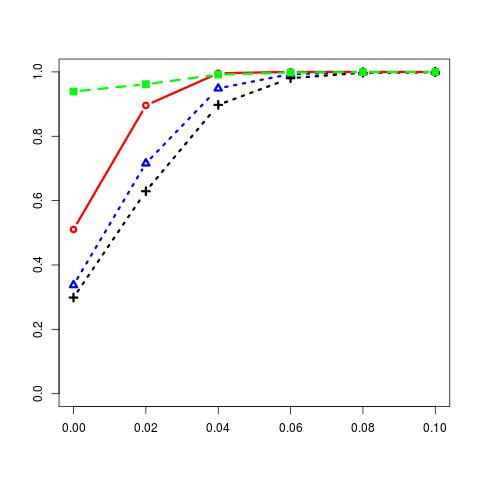}}
\qquad
\caption{Power as a function of $\rho$ when testing roots at $1$, $-1$, and $\pm i$. Blue dotted curve with triangle knot is for the seasonal iid bootstrap test. Red solid curve with circle knot is for the seasonal block bootstrap test. Black dotted curve with ``+" knot is for the non-seasonal bootstrap test. Green dashed curve with square knot is for the Wald test. In (a)-(f) sample size $T=30$. In (g)-(l) sample size $T=120$.  }
\label{fig:root1234_F}
\end{figure}

\section{Real Data Application of Seasonal Unit Root Test}

\subsection{Datasets}
Here we present the result of the seasonal unit root tests on four quarterly economic time series that have not been seasonally adjusted. The first dataset contains gas consumption in millions of therms in United Kingdom from quarter one, 1960 to quarter four, 1986. The second dataset gives the E-commerce retail sales as a percent of total sales in United States from quarter four, 1999 to quarter three, 2016. The third dataset presents the owned and securitized outstanding student loans in billions of dollars in United States from quarter one, 2006 to quarter four, 2016. The fourth includes \change{the logarithms of} the earnings per Johnson\&Johnson share in dollars from quarter one, 1960 to quarter four, 1980. The deterministic linear and quadratic trends and the deterministic seasonal component of these time series are first estimated with OLS and then removed from the data. The detrended and deseasonalized time series are presented in Figure \ref{fig:real_data}. Since \cite{osborn1988seasonality,osborn1989performance} have indicated possible periodic structure in economic time series, when investigating the stochastic seasonality of these time series we include tests catering to periodicity. Specifically, we implement the seasonal iid bootstrap augmented HEGY test (SIB), the seasonal block bootstrap unaugmented HEGY test (SBB), \change{the non-seasonal bootstrap augmented HEGY test (NSB) by \cite{burridge2004bootstrapping}, }and the Wald test (WALD) by \cite{ghysels1996periodic}.

\begin{figure}[H]
\centering
\subfloat[Subfigure 1 list of figures text][Gas]{
\includegraphics[width=0.20\textwidth]{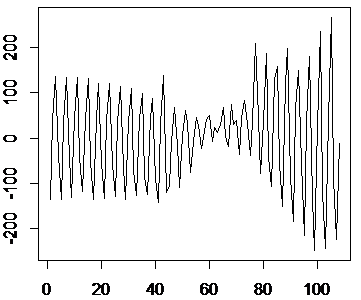}}
\qquad
\subfloat[Subfigure 2 list of figures text][E-Commerce]{
\includegraphics[width=0.20\textwidth]{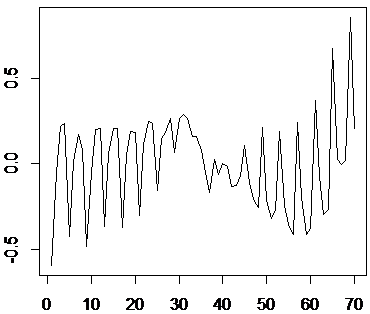}}
\qquad
\subfloat[Subfigure 3 list of figures text][Student Loan]{
\includegraphics[width=0.20\textwidth]{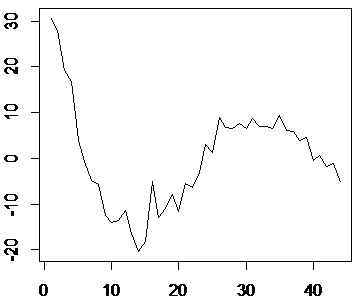}}
\qquad
\subfloat[Subfigure 4 list of figures text][Johnson\&Johnson]{
\includegraphics[width=0.20\textwidth]{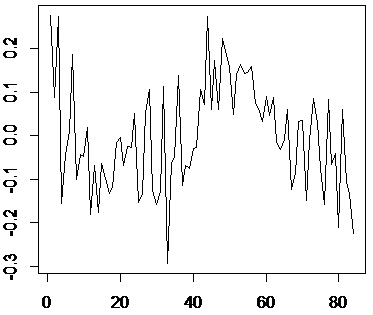}}
\caption{Quarterly time series with deterministic trend and seasonal component removed}
\label{fig:real_data}
\end{figure}

\subsection{Results}

\begin{table}[H]
\centering
\caption{\change{P-values of seasonal unit root tests on economic data}}
\label{table:real data}
\begin{tabular}{ccccccccc}
\hline
     & \multicolumn{4}{c}{Gas}                                                                                                                       & \multicolumn{4}{c}{E-Commerce}                                                                                                             \\ \cline {2-9} 
     & $H_{0}^{1}$ & $H_{0}^{2}$ & $H_{0}^{3,4}$ & $H_{0}^{1,2,3,4}$ & $H_{0}^{1}$ & $H_{0}^{2}$ & $H_{0}^{3,4}$ & $H_{0}^{1,2,3,4}$ \\ \hline  
SIB  & 0.068                           & 0.000                           & 0.944                             & 0.020                                 & 0.322                           & 0.162                           & 0.290                             & 0.540                                 \\
SBB  & 0.038                           & 0.000                           & 0.876                             & 0.026                                 & 0.544                           & 0.014                           & 0.252                             & 0.306                                 \\
NSB & 0.042                           & 0.000                           & 0.988                             & 0.208                                 & 0.476                           & 0.192                           & 0.508                             & 0.668                                 \\
WALD & 0.719                           & 0.013                           & 0.440                             & 0.108                                 & 0.438                           & 0.967                           & 0.473                             & 0.027                                 \\ \hline  
     & \multicolumn{4}{c}{Student Loan}                                                                                                              & \multicolumn{4}{c}{Johnson\&Johnson}                                                                                                          \\ \cline {2-9}   
     & $H_{0}^{1}$ & $H_{0}^{2}$ & $H_{0}^{3,4}$ & $H_{0}^{1,2,3,4}$ & $H_{0}^{1}$ & $H_{0}^{2}$ & $H_{0}^{3,4}$ & $H_{0}^{1,2,3,4}$ \\ \hline  
SIB  & 0.136                           & 0.002                           & 0.000                             & 0.000                                 & 0.226                           & 0.012                           & 0.002                             & 0.000                                 \\
SBB  & 0.128                           & 0.000                           & 0.000                             & 0.000                                 & 0.092                           & 0.000                           & 0.002                             & 0.000                                 \\
NSB & 0.110                           & 0.000                           & 0.000                             & 0.000                                 & 0.286                           & 0.036                           & 0.002                             & 0.006                                 \\
WALD & 0.362                           & 0.204                           & 0.979                             & 0.611                                 & 0.513                           & 0.028                           & 0.011                             & 0.002    \\ \hline                            
\end{tabular}
\end{table}

\subsubsection{Gas Consumption}
First, we investigate the p-values of the gas consumption time series. \change{Overall, from Table \ref{table:real data}, we observe that the seasonal iid bootstrap test, which is recommended for testing the concurrence of roots at $1$, $-1$, and $\pm i$, rejects at the size of 5\% the hypothesis that the gas consumption series has all roots at $1$, $-1$, and $\pm i$.} In a root-by-root analysis, we found that at the size of 5\%, all the tests unanimously reject root at $-1$, but on the other hand none of the tests rejects root at $\pm i$. Hence, the gas consumption time series may possess roots at $\pm i$. \change{Notice that in the gas consumption time series, the sample size $T=27$ is fairly small. In the test of root at 1, if the time series possesses nuisance roots and the sample size is small, the Wald test loses power; see Figure \ref{fig:root1_T}. Hence, when testing root at 1, we consider the high p-value from the Wald test unreliable. Since, in testing the root at 1, the p-values of all tests other than the Wald test are around 5\%, at the size of 5\% we conclude that the gas consumption process may have roots at $\pm i$, may not have a root at $-1$, and may or may not have a root at $1$.} 
\subsubsection{E-Commerce Sales}
\change{At a size of 5\%, none of the tests, except for the Wald test, can reject the hypothesis that the e-commerce sales series has all roots at $1$, $-1$, and $\pm i$. Notice that in the e-commerce sales context, the sample size $T=18$ is fairly small. When testing jointly the roots at $1$, $-1$, and $\pm i$ and when the sample size is small, the Wald test suffers severe upward size distortions, see Figure \ref{fig:root1234_F}. Hence, we ignore the small p-value of the Wald test when testing jointly the roots at $1$, $-1$, and $\pm i$ and conclude that the e-commerce sales series may simultaneously have roots at $1$, $-1$, and $\pm i$. This conclusion is consistent with the high p-values of the root-by-root tests.}
\subsubsection{Student Loans}
\change{When analyzing the student loans series, we focus on the p-values of the three bootstraps tests, whose superiority over the Wald test has been illustrated in the simulation. We observe that all the bootstrap tests reject at the size of 5\% the hypothesis that the student loans series has all roots at $1$, $-1$, and $\pm i$. In a root-by-root analysis, all the bootstrap tests unanimously reject the root at $-1$ and $\pm i$, but fail to reject the root at $1$. Hence, we conclude that the student loans series may have a root at $1$, but may not have roots at $-1$ or $\pm i$.} 
\subsubsection{Johnson\&Johnson Earnings}
\change{According to Table \ref{table:real data}, all of the tests reject, at the size of 5\%, the hypothesis that the Johnson\&Johnson earnings series has all roots at $1$, $-1$, and $\pm i$. In a root-by-root analysis, all of the tests reject the roots at $-1$ and $\pm i$, but fail to reject the root at $1$. Hence, we conclude that the Johnson\&Johnson earnings series may have a root at $1$, but may not have roots at $-1$ or $\pm i$.}

\section{Conclusion}
In this paper we analyze the augmented and the unaugmented HEGY tests in the periodically varying setting. For root at 1 or $-1$, the asymptotic distributions of the testing statistics are standard. However, for any combinations of roots at 1, $-1$, $i$, and $-i$, the asymptotic distributions are not standard, not pivotal, and cannot be easily pivoted. Therefore, when periodic variation exists, the HEGY test can be applied to test any single real roots, but cannot be directly applied to any combinations of roots. 

Bootstrap proves to be an effective remedy for the HEGY test in the periodically varying setting. The two bootstrap approaches, namely 1) the seasonal iid bootstrap augmented HEGY test and 2) the seasonal block bootstrap unaugmented HEGY test, turn out to be theoretically solid. In the simulation study, \change{we compare these two bootstrap tests with the non-seasonal bootstrap augmented HEGY test by \cite{burridge2004bootstrapping} and the Wald test by \cite{ghysels1996periodic}.} It turns out that the seasonal iid bootstrap augmented HEGY test has the best performance when we test root at 1, $-1$ and \change{when we test the concurrence of roots at $1$, $-1$, and $\pm i$}; on the other hand, the seasonal block bootstrap unaugmented HEGY test prevails when we test roots at $\pm i$. Real data application shows the importance of our bootstrap approaches in constructing powerful tests.
\section*{Acknowledgment}
\change{We are grateful to the Associate Editor and two anonymous referees for their insightful feedback.}
\bibliographystyle{elsarticle-num}
\bibliography{GSBB}

\newpage
\begin{center}

{\bfseries Appendix to: ``Bootstrap seasonal unit root test under periodic variation"}
\vspace{.5cm}

{\textsc{By Nan Zou and Dimitris N. Politis}}

\vspace{.28cm}

{\textit{University of Toronto and University of California-San Diego}}

\vspace{.28cm}

\begin{center}
\begin{minipage}{1\textwidth}
{\small \hspace{.5cm}The appendix includes the proofs of the theorems in the main manuscript. We first present the proof for the asymptotics of the unaugmented HEGY test, then the asymptotics of the augmented HEGY test, then the consistency of the seasonal iid bootstrap augmented HEGY test, and finally the consistency of the seasonal block bootstrap unaugmented HEGY test. Thoughout the appendix, let \linebreak $\bY_{t}=(Y_{4t-3},Y_{4t-2},Y_{4t-1},Y_{4t})'$,  $\bm{\Gamma}_{j}=E[\bm{V}_{t}\bm{V}_{t-j}']$, $\int \bW d \bW'$ denotes $\int_{0}^{1} \bW(r) d \bW(r)'$, and  $\int\bW\bW'$ denotes
$\int_{0}^{1}\bW(r)\bW(r)'dr$.}
\end{minipage}
\end{center}
\end{center}

\appendix

\section{Proof of Theorem \ref{unaug real}.}
\begin{Lemma}\label{le:unaug real1}
Suppose one of Assumption \ref{assump 1a} and Assumption \ref{assump 1b} and one of Assumption \ref{assump 2a} and Assumption \ref{assump 2b} hold. Then under $H_{0}^{1,2,3,4}$,
\begin{align*}
&T^{-1}\sum_{t=1}^{T}\bY_{t-1}\bm{V}_{t}'\Rightarrow \bTheta(1)\bOmega^{1/2}\{\int \bW d \bW'\}\bOmega^{1/2} \bTheta(1)'+\sum_{j=1}^{\infty}\bm{\Gamma}_{j}'\equiv\bm{Q}_{1},\\
&T^{-2}\sum_{t=1}^{T}\bY_{t-1}\bY_{t-1}'\Rightarrow \bTheta(1)\bOmega^{1/2} \{\int\bW\bW'\}\bOmega^{1/2}\bTheta(1)'\equiv\bm{Q}_{2},\\
&T^{-1}\sum_{t=1}^{T}\bm{V}_{t}\bm{V}_{t-j}'\stackrel {p}\rightarrow \bm{\Gamma}_{j}.
\end{align*}
\end{Lemma}
\begin{proof}
See \cite{hamilton1994time} (Proposition 18.1, pp. 547-548) for the proof with iid innovations, \cite{chan1988limiting} for the proof under Assumption \ref{assump 2a}, and \cite{de2000} for the proof under Assumption \ref{assump 2b}. 
\end{proof}

\begin{Lemma} \label{le:unaug real2}
Let $\bX_{U,j}=(Y_{j,0},\dots,Y_{j,4T-1})' $, $\bX_{U}=(\bX_{U,1},\bX_{U,2},\bX_{U,3},\bX_{U,4})$, where $U$ stands for unaugmented HEGY, and $\{Y_{j,4t+s}\}$ be defined by \eqref{PHEGY2}. Let $\bV=(V_{1},\dots,V_{4T})'$ and $\bUp$ be a $4\times 4$ matrix such that
$$
\bUp_{ij}=
\begin{cases}
(\bGamma_{0})_{ij}&\ \ \text{if} \ i<j,\\
0 &\ \ \text{if} \ i\geq j.
\end{cases}
$$
Then, under $H_{0}^{1,2,3,4}$,
\begin{align*}
(a)\\
&(4T)^{-2}(\bX_{U}'\bX_{U})_{11}\Rightarrow\frac{1}{4}\bm{c}_{1}'\bm{Q}_{2}\bm{c}_{1}\equiv\eta_{1},\\
&(4T)^{-2}(\bX_{U}'\bX_{U})_{22}\Rightarrow\frac{1}{4}\bm{c}_{2}'\bm{Q}_{2}\bm{c}_{2}\equiv\eta_{2},\\
&(4T)^{-2}(\bX_{U}'\bX_{U})_{33}\Rightarrow\frac{1}{8}(\bm{c}_{3}'\bm{Q}_{2}\bm{c}_{3}+\bm{c}_{4}'\bm{Q}_{2}\bm{c}_{4})\equiv\eta_{3},\\
&(4T)^{-2}(\bX_{U}'\bX_{U})_{44}\Rightarrow\frac{1}{8}(\bm{c}_{3}'\bm{Q}_{2}\bm{c}_{3}+\bm{c}_{4}'\bm{Q}_{2}\bm{c}_{4})\equiv\eta_{3},\\
&(4T)^{-1}(\bX_{U}'\bX_{U})_{ij} \stackrel{p}\rightarrow0, \text{ for } i\neq j.\\
\\
(b)\\
&(4T)^{-1}\bX_{U,1}'\bm{V}\Rightarrow\frac{1}{4}(\bm{c}_{1}'\bm{Q}_{1}\bm{c}_{1}+\bm{c}_{1}'\bUp\bm{c}_{1})\equiv\xi_{1},\\
&(4T)^{-1}\bX_{U,2}'\bm{V}\Rightarrow\frac{1}{4}(\bm{c}_{2}'\bm{Q}_{1}\bm{c}_{2}+\bm{c}_{2}'\bUp\bm{c}_{2})\equiv\xi_{2},\\
&(4T)^{-1}\bX_{U,3}'\bm{V}\Rightarrow\frac{1}{4}(\bm{c}_{3}'\bm{Q}_{1}\bm{c}_{3}+\bm{c}_{4}'\bm{Q}_{1}\bm{c}_{4}+\bm{c}_{3}'\bUp\bm{c}_{3}+\bm{c}_{4}'\bUp\bm{c}_{4})\equiv\xi_{3},\\
&(4T)^{-1}\bX_{U,4}'\bm{V}\Rightarrow\frac{1}{4}(\bm{c}_{3}'\bm{Q}_{1}\bm{c}_{4}-\bm{c}_{4}'\bm{Q}_{1}\bm{c}_{3}+\bm{c}_{3}'\bUp\bm{c}_{4}-\bm{c}_{4}'\bUp\bm{c}_{3})\equiv\xi_{4}.
\end{align*}
\end{Lemma}
\begin{proof}
For the proof of part (a), see the Lemma 3.2(a) of \cite{burridge2001} and its proof. For part (b), we only present the proof of the first statement. Other statements are proven in similar ways. By Lemma \ref{le:unaug real1}, 
\begin{align*}
(4T)^{-1}\bX_{U,1}'\bm{V}&=(4T)^{-1}\sum_{t=1}^{T}\sum_{s=-3}^{0}Y_{1,4t+s-1}V_{4t+s}\\ &=(4T)^{-1}\sum_{t=1}^{T}\sum_{s=-3}^{0}(\bm{c}_{1}'\bY_{t-1}+\sum_{i=-2}^{s}V_{4t-1+i})V_{4t+s}\\
&=(4T)^{-1}\sum_{t=1}^{T}\bm{c}_{1}'\bY_{t-1}\bm{V}_{t}'\bm{c}_{1}+(4T)^{-1}\sum_{t=1}^{T}\sum_{s=-3}^{0}\sum_{i=-2}^{s}V_{4t-1+i}V_{4t+s}\\
&\Rightarrow\frac{1}{4}(\bm{c}_{1}'\bm{Q}_{1}\bm{c}_{1}+\bm{c}_{1}'\bUp\bm{c}_{1}).&\mbox{\qedhere}
\end{align*} 
\end{proof}

\begin{proof}[Proof of Theorem \ref{unaug real}]
Under $H_{0}^{1,2,3,4}$, we have $(1-L^{4})Y_{\tau}=V_{\tau}$. Let $\bm{\hat{\pi}}=(\hat{\pi}_{1}^{U},\hat{\pi}_{2}^{U},\hat{\pi}_{3}^{U},\hat{\pi}_{4}^{U})'$, $\bm{t}=(t_{1}^{U},t_{2}^{U},t_{3}^{U},t_{4}^{U})'$, and $\hat{\sigma}^{2}=(4T)^{-1}(\bm{V}-\bX_{U}\bm{\hat{\pi}})'(\bm{V}-\bX_{U}\bm{\hat{\pi}})$. Then by Lemma \ref{le:unaug real2},
\begin{align*}
(4T)\bm{\hat{\pi}}
&=(\bX_{U}'\bX_{U})^{-1}\bX_{U}'\bm{V}\Rightarrow[\text{diag}(\eta_{1},\eta_{2},\eta_{3},\eta_{4})]^{-1}(\xi_{1},\xi_{2},\xi_{3},\xi_{4})'.
\end{align*}
Hence, $\bm{\hat{\pi}}=o_{p}(1)$. By Lemma \ref{le:unaug real1} and \ref{le:unaug real2},
\begin{align*}
\hat{\sigma}^{2}
=(4T)^{-1}\bm{V}'\bm{V}+o_{p}(1)= \text{tr}(\bm{\Gamma}_{0})/4+o_{p}(1).
\end{align*}
Hence,
\begin{align*}
\bm{t}&=\hat{\sigma}^{-1}[\text{diag}(\bX_{U}'\bX_{U})^{-1}]^{-1/2}(\bX_{U}'\bX_{U})^{-1}\bX_{U}'\bm{V}\\
&\Rightarrow (\text{tr}(\bm{\Gamma}_{0})/4)^{-1/2}[\text{diag}(\eta_{1},\eta_{2},\eta_{3},\eta_{4})]^{-1/2}(\xi_{1},\xi_{2},\xi_{3},\xi_{4})'.
\end{align*}
Further, Lemma \ref{le:unaug real2} (a) indicates an asymptotic orthogonality in the design matrix. Hence,  asymptotically the F-statistics equal the averages of the squares of the corresponding t-statsitics, e.g., $F_{3,4}^{U}-\frac{1}{2}((t_{3}^{U})^{2}+(t_{4}^{U})^{2})
\stackrel{p}\rightarrow 0$.  
\end{proof}

\section{Proof of Theorem \ref{aug real}.}

The proof follows the lines of \cite{said1984testing} and contains two parts. Firstly, we show when $T \rightarrow \infty$ and $k=k_{T} \rightarrow \infty$ simultaneously, the statistic of interest approximates a quantity free of $k$, and then we prove this quantity tends to a certain distribution as $T \rightarrow \infty$.

To begin with, notice that when $k\rightarrow \infty$, the error term of regression \eqref{aug HEGY} tends to a limit. Surprisingly, this limit is in general not $\ep_{\tau}$, because the regression \eqref{aug HEGY} falsely assumes non-periodic coefficients and thus in general cannot find the correct residuals $\ep_{\tau}$. To find the limit, recall that $\{\tV_{\tau}\}$ is defined as a misspecified constant parameter representation of $\{V_{4t+s}\}$. Under Assumption \ref{assump 1b}, the spectral densities of $\{\tV_{\tau}\}$ are finite and positive everywhere, so $\{\tV_{\tau}\}$ has AR$(\infty)$ and MA$(\infty)$ expressions 
\begin{equation}
\tpsi(L)\tV_{\tau}=\tzeta_{\tau} \ \text{and} \ \tV_{\tau}=\ttheta(L)\tzeta_{\tau},    
\label{misspec ARMA}
\end{equation}
where $\tpsi(z)=1-\sum_{i=1}^{\infty}\tpsi_{i}z^{i}$, $\ttheta(z)=1+\sum_{i=1}^{\infty}\ttheta_{i}z^{i}$. 
Let 
\begin{equation}\label{zeta_0}
\zeta_{\tau}^{(k)}=V_{\tau}-\sum_{i=1}^{k}\tpsi_{i}V_{\tau-i},    
\end{equation}
and $\zeta_{\tau}=V_{\tau}-\sum_{i=1}^{\infty}\tpsi_{i}V_{\tau-i}$. Since a misspecified constant parameter representation of $\zeta_{\tau}$ is $\tV_{\tau}-\sum_{i=1}^{\infty}\tpsi_{i}\tV_{\tau-i}$, which is exactly $\tzeta_{\tau}$ defined in \eqref{misspec ARMA}, no ambiguity arises. We can straightforwardly show that 
\begin{equation}
\frac{1}{4}\sum_{s=-3}^{0}Cov(\zeta_{4t+s-i},\zeta_{4t+s})=0, \ \forall i=1,2,\dots, \label{zeta_1}
\end{equation}
and
\begin{equation}
\frac{1}{4}\sum_{s=-3}^{0}Cov(V_{4t+s-i},\zeta_{4t+s})=0, \ \forall i=1,2,\dots. \label{zeta_2}
\end{equation}
Now we show when $T\rightarrow \infty$ and $k\rightarrow\infty$ simultaneously, the statistics of interest approximates quantities free of $k$. Let $\bm{X}$ be the design matrix of regression equation \eqref{aug HEGY},  $\bm{\hat{\beta}}=(\hat{\pi}_{1}^{A},\hat{\pi}_{2}^{A},\hat{\pi}_{3}^{A},\hat{\pi}_{4}^{A},\hat{\phi}_{1},\dots,\hat{\phi}_{k})'$ be the estimated coefficient vector of regression equation \eqref{aug HEGY}, and $\bm{\beta}=(0,0,0,0,\tpsi_{1},\dots,\tpsi_{k})'$, where $\tpsi_{i}$ is defined in \eqref{misspec ARMA}. Let $\bm{\zeta}^{(k)}=(\zeta^{(k)}_{1+k},\dots,\zeta^{(k)}_{4T})'$ and $\bm{\zeta}=(\zeta_{1+k},\dots,\zeta_{4T})'$. Define a $(4+k)\times(4+k)$ dimensional scaling matrix $\bm{D}_{T}=diag((4T-k)^{-1},(4T-k)^{-1},(4T-k)^{-1},(4T-k)^{-1},(4T-k)^{-1/2},\dots,(4T-k)^{-1/2})$. Then 
$$\bm{D}_{T}^{-1}(\bm{\hat{\beta}}-\bm{\beta})=(\bm{D}_{T}\bm{X}'\bm{X}\bm{D}_{T})^{-1}\bm{D}_{T}\bm{X}'\bm{\zeta}^{(k)}.$$
Let $\|\cdot\|$ be the $L_{2}$ induced norm of matrices. Now we define a diagonal matrix $\bm{R}$ such that  $\|\bm{D}_{T}\bm{X}'\bm{X}\bm{D}_{T}-\bm{R}\|$ converges to 0 in probability. Specifically, let 
$$\bm{R}=diag(R_{1},R_{2},R_{3},R_{4},\tbGamma),$$
where
\begin{align*}
R_{1}&=\frac{c_{1}'\bm{\Theta}(1)\sum_{\tau=k+1}^{4T} \bm{S}_{\tau}\bm{S}_{\tau}'\bm{\Theta}(1)'c_{1}}{(4T-k)^2}\\
R_{2}&=\frac{c_{2}'\bm{\Theta}(1)\sum_{\tau=k+1}^{4T} \bm{S}_{\tau}\bm{S}_{\tau}'\bm{\Theta}(1)'c_{2}}{(4T-k)^2}\\
R_{3}&=\frac{c_{3}'\bm{\Theta}(1)\sum_{\tau=k+1}^{4T} \bm{S}_{\tau}\bm{S}_{\tau}'\bm{\Theta}(1)'c_{3}+c_{4}'\bm{\Theta}(1)\sum_{\tau=k+1}^{4T} \bm{S}_{\tau}\bm{S}_{\tau}'\bm{\Theta}(1)'c_{4}}{2(4T-k)^{2}}\\
R_{4}&=R_{3}, \quad \bm{S}_{\tau}=\sum_{i=k+1}^{\tau}\bm{\ep}_{i}, \quad \tbGamma_{i,j}=\tgamma(|i-j|).
\end{align*}
The definition of $R_{j}$, $j=1,2,3,4$, follows from the multivariate Beveridge-Nielson Decomposition; see \cite{hamilton1994time}, pp. 545-546. The definition of $\tbGamma$ is due to the fact that $(4T-k)^{-1}\sum_{\tau=1+k}^{4T}V_{\tau-i}V_{\tau-j}$ converges in probability to the seasonal average of autocovariance of $\{V_{\tau}\}$ of lag $|i-j|$.

Following the definition of $\bm{R}$, we make the following decomposition:
\begin{equation}
\begin{aligned}
\bm{D}_{T}^{-1}(\bm{\hat{\beta}}-\bm{\beta})
&=(\bm{D}_{T}\bm{X}'\bm{X}\bm{D}_{T})^{-1}\bm{D}_{T}\bm{X}'\bm{\zeta}^{(k)}\\
&=[(\bm{D}_{T}\bm{X}'\bm{X}\bm{D}_{T})^{-1}-\bm{R}^{-1}]\bm{D}_{T}\bm{X}'\bm{\zeta}^{(k)}+\bm{R}^{-1}\bm{D}_{T}\bm{X}'(\bm{\zeta}^{(k)}-\bm{\zeta)}+\bm{R}^{-1}\bm{D}_{T}\bm{X}'\bm{\zeta}.
\end{aligned}
\label{Said decomposition}
\end{equation}
Notice the last term in the right hand side summation, $\bm{R}^{-1}\bm{D}_{T}\bm{X}'\bm{\zeta}$, is free of $k$. Later we will find out its asymptotic distribution as $T\rightarrow\infty$. But now we need to prove the first two terms in the right hand side of \eqref{Said decomposition} converge to zero as $T\rightarrow\infty$ and $k\rightarrow\infty$. To do so, it suffices to show
\begin{align}
\|(\bm{D}_{T}\bm{X}'\bm{X}\bm{D}_{T})^{-1}-\bm{R}^{-1}\|&=o_{p}(k^{-1/2})\label{norm 1},\\
\|\bm{D}_{T}\bm{X}'(\bm{\zeta}^{(k)}-\bm{\zeta})\|&=o_{p}(1)\label{norm 2},\\
\|\bm{D}_{T}\bm{X}'\bm{\zeta}\|&=O_{p}(k^{1/2})\label{norm 3},\\
\|\bm{R}^{-1}\|&=O_{p}(1).\label{norm 4}
\end{align}
Equation \eqref{norm 1} can be proven straightforwardly; see \cite{said1984testing}. For \eqref{norm 2}, notice
\begin{align*}
&E\|\bm{D}_{T}\bm{X}'(\bm{\zeta}^{(k)}-\bm{\zeta})\|^{2}\\
=&E[(4T-k)^{-2}\sum_{j=1}^{4}(\sum_{\tau=k+1}^{4T}Y_{j,\tau-1}(\zeta_{\tau}^{(k)}-\zeta_{\tau}))^{2}+(4T-k)^{-1}\sum_{i=1}^{k}(\sum_{\tau=k+1}^{4T}V_{\tau-i}(\zeta_{\tau}^{(k)}-\zeta_{\tau}))^{2}].
\end{align*}
Notice that $\zeta_{\tau}^{(k)}-\zeta_{\tau}=\sum_{i=k+1}^{\infty}\tpsi_{i}V_{\tau-i}$. Under Assumption \ref{assump 1b}, $\{V_{4t+s}\}$ is a VARMA sequence with finite orders, thus $\{\tV_{\tau}\}$ has an ARMA expression with finite orders; see \cite{osborn1991implications}. 
Hence, $\tpsi(L)$ has exponentially decaying coefficient $\tpsi_{i}$. It follows straightforwardly that $E\|\bm{D}_{T}\bm{X}'(\bm{\zeta}^{(k)}-\bm{\zeta})\|^{2}\rightarrow 0$. For \eqref{norm 3}, notice that
$$E\|\bm{D}_{T}\bm{X}'\bm{\zeta}\|^{2}=E[(4T-k)^{-2}\sum_{j=1}^{4}(\sum_{\tau=k+1}^{4T}Y_{j,\tau-1}\zeta_{\tau})^{2}+(4T-k)^{-1}\sum_{i=1}^{k}(\sum_{\tau=k+1}^{4T}V_{\tau-i}\zeta_{\tau})^{2}].$$
By \eqref{zeta_1}, \eqref{zeta_2}, and the stationarity of $\{\bep_{\tau}\}$,
\begin{equation}\label{norm 3_1}
\begin{aligned}
&E[(4T-k)^{-1}(\sum_{\tau=k+1}^{4T}V_{\tau-i}\zeta_{\tau})^{2}]\\
&=\frac{1}{4}\sum_{s=-3}^{0}\sum_{h=-\infty}^{\infty}Cov(V_{4t+s-i}\zeta_{4t+s},V_{4t+s-h-i}\zeta_{4t+s-h})+o(1)\\
&=\frac{1}{4}\sum_{s=-3}^{0}\sum_{h=-\infty}^{\infty}Cov(V_{s-i}\zeta_{s},V_{s-h-i}\zeta_{s-h})+o(1).
\end{aligned}
\end{equation}
Without loss of generality we can let $i=1$ and $s=0$ in \eqref{norm 3_1}. By Assumption \ref{assump 2a}, \ref{assump 2b}, and \eqref{zeta_0}, we can write $V_{\tau}$ and $\zeta_{\tau}$ as linear combinations of $\ep_{\tau}$. By doing so, we can straightforwardly get 
\begin{equation}\label{norm 3_2}
\sum_{h=-\infty}^{\infty}Cov(V_{-1}\zeta_{0},V_{-h-1}\zeta_{-h})\leq \text{const. } sup_{i_{1},j_{1},i_{2},j_{2}}\sum_{h=-\infty}^{\infty} |Cov(\ep_{i_{1}-1}\ep_{j_{1}},\ep_{i_{2}-h-1}\ep_{j_{2}-h})|.
\end{equation}
The right hand side of this inequality is assumed to be bounded under Assumption \ref{assump 2a}. On the other hand, the right hand side is also bounded under Assumption \ref{assump 2b}, by the lemma below.
\begin{Lemma}
\label{boundedness}
Under Assumption \ref{assump 2b}, there exists $K>0$ such that for all $i_1$, $i_2$, $j_1$, and $j_2$, $$\sum_{h=-\infty}^{\infty} |Cov(\ep_{i_{1}}\ep_{j_{1}},\ep_{i_{2}-h}\ep_{j_{2}-h})|<K.$$
\end{Lemma}
\begin{proof}
Without loss of generality, assume that $\{\ep_{t}\}_{t=1}^{n}$ is a strictly stationary strong mixing time series, and $\{\ep_{t}\}$'s strong mixing coefficient $\alpha(h)$ satisfies $\sum_{h=1}^{\infty}\alpha^{\delta/(4+\delta)}(h)<\infty$. 
By Lemma A.0.1 of \cite{politis1999springer},
\begin{align*}
&\mathrel{\phantom{\leq}}|Cov(\ep_{i_{1}}\ep_{j_{1}},\ep_{i_{2}-h}\ep_{j_{2}-h})|\\
&\leq\text{const.}\ \alpha(min(|i_{1}-(i_{2}-h)|,|i_{1}-(j_{2}-h)|,|j_{1}-(i_{2}-h)|,|j_{1}-(j_{2}-h)|)))^{1-\frac{1}{(4+\delta)/2}-\frac{1}{(4+\delta)/2}}.
\end{align*}
Hence,
\begin{align*}
&\sum_{h=-\infty}^{\infty} |Cov(\ep_{i_{1}}\ep_{j_{1}},\ep_{i_{2}-h}\ep_{j_{2}-h})|\\
&\leq 
\text{const. } \sum_{h=-\infty}^{\infty}(\alpha(min(|i_{1}-(i_{2}-h)|,|i_{1}-(j_{2}-h)|,|j_{1}-(i_{2}-h)|,|j_{1}-(j_{2}-h)|)))^{\frac{\delta}{4+\delta}}\\
&\leq \text{const. } \sum_{h=-\infty}^{\infty}(\alpha(|i_{1}-(i_{2}-h)|)^{\frac{\delta}{4+\delta}}+\alpha(|i_{1}-(j_{2}-h)|)^{\frac{\delta}{4+\delta}}+\alpha(|j_{1}-(i_{2}-h)|)^{\frac{\delta}{4+\delta}}+\alpha(|j_{1}-(j_{2}-h)|)^{\frac{\delta}{4+\delta}})\\
&\leq \text{const. }
\sum_{h=-\infty}^{\infty}\alpha(|h|)^{\frac{\delta}{4+\delta}}<\infty. &\mbox{\qedhere}
\end{align*}
\end{proof}
By \eqref{norm 3_1}, \eqref{norm 3_2}, Assumption \ref{assump 2a}, \ref{assump 2b}, and Lemma \ref{boundedness}, we have $E[(4T-k)^{-1}(\sum_{\tau=k+1}^{4T}V_{\tau-i}\zeta_{\tau})^{2}]=O(1)$. Similarly,  $E[((4T-k)^{-1}\sum_{\tau=k+1}^{4T}Y_{j,\tau-1}\zeta_{\tau})^{2}]=O(1)$. Hence, \eqref{norm 3} follows. 
Now justify \eqref{norm 4}. By Assumption \ref{assump 1b}, the determinant of $\bPsi(z)$ has all its roots outside unit circle. Hence, $\{\tV_{\tau}\}$ is invertible; hence, $\|\tbGamma^{-1}\|=O(1)$. In addition,
\begin{equation}\label{norm 4_1}
\frac{\bc_{j}'\bm{\Theta}(1)\bm{\Omega}^{1/2}\sum_{\tau=k+1}^{4T} \bm{S}_{\tau}\bm{S}_{\tau}'\bm{\Omega}^{1/2}\bm{\Theta}(1)'\bc_{j}}{(4T-k)^2}\Rightarrow \bc_{j}'\bm{\Theta}(1)\bm{\Omega}^{1/2}\int \bW\bW'\bm{\Omega}^{1/2}\bm{\Theta}(1)'\bc_{j}.
\end{equation}
Since
$$P(\bc_{j}'\bm{\Theta}(1)\bm{\Omega}^{1/2}\int \bW\bW'\bm{\Omega}^{1/2}\bm{\Theta}(1)'\bc_{j}=0)=0,$$
we have that for all $\ep>0$, there exists $M_{\ep}>0,$ such that $P(\bc_{j}'\bm{\Theta}(1)\bm{\Omega}^{1/2}\int \bW\bW'\bm{\Omega}^{1/2}\bm{\Theta}(1)'\bc_{j}<M_{\ep})<\ep$.
\eqref{norm 4} follows from the definition of $O_{p}(1)$.

Combining \eqref{Said decomposition}, \eqref{norm 1}, \eqref{norm 2}, \eqref{norm 3}, and \eqref{norm 4}, we have
\begin{equation*}
\bm{D}_{T}^{-1}(\bm{\hat{\beta}}-\bm{\beta})=\bm{R}^{-1}\bm{D}_{T}\bm{X}'\bm{\zeta}+o_{p}(1).    
\end{equation*}
Now we find the asymptotic distribution of $\bm{R}^{-1}\bm{D}_{T}\bm{X}'\bm{\zeta}$. In a straightforward way, the limiting distribution of $\bm{R}^{-1}$ can be derived by \eqref{norm 4_1}. Since 
$$\bm{D}_{T}\bm{X}'\bm{\zeta}=(4T-k)^{-1}\sum_{\tau=k+1}^{4T}Y_{j,\tau-1}\zeta_{\tau}+(4T-k)^{-1}\sum_{\tau=k+1}^{4T}V_{\tau-i}\zeta_{\tau},$$
we can derive the limiting distribution of $(4T-k)^{-1}\sum_{\tau=k+1}^{4T}Y_{j,\tau-1}\zeta_{\tau}$ by the lemma below, and the limiting distribution of $(4T-k)^{-1}\sum_{\tau=k+1}^{4T}V_{\tau-i}\zeta_{\tau}$ in a similar way. 
\begin{Lemma}\label{le:unaug real3}
\begin{align*}
&\frac{1}{4T}\sum_{\tau=1}^{4T}Y_{1,\tau-1}\zeta_{\tau}\Rightarrow  Var(\tzeta_{\tau})\tilde{\theta}(1)\int_{0}^{1} W_{1}(r)dW_{1}(r), \\
&\frac{1}{4T}\sum_{\tau=1}^{4T}Y_{2,\tau-1}\zeta_{\tau}\Rightarrow  Var(\tzeta_{\tau})\tilde{\theta}(-1)\int_{0}^{1} W_{2}(r)dW_{2}(r),
\end{align*}
\begin{align*}
&(\frac{1}{4T}\sum_{\tau=1}^{4T}Y_{3,\tau-1}\zeta_{\tau})^{2}+(\frac{1}{4T}\sum_{\tau=1}^{4T}Y_{4,\tau-1}\zeta_{\tau})^{2}\\
&\Rightarrow\frac{ Var(\tzeta_{\tau})[\frac{1}{4}\bc_{4}'\bTheta(1)\bOmega\bTheta(1)'\bc_{4}\int W_{4}(r)dW_{4}(r)+\frac{1}{4}\bc_{3}'\bTheta(1)\bOmega\bTheta(1)'\bc_{3}\int W_{3}(r)dW_{3}(r)]^{2}}{\frac{1}{4}(\bc_{4}'\bTheta(1)\bOmega\bTheta(1)'\bc_{4}+\bc_{3}'\bTheta(1)\bOmega\bTheta(1)'\bc_{3})}\\
&\mathrel{\phantom{\Rightarrow}}+\frac{ Var(\tzeta_{\tau})[\sqrt{\frac{1}{4}\bc_{4}'\bTheta(1)\bOmega\bTheta(1)'\bc_{4}\frac{1}{4}\bc_{3}'\bTheta(1)\bOmega\bTheta(1)'\bc_{3}}(\int_{0}^{1} W_{3}(r)d W_{4}(r)-\int W_{4}(r)dW_{3}(r))]^{2}}{\frac{1}{4}(\bc_{4}'\bTheta(1)\bOmega\bTheta(1)'\bc_{4}+\bc_{3}'\bTheta(1)\bOmega\bTheta(1)'\bc_{3})}.
\end{align*}
\label{asy zeta}
\end{Lemma}
\begin{proof}[Proof of Lemma \ref{asy zeta}]
Firstly we focus on the convergence of $\frac{1}{4T}\sum_{\tau=1}^{4T}Y_{1,\tau-1}\zeta_{\tau}$. The convergence of $\frac{1}{4T}\sum_{\tau=1}^{4T}Y_{2,\tau-1}\zeta_{\tau}$ can be proven analogously. Define $\tbPsi(z)$ such that $\bzeta_{t}=\tbPsi(\bB)\bV_{t}$. Let $\xi_{\tau}=\tpsi(L)Y_{\tau}$, $\xi_{1,\tau}=\tpsi(L)Y_{1,\tau}$, $\bxi_{t}=(\xi_{4t-3},\xi_{4t-2},\xi_{4t-1},\xi_{4t})'$, $\bzeta_{t}=(\zeta_{4t-3},\zeta_{4t-2},\zeta_{4t-1},\zeta_{4t})'$. Then $\bB\bxi_{t}=\bzeta_{t}$, and 
\begin{align*}
&\frac{1}{4T}\sum_{\tau=1}^{4T}Y_{1,\tau-1}\zeta_{\tau}\\
&=\ttheta(1)\frac{1}{4T}\sum_{t=1}^{T}\sum_{s=-3}^{0}\xi_{1,4t+s-1}\zeta_{4t+s} \ \text{ (by Beveridge-Nielson Decomposition, up to }\op)\\
&=\ttheta(1)\frac{1}{4T}\sum_{t=1}^{T}[\bc_{1}'\bxi_{t-1}\bzeta_{t}'\bc_{1}+\sum_{s=-3}^{0}\sum_{k=-3}^{s-1}\zeta_{4t+k}\zeta_{4t+s}]\\
&\Rightarrow\frac{1}{4}\ttheta(1)\bc_{1}'\tbPsi(1)\bTheta(1)\bOmega^{1/2}\int \bW d \bW' \bOmega^{1/2}\bTheta(1)'\tbPsi(1)'\bc_{1} \ \text{(by \eqref{zeta_1}, \eqref{zeta_2}, and FCLT)}\\
&\mathrel{\phantom{\Rightarrow}}+\frac{1}{4}\ttheta(1)[\sum_{s=-3}^{0}\sum_{k=-3}^{s-1}E\zeta_{4t+k}\zeta_{4t+s}+\bc_{1}'\sum_{i=1}^{\infty}E\bzeta_{t-i}\bzeta_{t}'\bc_{1}]\\
&=\frac{1}{4}\ttheta(1)\bc_{1}'\tbPsi(1)\bTheta(1)\bOmega^{1/2}\int \bW d \bW' \bOmega^{1/2}\bTheta(1)'\tbPsi(1)'\bc_{1} \ \text{(by \eqref{zeta_1}, } \{\tzeta_{\tau}\} \text{ is white noise)}\\
&=\frac{1}{4}\ttheta(1)(\tpsi(1))^{2}\bc_{1}'\bTheta(1)\bOmega^{1/2}\int \bW d \bW' \bOmega^{1/2}\bTheta(1)'\bc_{1} \ (\text{since } \bc_{1}'\tbPsi(1)=\tpsi(1)\bc_{1}')\\
&= Var(\tzeta_{\tau})\tilde{\theta}(1)\int_{0}^{1} W_{1}(r)dW_{1}(r)\\
&\text{(by \cite{osborn1991implications}, p. 378, } \frac{1}{4}\bc_{1}'\bTheta(1)\bOmega\bTheta(1)'\bc_{1}= Var(\tzeta_{\tau})\tilde{\theta}(1)^{2}).
\end{align*}
Secondly we show the convergence of $(\frac{1}{4T}\sum_{\tau=1}^{4T}Y_{3,\tau-1}\zeta_{\tau})^{2}+(\frac{1}{4T}\sum_{\tau=1}^{4T}Y_{4,\tau-1}\zeta_{\tau})^{2}$. Let $\xi_{3,\tau}=\tpsi(L)Y_{3,\tau}$, $\tpsi_{a}=(\tpsi(i)+\tpsi(-i))/2$, $\tpsi_{b}=(\tpsi(i)-\tpsi(-i))/2i$, 
$\ttheta_{a}=(\ttheta(i)+\ttheta(-i))/2$, and $\ttheta_{b}=(\ttheta(i)-\ttheta(-i))/2i$. Then 
\begin{align*}
\allowdisplaybreaks
&\frac{1}{4T}\sum_{\tau=1}^{4T}Y_{3,\tau-1}\zeta_{\tau}\\
&=\frac{1}{4T}\sum_{\tau=1}^{4T}(\ttheta_{a}\xi_{3,\tau-1}-\ttheta_{b}\xi_{4,\tau-1})\zeta_{\tau}\\
&(\text{by Beveridge-Nielson Decomposition, up to }\op) \\
&=\frac{1}{4T}\sum_{t=1}^{T}\ttheta_{a}[\bc_{3}'\bxi_{t-1}\bzeta_{t}'\bc_{3}+\bc_{4}'\bxi_{t-1}\bzeta_{t}'\bc_{4}-\sum_{s=-3}^{-2}\zeta_{4t+s}\zeta_{4t+s+2}]\\
&\mathrel{\phantom{=}}-\frac{1}{4T}\sum_{t=1}^{T}\ttheta_{b}[\bc_{3}'\bxi_{t-1}\bzeta_{t}'\bc_{4}-\bc_{4}'\bxi_{t-1}\bzeta_{t}'\bc_{3}-\sum_{s=-3}^{-1}\zeta_{4t+s}\zeta_{4t+s+1}+\zeta_{4t-3}\zeta_{4t}]\\
&\Rightarrow \frac{1}{4}\ttheta_{a}[\bc_{3}'\tbPsi(1)\bTheta(1)\bOmega^{1/2}\int \bW d \bW' \bOmega^{1/2}\bTheta(1)'\tbPsi(1)'\bc_{3}\\
&\mathrel{\phantom{\Rightarrow \frac{1}{4}\ttheta_{a}[}}+\bc_{4}'\tbPsi(1)\bTheta(1)\bOmega^{1/2}\int \bW d \bW' \bOmega^{1/2}\bTheta(1)'\tbPsi(1)'\bc_{4}]\\
&\mathrel{\phantom{=}}-\frac{1}{4}\ttheta_{b}[\bc_{3}'\tbPsi(1)\bTheta(1)\bOmega^{1/2}\int \bW d \bW' \bOmega^{1/2}\bTheta(1)'\tbPsi(1)'\bc_{4}\\
&\mathrel{\phantom{\mathrel{\phantom{=}}-\frac{1}{4}\ttheta_{b}[}}-\bc_{4}'\tbPsi(1)\bTheta(1)\bOmega^{1/2}\int \bW d \bW' \bOmega^{1/2}\bTheta(1)'\tbPsi(1)'\bc_{3}]\\
&(\text{by \eqref{zeta_1}, \eqref{zeta_2}, and FCLT, the covariances of }\{\zeta_{4t+s}\} \text{ cancel out since } \{\tzeta_{\tau}\} \text{ is white noise})\\
&= \frac{1}{4}\ttheta_{a}|\tpsi(i)|^{2}[\bc_{4}'\bTheta(1)\bOmega^{1/2}\int \bW d \bW' \bOmega^{1/2}\bTheta(1)'\bc_{4}+\bc_{3}'\bTheta(1)\bOmega^{1/2}\int \bW d \bW' \bOmega^{1/2}\bTheta(1)'\bc_{3}]\\
&\mathrel{\phantom{=}}-\frac{1}{4}\ttheta_{b}|\tpsi(i)|^{2}[\bc_{3}'\bTheta(1)\bOmega^{1/2}\int \bW d \bW' \bOmega^{1/2}\bTheta(1)'\bc_{4}-\bc_{4}'\bTheta(1)\bOmega^{1/2}\int \bW d \bW' \bOmega^{1/2}\bTheta(1)'\bc_{3}]\\
&(\text{since } \bc_{3}'\tbPsi(1)=\tpsi_{b}\bc_{4}'+\tpsi_{a}\bc_{3}', \ \bc_{4}'\tbPsi(1)=\tpsi_{a}\bc_{4}'-\tpsi_{b}\bc_{3}', \ \text{and } \tpsi_{a}^{2}+\tpsi_{b}^{2}=|\tpsi(i)|^{2}).
\end{align*}
Similarly,
\begin{align*}
&\frac{1}{4T}\sum_{\tau=1}^{4T}Y_{4,\tau-1}\zeta_{\tau}\\
&=\frac{1}{4}\ttheta_{b}|\tpsi(i)|^{2}[\bc_{4}'\bTheta(1)\bOmega^{1/2}\int \bW d \bW' \bOmega^{1/2}\bTheta(1)'\bc_{4}+\bc_{3}'\bTheta(1)\bOmega^{1/2}\int \bW d \bW' \bOmega^{1/2}\bTheta(1)'\bc_{3}]\\
&\mathrel{\phantom{=}}+\frac{1}{4}\ttheta_{a}|\tpsi(i)|^{2}[\bc_{3}'\bTheta(1)\bOmega^{1/2}\int \bW d \bW' \bOmega^{1/2}\bTheta(1)'\bc_{4}-\bc_{4}'\bTheta(1)\bOmega^{1/2}\int \bW d \bW' \bOmega^{1/2}\bTheta(1)'\bc_{3}]
\end{align*}
The lemma follows from  $|\tpsi(i)|^{2}=|\ttheta(i)|^{-2}$ and the fact that, by \cite{osborn1991implications}, $$ Var(\tzeta_{\tau}) |\ttheta(i)|^{2}=\frac{1}{4}(\bc_{4}'\bTheta(1)\bOmega\bTheta(1)'\bc_{4}+\bc_{3}'\bTheta(1)\bOmega\bTheta(1)'\bc_{3}). \eqno\mbox{\qedhere} $$ \end{proof}
As mentioned, the asymptotic distribution of $\bm{\hat{\beta}}$ can be derived straightforwardly from Lemma \ref{le:unaug real3}. Now we come to the asymptotic distribution of the t-statistics and the F-statistics. Notice,
\begin{align*}
t_{j}^{A}&=\hat{\sigma}^{-1}[[(\bm{X}'\bm{X})^{-1}]_{jj}]^{-1/2}[(\bm{X}'\bm{X})^{-1}\bm{X}'\bm{\zeta}^{(k)}]_{j}\\
&=\hat{\sigma}^{-1}[[[(4T-k)^{-2}(\bm{X}'\bm{X})^{-1}]_{jj}]^{-1/2}-[[\bm{R}^{-1}]_{jj}]^{-1/2}](4T-k)[(\bm{X}'\bm{X})^{-1}\bm{X}'\bm{\zeta}^{(k)}]_{j}\\
&\mathrel{\phantom{=}}+\hat{\sigma}^{-1}[[\bm{R}^{-1}]_{jj}]^{-1/2}((4T-k)(\bm{X}'\bm{X})^{-1}\bm{X}'\bm{\zeta}^{(k)}-\bm{R}^{-1}(4T-k)^{-1}\bm{X}'\bzeta)_{j}\\
&\mathrel{\phantom{=}}+\hat{\sigma}^{-1}[[\bm{R}^{-1}]_{jj}]^{-1/2}(\bm{R}^{-1}(4T-k)^{-1}\bm{X}'\bzeta)_{j}\\
&=\hat{\sigma}^{-1}[[\bm{R}^{-1}]_{jj}]^{-1/2}(\bm{R}^{-1}(4T-k)^{-1}\bm{X}'\bzeta)_{j}+o_{p}(1).
\end{align*}
By the consistency of $\bm{\hat{\beta}}$, we have $\hat{\sigma}^{2}\stackrel{p}\rightarrow Var(\tzeta_{\tau})$. The asymptotic distributions of the t-statistics follows straightforwardly from Lemma \ref{asy zeta}. Further, the asymptotic distributions of the F-statistics are identical with the asymptotic distributions of the averages of the squares of the corresponding t-statistics because of the asymptotic orthogonality of the regression. Hence, the proof of Theorem \ref{aug real} is complete.
\section{Proof of Proposition \ref{iid FCLT}.}
Define $\{\phi_{i,s}\}$ and $\{\ep_{4t+s}^{(k)}\}$ as the prediction coefficients and errors of the predictive equation
$$(1-L^{4})Y_{4t+s}=\sum_{j=1}^{4}\pi_{j,s}Y_{j,4t+s-1}+\sum_{i=1}^{k}\phi_{i,s}(1-L^{4})Y_{4t+s-i}+\ep_{4t+s}^{(k)}.$$ 
More specifically, $\{\phi_{i,s}\}$ are defined such that for each $s=1,2,3,4$,  $\ep_{4t+s}^{(k)}\stackrel{def}{=}(1-L^{4})Y_{4t+s}-\sum_{j=1}^{4}\pi_{j,s}Y_{j,4t+s-1}-\sum_{i=1}^{k}\phi_{i,s}(1-L^{4})Y_{4t+s-i}$ is orthogonal to $Y_{j,4t+s-1}$, $j=1,2,3,4$ and $(1-L^{4})Y_{4t+s-i}$, $i=1,\dots,k.$ Define $\{i_{\tau}\}$ and $\{I_{t}\}$ such that $\ep_{\tau}^{\star}=\check{\ep}_{i_{\tau}}$ and  $\ep^{\star}_{4t+s}=\check{\ep}_{4I_{t}+s}$. By Algorithm \ref{seasonal iid bootstrap}, $\{i_{\tau}\}$ is a sequence of independent but not identical random variables, while $\{I_{t}\}$ is a sequence of iid random variables. Let
\begin{align*}
\up_{T,\tau}^{(1)}&= 
(\ep_{i_{\tau}}^{(k)}-E^{\circ}\ep_{i_{\tau}}^{(k)})/Std^{\circ}(\ep_{i_{\tau}}^{(k)})\\
\up_{T,\tau}^{(2)}&= 
(-1)^{\tau}(\ep_{i_{\tau}}^{(k)}-E^{\circ}\ep_{i_{\tau}}^{(k)})/Std^{\circ}((-1)^{\tau}\ep_{i_{\tau}}^{(k)})\\
\up_{T,\tau}^{(3)}&= 
\sqrt{2}\sin(\frac{\pi \tau}{2})(\ep_{i_{\tau}}^{(k)}-E^{\circ}\ep_{i_{\tau}}^{(k)})/Std^{\circ}(\sqrt{2}\sin(\frac{\pi \tau}{2})\ep_{i_{\tau}}^{(k)})\\
\up_{T,\tau}^{(4)}&= 
\sqrt{2}\cos(\frac{\pi \tau}{2})(\ep_{i_{\tau}}^{(k)}-E^{\circ}\ep_{i_{\tau}}^{(k)})/Std^{\circ}(\sqrt{2}\cos(\frac{\pi \tau}{2})\ep_{i_{\tau}}^{(k)})
\end{align*}
Let $R_{T}^{\star}$ be the partial sum of $\up_{T,\tau}$ above. Formally,
$$R_{T}^{\star}(u_{1},u_{2},u_{3},u_{4})=(
\frac {1}{\sqrt{4T}}\sum_{\tau=1}^{\lfloor 4Tu_{1} \rfloor}\up_{T,\tau}^{(1)},
\frac {1}{\sqrt{4T}}\sum_{\tau=1}^{\lfloor 4Tu_{2}\rfloor}\up_{T,\tau}^{(2)},
\frac {1}{\sqrt{4T}}\sum_{\tau=1}^{\lfloor 4Tu_{3} \rfloor}\up_{T,\tau}^{(3)},
\frac {1}{\sqrt{4T}}\sum_{\tau=1}^{\lfloor 4Tu_{4} \rfloor}\up_{T,\tau}^{(4)})'.$$
Let $\|\cdot\|$ denote the $L_{2}$ norm. To justify Proposition \ref{iid FCLT}, it suffices to show  
\begin{equation}\label{iid FCLT1}
\|S_{T}^{\star}-R_{T}^{\star}\|\stackrel{p}\rightarrow 0 \text{ uniformly in $u_{1}$, $u_{2}$, $u_{3}$ and $u_{4}$, }    
\end{equation}
\begin{equation}\label{iid FCLT2}
\text{ and } R_{T}^{\star}\Rightarrow \bW^{\star} \text{ in probability, }
\end{equation}
because the unconditional convergence in \eqref{iid FCLT1} implies that in probability the conditional distribution of $\|S_{T}^{\star}-R_{T}^{\star}\|$ given $\{Y_{4t+s}\}$ converges to zero. To prove \eqref{iid FCLT1}, we can without loss of generality focus on the uniform convergence of the first coordinate, that is, uniformly in $u_{1}$, \
\begin{equation*}
|\frac{1}{\sqrt{4T}}\sum_{\tau=1}^{\lfloor 4Tu_{1} \rfloor}\ep_{\tau}^{\star}/\sigma_{1}^{\star}-\frac {1}{\sqrt{4T}}\sum_{\tau=1}^{\lfloor 4Tu_{1} \rfloor}\up_{T,\tau}^{(1)}|\stackrel{p}\rightarrow 0.
\end{equation*} 
Notice uniformly in $u_{1}$, 
\begin{equation}
\begin{aligned}
&\frac{1}{\sqrt{4T}}\sum_{\tau=1}^{\lfloor 4Tu_{1} \rfloor}\ep_{\tau}^{\star}\\
=&\frac{1}{\sqrt{4T}}\sum_{s=-3}^{0}\sum_{t=1}^{\lfloor Tu_{1}\rfloor}\ep_{4t+s}^{\star}+o_{p}(1)\\
=&\frac{1}{\sqrt{4T}}\sum_{s=-3}^{0}\sum_{t=1}^{\lfloor Tu_{1}\rfloor}(\hat{\ep}_{4I_{t}+s}-\frac{1}{T}\sum_{t=1}^{T}\hat{\ep}_{4t+s})+o_{p}(1)\\
=&\frac{1}{\sqrt{4T}}\sum_{s=-3}^{0}\sum_{t=1}^{\lfloor Tu_{1}\rfloor}(\ep_{4I_{t}+s}^{(k)}-\frac{1}{T}\sum_{t=1}^{T}\ep_{4t+s}^{(k)})\\
&-\frac{1}{\sqrt{4T}}\sum_{s=-3}^{0}\sum_{t=1}^{\lfloor Tu_{1}\rfloor}\sum_{j=1}^{4}(\hat{\pi}_{j,s}^{A}-\pi_{j,s})(Y_{j,4I_{t}+s-1}-\frac{1}{T}\sum_{t=1}^{T}Y_{j,4t+s-1})\\
&-\frac{1}{\sqrt{4T}}\sum_{s=-3}^{0}\sum_{t=1}^{\lfloor Tu_{1}\rfloor}\sum_{i=1}^{k}(\hat{\phi}_{i,s}-\phi_{i,s})((1-L^{4})Y_{4I_{t}+s-i}-\frac{1}{T}\sum_{t=1}^{T}(1-L^{4})Y_{4t+s-i})\\
&+o_{p}(1)\\
=&\frac{1}{\sqrt{4T}}\sum_{s=-3}^{0}\sum_{t=1}^{\lfloor Tu_{1}\rfloor}(\ep_{4I_{t}+s}^{(k)}-\frac{1}{T}\sum_{t=1}^{T}\ep_{4t+s}^{(k)})-B_{T}(u_{1})-C_{T}(u_{1})+o_{p}(1)\\
=&\frac{1}{\sqrt{4T}}\sum_{\tau=1}^{\lfloor 4Tu_{1} \rfloor}(\ep_{i_{\tau}}^{(k)}-E^{\circ}\ep_{i_{\tau}}^{(k)})-B_{T}(u_{1})-C_{T}(u_{1})+o_{p}(1),
\end{aligned}
\label{iid FCLT3}
\end{equation}
where $B_{T}(u_{1})$ and $C_{T}(u_{1})$ have obvious definitions. 

Now we show 
$B_{T}(u_{1})\stackrel{p}\rightarrow 0$, and $C_{T}(u_{1})\stackrel{p}\rightarrow 0$, uniformly in $u_{1}$. For $B_{T}(u_{1})$, notice if $\pi_{j,s}\neq 0$, then  $\{Y_{j,4t+s}\}$ is weakly stationary. Hence, by \cite{berk1974}, $\hat{\pi}_{j,s}^{A}-\pi_{j,s}=O_{p}(T^{-1/2})$. It follows straightforwardly that $B_{T}(u_{1})\stackrel{p}\rightarrow 0$ uniformly in $u_{1}$. On the other hand, if $\pi_{j,s}=0$, then by Theorem \ref{aug real}, $\hat{\pi}_{j,s}^{A}-\pi_{j,s}=O_{p}(T^{-1})$. Let  $$Q_{T}(u_{1})=\frac{1}{\sqrt{4T}}\sum_{t=1}^{\lfloor Tu_{1}\rfloor}( Y_{j,4I_{t}+s-1}-\frac{1}{T}\sum_{t=1}^{T}Y_{j,4t+s-1}).$$ It suffices to show that $\sup_{0\leq u_{1} \leq 1} Q_{T}(u_{1})=o_{p}(T)$. By continuous mapping theorem, it suffices to prove $(4T)^{-1}Q_{T}(\cdot) \Rightarrow 0(\cdot)$, where $0(\cdot)\equiv 0$. It is straightforward to show the weak convergence of the finite dimensional distributions of $(4T)^{-1}Q_{T}(\cdot)$. Furthermore, for all $r_{1}\leq r\leq r_{2}$,
\begin{equation*}
\begin{aligned}
&E[(\frac{Q_{T}(r_{2})}{T}-\frac{Q_{T}(r)}{T})^{2}(\frac{Q_{T}(r)}{T}-\frac{Q_{T}(r_{1})}{T})^{2}]\\
=&E[Var^{\circ}[\frac{Q_{T}(r_{2})}{T}-\frac{Q_{T}(r)}{T}]Var^{\circ}[\frac{Q_{T}(r)}{T}-\frac{Q_{T}(r_{1})}{T}]]\rightarrow 0.
\end{aligned}
\end{equation*}
By \cite{billingsley1999}, pp. 146-147, $(4T)^{-1}Q_{T}(\cdot)$ is tight. Hence $(4T)^{-1}Q_{T}(\cdot) \Rightarrow 0(\cdot)$, and consequently $B_{T}(u_{1})\stackrel{p}\rightarrow 0$ uniformly in $u_{1}$. For $C_{T}(u_{1})$, in light of the derivation of Theorem \ref{aug real}, it can be shown that $\hat{\phi}_{i,s}-\phi_{i,s}=O_{p}(T^{-1/2})$ holds not only under alternative hypotheses but also under the null. Hence, it follows that uniformly in $u_{1}$, $C_{T}(u_{1})\stackrel{p}\rightarrow 0$. Therefore, recalling \eqref{iid FCLT3}, we have
$$\sup_{0\leq u_{1}\leq 1}|\frac{1}{\sqrt{4T}}\sum_{\tau=1}^{\lfloor 4Tu_{1} \rfloor}\ep_{\tau}^{\star}-\frac{1}{\sqrt{4T}}\sum_{\tau=1}^{\lfloor 4Tu_{1} \rfloor}(\ep_{i_{\tau}}-E^{\circ}\ep_{i_{\tau}})|\stackrel{p}\rightarrow 0.$$
Further, it is straightforward to show $E[B_{T}^{2}(1)]\stackrel{p}\rightarrow 0$, and  $E[C_{T}^{2}(1)]\stackrel{p}\rightarrow 0$. Using the same decomposition as in \eqref{iid FCLT3}, we have $\sigma_{1}^{\star}-Std^{\circ}(\ep_{i_{\tau}}^{(k)})\stackrel{p}\rightarrow 0$. Hence we have proven \eqref{iid FCLT1}.

Secondly we prove \eqref{iid FCLT2}. Notice that the standard deviations in the definition of $\{\up_{T,\tau}^{(j)}\}$ are bounded in probability. For example, 
$$Std^{\circ}(\ep_{i_{\tau}}^{(k)})=Std^{\circ}(\ep_{4I_{t}+s}^{(k)})=Std(\ep_{4t+s}^{(k)})+o_{p}(1)=Std(\ep_{4t+s})+o_{p}(1),$$
Further, conditional on $\{Y_{4t+s}\}$, for fixed $j=1,\dots,4$, $\up_{T,1}^{(j)}, \up_{T,2}^{(j)},\dots,\up_{T,T}^{(j)}$ are row-wise independent random variables. Finally, for all $u\geq 0$,
$$Var^{\circ}[\frac{1}{\sqrt{4T}}\sum_{m=1}^{\lfloor 4Tu \rfloor}\up_{T,m}^{(j)}]\stackrel{p}\rightarrow u,$$
$$Cov^{\circ}(\frac{1}{\sqrt{4T}}\sum_{m=1}^{\lfloor 4Tu \rfloor}\up_{T,m}^{(j)},\frac{1}{\sqrt{4T}}\sum_{m=1}^{\lfloor 4Tu \rfloor}\up_{T,m}^{(i)})\stackrel{p}\rightarrow 0 \quad \text{for} \ i\neq j.$$
Hence, the conditions of Theorem 3.3 of \cite{helland1982central} are satisfied. From \cite{kreiss2016}, the convergence of $R_{T}^{\star}$ to $\bW^{\star}$ follows.
\section{Proof of Proposition \ref{SBB FCLT}.}
\begin{proof}
Without loss of generality, assume block size $b$ is a multiple of four. Let $i_{m}=I_{(m-1)b+1}$, where $I_{t}$ is defined in Algorithm \ref{seasonal block boot}. Then the $m$th block of $\{V_{t}^{*}\}$ starts from $\check{V}_{i_{m}}$. Recall $l=\lfloor 4T/b \rfloor$ denotes the number of blocks. Let $\up_{l,m}^{(j)}$ be the rescaled aggregation of the $m$th block, defined by
\begin{align*}
\up_{l,m}^{(1)}&= \frac {1}{\sqrt{b}} \sum _{h=1}^{b}(V_{i_{m}+h-1}-E^{\circ}V_{i_{m}+h-1})/Std^{\circ}(\frac {1}{\sqrt{b}}\sum _{h=1}^{b}V_{i_{m}+h-1})\\
\up_{l,m}^{(2)}&= \frac {1}{\sqrt{b}} \sum _{h=1}^{b}(-1)^{h}(V_{i_{m}+h-1}-E^{\circ}V_{i_{m}+h-1})/Std^{\circ}(\frac {1}{\sqrt{b}}\sum _{h=1}^{b}(-1)^{h}V_{i_{m}+h-1})\\
\up_{l,m}^{(3)}&= \frac {1}{\sqrt{b}} \sum _{h=1}^{b}\sqrt{2}\sin(\frac{\pi h}{2})(V_{i_{m}+h-1}-E^{\circ}V_{i_{m}+h-1})/Std^{\circ}(\frac {1}{\sqrt{b}}\sum _{h=1}^{b}\sqrt{2}\sin(\frac{\pi h}{2})V_{i_{m}+h-1})\\
\up_{l,m}^{(4)}&= \frac {1}{\sqrt{b}} \sum _{h=1}^{b}\sqrt{2}\cos(\frac{\pi h}{2})(V_{i_{m}+h-1}-E^{\circ}V_{i_{m}+h-1})/Std^{\circ}(\frac {1}{\sqrt{b}}\sum _{h=1}^{b}\sqrt{2}\cos(\frac{\pi h}{2})V_{i_{m}+h-1})
\end{align*}
Let $R_{T}^{*}$ be the partial sum of the block aggregations above. Formally,

$$R_{T}^{*}(u_{1},u_{2},u_{3},u_{4})=(
\frac {1}{\sqrt{l}}\sum_{m=1}^{\lfloor lu_{1} \rfloor}\up_{l,m}^{(1)},
\frac {1}{\sqrt{l}}\sum_{m=1}^{\lfloor lu_{2}\rfloor}\up_{l,m}^{(2)},
\frac {1}{\sqrt{l}}\sum_{m=1}^{\lfloor lu_{3} \rfloor}\up_{l,m}^{(3)},
\frac {1}{\sqrt{l}}\sum_{m=1}^{\lfloor lu_{4} \rfloor}\up_{l,m}^{(4)})'$$
To prove theorem \ref{SBB FCLT}, it suffices to show
\begin{equation}
\|S_{T}^{*}-R_{T}^{*}\|\stackrel{p}\rightarrow 0 \text{ uniformly in $u_{1}$, $u_{2}$, $u_{3}$ and $u_{4}$, }
\label{SBB FCLT1}
\end{equation}
\begin{equation}
\text{ and } R_{T}^{*}\Rightarrow \bW^{*} \text{ in probability,}
\label{SBB FCLT2}
\end{equation}
where $\|\cdot\|$ denotes the $L_{2}$ norm.  

The proof of \eqref{SBB FCLT1} is similar to that of \eqref{iid FCLT1}. Here we only present the proof of \eqref{SBB FCLT2}. First we assume Assumption \ref{assump 1a} and \ref{assump 2a}. In this scenario, it is sufficient to show that the following three properties hold:
\begin{align}
&\sum\limits_{m=1}^{\lfloor lu \rfloor}E^{\circ}[\up_{l,m}^{(i)^{2}}] \stackrel {p} \rightarrow u, \: \forall u\geq 0, \: and \: \forall \: i=1,\dots,4, \label{SBB FCLT3} \\
&\sum\limits_{m=1}^{\lfloor lu \rfloor}E^{\circ}[\up_{l,m}^{(i)^{2}}1(|\up_{l,m}|>\ep)] \stackrel {p} \rightarrow 0, \: \forall u\geq 0, \:  \: \forall \: i=1,\dots,4, \label{SBB FCLT4} \\
&\sum\limits_{m=1}^{\lfloor lu \rfloor}E^{\circ}[\up_{l,m}^{(i)}\up_{l,m}^{(j)}] \stackrel {p} \rightarrow 0, \: \forall u\geq 0, \  \forall \: i,j \in \{1,2,3,4\}, \: i\neq j.
\label{SBB FCLT5}
\end{align}
Given \eqref{SBB FCLT3}, \eqref{SBB FCLT4}, and \eqref{SBB FCLT5}, \cite{helland1982central} shows that if each row of $\{\up_{l,m}\}$ is a martingale difference sequence, then $$\sum_{m=1}^{\lfloor lu \rfloor}\up_{l,m}\Rightarrow \bW^{*}(u).$$ 
By Beveridge-Neilson Decomposition, e.g., Proposition 17.2, \cite{hamilton1994time},   p. 504, Helland's result can be generalized to the case when each row of $\{\up_{l,m}\}$ is a convolution of a constant sequence and a martingale difference array. Further, Helland's result can be generalized to the bootstrap world with \cite{kreiss2016}. Hence it suffices to show \eqref{SBB FCLT3}, \eqref{SBB FCLT4}, and \eqref{SBB FCLT5}. 

To verify \eqref{SBB FCLT3} and \eqref{SBB FCLT4}, notice that for all $u\geq 0$, $i=1,\dots,4$,
$$\sum_{m=1}^{\lfloor lu \rfloor}E^{\circ}[(\up_{l,m}^{(i)})^{2}]=\lfloor lu \rfloor / l \rightarrow u,$$
and, by the dominated convergence theorem,
$$\sum_{m=1}^{\lfloor lu \rfloor}E^{\circ}[(\up_{l,m}^{(i)})^{2}1(|\up_{l,m}|>\ep)] \stackrel {p} \rightarrow 0.$$
Hence, it remains to verify the \eqref{SBB FCLT5}, which indicates asymptotic independence between coordinates of $R_{T}^{*}$. Notice that the \eqref{SBB FCLT5} needs to be proved for all $i,j \in \{1,2,3,4\}, \: i\neq j$. Here we cite as an example the case $i=1$ and $j=3$. The rest of cases can be shown by similar calculations. Notice,
\begin{align*}
\sum_{m=1}^{\lfloor lt \rfloor}E^{\circ}[\up_{l,m}^{(1)}\up_{l,m}^{(3)}]
=&\frac{E^{\circ}[\frac{1}{\sqrt{b}}\sum_{h=1}^{b}V_{i_{1}+h-1}\frac{1}{\sqrt{b}}\sum_{r=1}^{b}\sqrt{2}\sin(\pi r/2)V_{i_{1}+r-1}]}
{Std^{\circ}[\frac{1}{\sqrt{b}}\sum_{h=1}^{b}V_{i_{1}+h-1}]Std^{\circ}[\frac{1}{\sqrt{b}}\sum_{r=1}^{b}\sqrt{2}\sin(\pi r/2)V_{i_{1}+r-1}]}\\
&-\frac{E^{\circ}[\frac{1}{\sqrt{b}}\sum_{h=1}^{b}V_{i_{1}+h-1}]E^{\circ}[\frac{1}{\sqrt{b}}\sum_{r=1}^{b}\sqrt{2}\sin(\pi r/2)V_{i_{1}+r-1}]}{Std^{\circ}[\frac{1}{\sqrt{b}}\sum_{h=1}^{b}V_{i_{1}+h-1}]Std^{\circ}[\frac{1}{\sqrt{b}}\sum_{r=1}^{b}\sqrt{2}\sin(\pi r/2)V_{i_{1}+r-1}]}.
\end{align*}
Since $$E^{\circ}[\frac{1}{\sqrt{b}}\sum_{h=1}^{b}V_{i_{1}+h-1}]\stackrel {p} \rightarrow 0,\: E^{\circ}[\frac{1}{\sqrt{b}}\sum_{r=1}^{b}\sqrt{2}\sin(\pi r/2)V_{i_{1}+r-1}]\stackrel {p} \rightarrow 0,$$
and both $Std^{\circ}[\frac{1}{\sqrt{b}}\sum_{h=1}^{b}V_{i_{1}+h-1}]$ and $Std^{\circ}[\frac{1}{\sqrt{b}}\sum_{r=1}^{b}\sqrt{2}\sin(\pi r/2)V_{i_{1}+r-1}]$ converge in probability to constants (\cite{dudek2014generalized}), we only need to show that $$E^{\circ}[\frac{1}{\sqrt{b}}\sum_{h=1}^{b}V_{i_{1}+h-1}\frac{1}{\sqrt{b}}\sum_{r=1}^{b}\sqrt{2}\sin(\pi r/2)V_{i_{1}+r-1}] \stackrel {p} \rightarrow 0.$$
Notice, 
\begin{align*}
&E^{\circ}[\frac{1}{\sqrt{b}}\sum_{h=1}^{b}V_{i_{1}+h-1}\frac{1}{\sqrt{b}}\sum_{r=1}^{b}\sqrt{2}\sin(\pi r/2)V_{i_{1}+r-1}]\\
&=\frac{\sqrt{2}}{b(T-b/4)}\sum_{i=1}^{T-b/4}\sum_{h=1}^{b}\sum_{r=1}^{b}\sin(\pi r/2)V_{4i+h-4}V_{4i+r-4}\\
&=-A+B+o_{p}(1), 
\end{align*}
where
\begin{align*}
A&=\frac{\sqrt{2}}{b(T-b/4)}\sum_{h=1}^{b/4}\sum_{j=1}^{T-b/4}V_{4j-3}V_{4j+4h-6},\\ B&=\frac{\sqrt{2}}{b(T-b/4)}\sum_{h=1}^{b/4}\sum_{j=1}^{T-b/4}V_{4j-3}V_{4j+4h-4}.
\end{align*}
By Assumptions \ref{assump 1a} and \ref{assump 2a}, it is straightforward to show
\begin{equation}
A\stackrel {p} \rightarrow 0, B\stackrel {p} \rightarrow 0.
\label{SBB FCLT6}
\end{equation}
Thus, we complete the proof under Assumption \ref{assump 1a} and \ref{assump 2a}. Now we assume Assumption \ref{assump 1a} and \ref{assump 2b}. Let $\bm{\up}_{l,m}=(\up_{l,m}^{(1)},\up_{l,m}^{(2)},\up_{l,m}^{(3)},\up_{l,m}^{(4)})'$. Let $\lambda_{l,j}$, $j=1,\dots,4$, be the eigenvalues of $Var\sum_{m=1}^{l}\bm{\up}_{l,m}$. By Corollary 4.2 of \cite{wooldridge1988some}, it is sufficient to show that the following two properties hold:
\begin{align}
E^{\circ}(\frac{1}{\sqrt{l}}\sum_{m=1}^{\lfloor lt \rfloor}\up_{l,m}^{(i)})(\frac{1}{\sqrt{l}}\sum_{m=1}^{\lfloor lt \rfloor}\up_{l,m}^{(j)})& \stackrel {p} \rightarrow t\mathbb{1}\{i=j\}, \ \text{for all} \ t\geq 0, \ \text{for all} \ i,j=1,\dots,4 ,\label{SBB FCLT7}\\
\lambda_{l,j}^{-1}&=O(l^{-1}), \ \text{for all} \ j=1,\dots,4.\label{SBB FCLT8}
\end{align}
Notice, to show \eqref{SBB FCLT7}, it suffices to show \eqref{SBB FCLT6}, which follows from Assumptions \ref{assump 1a} and \ref{assump 2b} and Lemma \ref{boundedness}. In addition, \eqref{SBB FCLT8} follows from the continuity of the eigenvalue function. Hence, we have completed the proof under \ref{assump 1a} and \ref{assump 2b}. 

Until now we assume that block size $b$ is a multiple of four. When $b$ is not a multiple of four, it is straightforward to show \eqref{SBB FCLT1}. For \eqref{SBB FCLT2}, let $$R_{T,s}^{*}=(\frac{1}{\sqrt{l/4}}\sum_{t=1}^{\lfloor\lfloor lu_{1}\rfloor/4\rfloor}\up_{l,4t+s}^{(1)},\frac{1}{\sqrt{l/4}}\sum_{t=1}^{\lfloor\lfloor lu_{2}\rfloor/4\rfloor}\up_{l,4t+s}^{(2)},\frac{1}{\sqrt{l/4}}\sum_{t=1}^{\lfloor\lfloor lu_{3}\rfloor/4\rfloor}\up_{l,4t+s}^{(3)},\frac{1}{\sqrt{l/4}}\sum_{t=1}^{\lfloor\lfloor lu_{4}\rfloor/4\rfloor}\up_{l,4t+s}^{(4)})'.$$
Since $\{R_{T,s}^{*}$, $s=-3,\dots,0\}$ are mutually independent with respect to $P^{\circ}$, and $R_{T,s}^{*}\Rightarrow \bW^{*}$ in probability for all $s=-3,\dots,0$, we have $R_{T}^{*}=\frac{1}{2}\sum_{s=-3}^{0}R_{T,s}^{*}+o_{p}(1)\Rightarrow \bW^{*}$ in probability.
\end{proof}
\end{document}